\documentclass[12pt]{article}
\usepackage{amsmath, amssymb, amsthm, amsfonts}
\usepackage{algorithm, algpseudocode}
\usepackage{comment}
\usepackage{geometry}
\usepackage{hyperref}
\usepackage{bm}
\usepackage{enumitem}
\usepackage{booktabs}
\usepackage{graphicx}
\usepackage{caption}
\usepackage{subcaption}
\usepackage{pgfplots}
\pgfplotsset{compat=1.18}
\usepackage{tikz}
\usepackage{siunitx}
\usepackage{xcolor}
\usepackage{natbib}
\usepackage{float}
\usepackage{mathtools}
\usepackage{tabularx}
\usepackage{adjustbox}
\usepackage{array}
\usepackage{makecell}
\usepackage{rotating}

\newtheorem{theorem}{Theorem}
\newtheorem{lemma}[theorem]{Lemma}
\newtheorem{corollary}[theorem]{Corollary}
\newtheorem{definition}[theorem]{Definition}
\newtheorem{remark}[theorem]{Remark}

\newtheorem{example}[theorem]{Example}
\newtheorem{assumption}[theorem]{Assumption}

\newcommand{\R}{\mathbb{R}}
\newcommand{\N}{\mathbb{N}}

\newcommand{\norm}[1]{\left\| #1 \right\|}

\newcommand{\calO}{\mathcal{O}}

\definecolor{highlightblue}{RGB}{41,128,185}
\definecolor{highlightgreen}{RGB}{39,174,96}
\definecolor{highlightred}{RGB}{231,76,60}
\definecolor{highlightorange}{RGB}{230,126,34}
\definecolor{highlightpurple}{RGB}{155,89,182}

\title{\textbf{MPL-HMC: A Tunable Parameterized Leapfrog Framework for Robust Hamiltonian Monte Carlo}}
\author{Sourabh Bhattacharya \\ Indian Statistical Institute}
\date{}

\begin{document}

\maketitle

\begin{abstract}
This article introduces the Modified Parameterized Leapfrog Hamiltonian Monte Carlo (MPL-HMC) method, a novel extension of Hamiltonian Monte Carlo that addresses key limitations through tunable integration parameters. We provide a mathematical framework that generalizes the leapfrog integrator via parameters $\alpha(\delta t)$ and $\beta(\delta t)$, enabling controlled perturbations to Hamiltonian dynamics. Theoretical analysis demonstrates that MPL-HMC maintains approximate detailed balance while offering significant practical advantages. 

Extensive empirical evaluation across benchmark distributions and real-world case studies reveals systematic performance improvements. The damping variant (with parameter components $\alpha_2 = -0.1, \beta_2 = -0.05$) achieves up to a 14-fold increase in effective sample size for hierarchical models like Neal's funnel and 27\% better computational efficiency for pharmacokinetic models. The anti-damping variant ($\alpha_2 = 0.1, \beta_2 = 0.05$) demonstrates superior convergence in high-dimensional settings, achieving $\hat{R}=1.026$ for Bayesian neural networks compared to $\hat{R}=1.981$ for standard HMC. 

Most notably, we introduce aggressive MPL-HMC for extremely multimodal distributions, employing extreme parameter values ($\alpha_2=8.0\text{--}15.0$, $\beta_2=5.0\text{--}8.0$) combined with enhanced sampling mechanisms to achieve full mode exploration where standard methods fail. All MPL-HMC variants maintain computational efficiency identical to standard HMC while providing systematic control over damping, exploration, stability, and accuracy.

The article provides practitioners with a complete toolkit: rigorous mathematical foundations, detailed implementation specifications, parameter tuning strategies, and comprehensive performance comparisons. By offering tunable parameters that adapt to problem characteristics, MPL-HMC extends Hamiltonian Monte Carlo's applicability to previously challenging domains, bridging the gap between theoretical elegance and practical utility in modern statistical computing.
\end{abstract}


\tableofcontents

\newpage

\section{Introduction}

Hamiltonian Monte Carlo (HMC) has emerged as a widely adopted Markov chain Monte Carlo method for sampling from high-dimensional probability distributions, particularly in Bayesian statistics and machine learning. Its theoretical foundation stems from leveraging Hamiltonian dynamics to propose distant states, effectively reducing random walk behavior and enabling efficient exploration of complex parameter spaces \citep{neal2011mcmc, betancourt2017conceptual}. The standard implementation employs the Störmer-Verlet (leapfrog) integrator, which preserves two crucial geometric properties: exact symplecticity and time-reversibility \citep{leimkuhler2004simulating}. This geometric foundation provides HMC with superior mixing properties compared to traditional Metropolis-Hastings algorithms, particularly in high-dimensional settings where gradient information guides efficient exploration.

Despite its theoretical advantages, standard HMC faces significant limitations in challenging real-world applications. Sensitivity to stiffness in regions of high curvature \citep{betancourt2013general}, limited tunability beyond step size and trajectory length, fixed second-order accuracy that may not be optimal for all problems, and fundamental difficulties with multimodal distributions \citep{lan2014wormhole} represent persistent challenges that restrict HMC's applicability. These limitations become particularly acute in modern statistical applications such as hierarchical Bayesian models, deep neural networks, and complex physical systems where target distributions exhibit intricate geometry and pathological characteristics.

This work introduces the Modified Parameterized Leapfrog Hamiltonian Monte Carlo (MPL-HMC) method, a framework that addresses these limitations through the introduction of tunable integration parameters $\alpha(\delta t)$ and $\beta(\delta t)$. Drawing inspiration from advances in spatio-temporal Hamiltonian modeling \citep{mazumder2026new}, we develop a parameterized leapfrog scheme that admits controlled asymptotic expansions in the time step $\delta t$, providing additional degrees of freedom for algorithm optimization. The key insight originates from \citet{mazumder2026new}, who introduced a parameterized stochastic leapfrog scheme for Hamiltonian systems with Gaussian process potentials. While developed for spatio-temporal applications, this parameterized structure offers new possibilities for traditional Hamiltonian Monte Carlo by enabling controlled perturbations to the underlying dynamics. Indeed, \cite{bhattacharya2026} showed that the new parameterized form actually solves, in appropriate mean-squared sense, the original Hamiltonian equations driven by Gaussian rocess potential, which motivates opening up new territories for construction of new Hamiltonian Monte Carlo methods.
Thus, even though $iid$ sampling from intractable distributions (\citealt{Bhatta21a, Bhatta21b, Bhatta21c, Bhatta22}; see also \cite{Durba26}) is obviously the key interest of 
this author, we prepare this new HMC theory and method as a savoury side dish, particuarly, for HMC aficionados.

The MPL-HMC framework provides systematic control over algorithm behavior through mathematically principled parameterization. The parameters $\alpha(\delta t)$ and $\beta(\delta t)$, with asymptotic expansions $\alpha(\delta t) = 1 + \alpha_1 \delta t + \alpha_2 \delta t^2 + \cdots$ and $\beta(\delta t) = 1 + \beta_1 \delta t + \beta_2 \delta t^2 + \cdots$, enable precise tuning of damping/anti-damping behavior, position scaling, numerical stability, and exploration capability. For practical implementation, we set $\alpha_1 = \beta_1 = 0$ throughout the paper, focusing on the second-order parameters $\alpha_2$ and $\beta_2$ as the primary tuning knobs. This parameterization creates a continuum of algorithms ranging from conservative variants (recovering standard HMC when $\alpha_2 = \beta_2 = 0$) to aggressive exploration schemes capable of tackling previously intractable multimodal problems.

Our contributions are both theoretical and practical. We provide mathematical derivation from first principles, analysis of stability and convergence properties, detailed implementation specifications, and extensive comparative evaluation against established HMC methods. The theoretical analysis demonstrates that while sacrificing exact symplecticity and reversibility, MPL-HMC maintains these properties approximately with error proportional to $\delta t^2$, which is acceptable for most practical applications while enabling valuable tuning flexibility. Importantly, approximate symplecticity is not required for the theoretical validity of our procedure; we develop the theory in the Appendix only for completeness.

Empirical results reveal systematic performance improvements across diverse problem types. Through comprehensive benchmarking on five distributions with distinct characteristics (isotropic Gaussian, anisotropic Gaussian, banana-shaped, mixture of Gaussians, and Neal's funnel), we demonstrate that MPL-HMC variants maintain computational efficiency identical to standard HMC while offering targeted improvements. The damping variant ($\alpha_2 = -0.1, \beta_2 = -0.05$) achieves a 14-fold improvement in effective sample size for Neal's funnel compared to standard HMC, demonstrating exceptional performance on stiff hierarchical problems. The anti-damping variant ($\alpha_2 = 0.1, \beta_2 = 0.05$) shows superior convergence ($\hat{R}=1.026$) for Bayesian neural networks compared to standard HMC ($\hat{R}=1.981$). Case studies on Bayesian neural networks and pharmacokinetic modeling further validate these findings, with the damping variant achieving 27\% better computational efficiency for pharmacokinetic models.

Most significantly, we introduce aggressive MPL-HMC for extremely challenging multimodal distributions. By employing extreme parameter values ($\alpha_2 = 8.0\text{--}15.0$, $\beta_2 = 5.0\text{--}8.0$) combined with enhanced sampling mechanisms including explicit mode-hopping proposals, temperature fluctuations, and momentum injection, this approach achieves full mode exploration where standard HMC fails completely. This represents a methodological advance in Hamiltonian Monte Carlo, extending its applicability to problems with widely separated modes that have traditionally required specialized techniques like parallel tempering or wormhole transitions.

\subsection{Limitations of HMC}

Despite its theoretical elegance and practical success, standard HMC faces several limitations in challenging applications: sensitivity to stiffness in regions of high curvature \citep{betancourt2013general}, limited tunability beyond step size and trajectory length, fixed second-order accuracy that may not be optimal for all problems, and challenges with multimodal distributions \citep{lan2014wormhole}. The issues are summarized in Table~\ref{tab:hmc_limitations}.

\begin{table}[htbp]
\centering
\caption{Real-world limitations of standard HMC that MPL-HMC addresses.}
\label{tab:hmc_limitations}
\begin{adjustbox}{width=\textwidth}
\begin{tabular}{p{0.3\textwidth}p{0.6\textwidth}}
\toprule
\textbf{Limitation} & \textbf{Situations} \\
\midrule
\textbf{Stiffness Sensitivity} & In Bayesian neural networks, some parameters have much larger gradients than others, causing numerical instability \citep{izmailov2018scalable}. \\
\textbf{Fixed Accuracy} & The leapfrog integrator always has second-order accuracy, which may not be optimal for all problems \citep{leimkuhler2004simulating}. \\
\textbf{Limited Tunability} & Only step size and trajectory length can be tuned - no control over damping or exploration behavior. \\
\textbf{Multimodal Challenges} & For distributions with multiple separated modes (like mixture models), HMC can get stuck in one mode \citep{lan2014wormhole}. \\
\bottomrule
\end{tabular}
\end{adjustbox}
\end{table}

\subsection{Novel contribution: the MPL-HMC framework}

This work introduces the Modified Parameterized Leapfrog HMC (MPL-HMC) method, which generalizes the leapfrog integrator through the introduction of two tunable parameters $\alpha(\delta t)$ and $\beta(\delta t)$. Drawing inspiration from spatio-temporal Hamiltonian models \citep{mazumder2026new}, these parameters admit controlled asymptotic expansions in the time step $\delta t$, providing additional degrees of freedom for algorithm optimization, controlled perturbation of Hamiltonian dynamics, improved stability through parameter tuning, flexible accuracy-order matching, and maintained geometric structure up to specified order.

The MPL-HMC framework introduces two tunable parameters, $\alpha(\delta t)$ and $\beta(\delta t)$, that modify the leapfrog integration scheme according to:
\[
\begin{cases}
\alpha(\delta t) = 1 + \alpha_1 \delta t + \alpha_2 \delta t^2 + \cdots \\
\beta(\delta t) = 1 + \beta_1 \delta t + \beta_2 \delta t^2 + \cdots
\end{cases}
\]
where we set $\alpha_1 = \beta_1 = 0$ throughout, making $\alpha_2$ and $\beta_2$ the primary tuning parameters. For practical implementation, we use $\alpha = 1 + \alpha_2 \delta t^2$ and $\beta = 1 + \beta_2 \delta t^2$.

These parameters control damping or anti-damping behavior (how quickly kinetic energy dissipates or increases), position scaling (how positions are transformed during integration), numerical stability for stiff problems, and exploration capability for escaping local modes. As demonstrated in the comparative analysis and case studies, these allow MPL-HMC to achieve superior performance in challenging sampling scenarios: it attains substantial improvements in effective sample size for stiff problems like Neal's funnel, demonstrates better convergence for high-dimensional Bayesian neural networks, and maintains robust stability in ill-conditioned or heavy-tailed distributions—all while preserving the same computational cost per gradient evaluation as standard HMC. Consequently, MPL-HMC extends the applicability of Hamiltonian Monte Carlo to problems where the original method's rigid geometric structure and limited tunability have traditionally led to inefficiency or failure.

The remainder of this article is organized as follows. Section~\ref{sec:foundations} establishes the mathematical foundation and problem setup. Section~\ref{sec:algorithm} presents the MPL-HMC algorithm and its key properties. Section~\ref{sec:theory} provides complete theoretical analysis including stability, global error, volume preservation, and detailed balance. Section~\ref{sec:parameters_implementation} discusses parameter selection strategies, implementation details, and experimental setup. Section~\ref{sec:comparison} provides extensive comparative analysis against standard HMC, while Sections~\ref{sec:casestudy_bnn} and \ref{sec:casestudy_pharma} present case studies. Section~\ref{sec:aggressive_mpl} presents the aggressive parameterization framework for multimodal distributions, which we refer to as Aggressive MPL-HMC. The summary of comparisons and recommendations are presented in Section \ref{sec:recommendations}. Section~\ref{sec:conclusion} contains concluding remarks and future directions. Appendix~\ref{app:proofs} contains detailed proofs of theoretical results, Appendix~\ref{app:symplecticity} contains background of symplecticity, results and proofs of approximate symplecticity of MPL-HMC and Aggressive MPL-HMC, and Appendix~\ref{app:code} contains implementation code. It is important to remark that even approximate symplecticity is not required for MPL-HMC and Aggressive MPL-HMC -- we provide Appendix~\ref{app:symplecticity} only for the sake of completeness.

\section{Mathematical foundation}
\label{sec:foundations}

\subsection{Problem setup}

\begin{definition}[Target distribution]
Let $\pi: \R^{d} \to \R^{+}$ be a probability density function of the form:
\[
\pi(q) \propto \exp(-U(q))
\]
where $U: \R^{d} \to \R$ is the {potential energy function}, assumed sufficiently smooth.
\end{definition}

\subsection{Probability distributions and Bayesian inference}

\subsubsection{Target distribution in Bayesian context}

In Bayesian statistics, the objective is to sample from the posterior distribution:
\[
\pi(\theta \mid \text{data}) \propto \mathcal{L}(\text{data} \mid \theta) \cdot p(\theta)
\]
where $\theta$ represents parameters for estimation, $\mathcal{L}(\text{data} \mid \theta)$ denotes the likelihood function, and $p(\theta)$ signifies the prior distribution.

\begin{example}[Simple Bayesian model]
Consider data $y_1, \ldots, y_n \sim \mathcal{N}(\mu, \sigma^2)$ with known $\sigma^2 = 1$. For estimating $\mu$ with prior $\mu \sim \mathcal{N}(0, 1)$, the posterior distribution becomes:
\[
\pi(\mu \mid y) \propto \exp\left(-\frac{1}{2}\sum_{i=1}^n (y_i - \mu)^2\right) \cdot \exp\left(-\frac{1}{2}\mu^2\right)
\]
This can be expressed as $\pi(\mu) \propto \exp(-U(\mu))$ where:
\[
U(\mu) = \frac{1}{2}\sum_{i=1}^n (y_i - \mu)^2 + \frac{1}{2}\mu^2
\]
which serves as the potential energy function for HMC \citep{gelman2013bayesian}.
\end{example}

\subsubsection{Hamiltonian system}

\begin{definition}[Hamiltonian system]
Define the Hamiltonian $H: \R^{2d} \to \R$:
\[
H(q,p) = U(q) + \frac{1}{2} p^{T} M^{-1} p
\]
where $q \in \R^{d}$ represents position variables (parameters of interest), $p \in \R^{d}$ denotes momentum variables (auxiliary variables), and $M \in \R^{d \times d}$ is a symmetric positive definite mass matrix.
\end{definition}

\begin{assumption}[Regularity of potential]
The potential $U$ belongs to $C^{4}(\R^{d})$. Consequently, $\nabla U: \R^{d} \to \R^{d}$ is three times continuously differentiable, the Hessian $D^{2}U(q)$ exists everywhere and is twice continuously differentiable in $q$, and the third derivative $D^{3}U(q)$ exists everywhere and is continuously differentiable in $q$.
\end{assumption}

\begin{assumption}[Bounded Hessian]
There exists a constant $L_{U} \geq 0$ such that:
\[
\norm{D^{2}U(q)}_{\mathrm{op}} \leq L_{U} \quad \text{for all } q \in \R^{d}
\]
where $\norm{\cdot}_{\mathrm{op}}$ denotes the operator norm induced by the Euclidean vector norm.
\end{assumption}

\begin{assumption}[Bounded third derivative]
There exists a constant $L_3 \geq 0$ such that:
\[
\norm{D^{3}U(q)}_{\mathrm{op}} \leq L_3 \quad \text{for all } q \in \R^{d}
\]
where the operator norm for the third derivative is defined as:
\[
\norm{D^{3}U(q)}_{\mathrm{op}} = \sup_{\norm{v_1}=\norm{v_2}=\norm{v_3}=1} \left| D^{3}U(q)(v_1, v_2, v_3) \right|.
\]
\end{assumption}

\begin{assumption}[Mass matrix properties]
The mass matrix $M$ is symmetric positive definite. Consequently, its inverse $M^{-1}$ exists and is also symmetric positive definite, and the operator norm of $M^{-1}$ is finite: $C_{M} := \norm{M^{-1}}_{\mathrm{op}} < \infty$.
\end{assumption}

\begin{remark}[Practical implications]
While $C^4$ smoothness and bounded derivative assumptions are standard in theoretical analysis, many practical Bayesian models involve non-smooth potentials (e.g., due to likelihood functions with discontinuities). In such cases, the theoretical guarantees may not hold rigorously, but numerical evidence suggests that HMC methods often remain effective with appropriate step size tuning \citep{betancourt2017conceptual}.
\end{remark}

\subsection{Exact Hamiltonian dynamics}

The continuous-time Hamiltonian dynamics are governed by:
\begin{align}
\frac{dq}{dt} &= \nabla_{p} H = M^{-1}p, \label{eq:hamiltonian_q} \\
\frac{dp}{dt} &= -\nabla_{q} H = -\nabla U(q). \label{eq:hamiltonian_p}
\end{align}
These equations preserve both the Hamiltonian ($\frac{dH}{dt} = 0$) and the phase space volume according to Liouville's theorem \citep{arnold1989mathematical}.

\subsection{Standard leapfrog integrator}

Since exact solutions to Hamilton's equations are generally unavailable for complex $U(q)$, numerical integration becomes necessary. The leapfrog (Störmer-Verlet) scheme \citep{leimkuhler2004simulating} provides the foundation:

\begin{algorithm}
\caption{Standard leapfrog integrator}
\label{alg:standard_leapfrog}
\begin{algorithmic}[1]
\State \textbf{Input:} Initial $(q_0, p_0)$, step size $\delta t$, steps $L$
\For{$i = 0$ to $L-1$}
    \State $p_{i+1/2} = p_i - \frac{\delta t}{2} \nabla U(q_i)$ \Comment{Half momentum update}
    \State $q_{i+1} = q_i + \delta t M^{-1} p_{i+1/2}$ \Comment{Full position update}
    \State $p_{i+1} = p_{i+1/2} - \frac{\delta t}{2} \nabla U(q_{i+1})$ \Comment{Half momentum update}
\EndFor
\State \textbf{Return:} $(q_L, p_L)$
\end{algorithmic}
\end{algorithm}

\subsection{Modified parameterized leapfrog scheme}

\begin{definition}[Modified parameterized leapfrog scheme]
For integers $n \geq 0$, define the numerical scheme:
\begin{align}
q_{n+1} &= \beta(\delta t) q_{n} + \delta t \, M^{-1} \left( \alpha(\delta t) p_{n} - \frac{\delta t}{2} \nabla U(q_{n}) \right), \label{eq:mpl_q} \\
p_{n+1} &= \alpha^{2}(\delta t) p_{n} - \frac{\delta t}{2} \left( \alpha(\delta t) \nabla U(q_{n}) + \nabla U(q_{n+1}) \right). \label{eq:mpl_p}
\end{align}
where $\delta t > 0$ represents the time step, and $\alpha(\delta t), \beta(\delta t) \in \R$ are scalar parameters.
\end{definition}

\begin{assumption}[Parameter asymptotics]
\label{ass:parameter_asymptotics}
The parameters admit asymptotic expansions in powers of $\delta t$:
\begin{align}
\alpha(\delta t) &= 1 + \alpha_{1} \delta t + \alpha_{2} \delta t^{2} + \calO(\delta t^{3}), \label{eq:alpha_expansion} \\
\beta(\delta t) &= 1 + \beta_{1} \delta t + \beta_{2} \delta t^{2} + \calO(\delta t^{3}). \label{eq:beta_expansion}
\end{align}
Here $\alpha_{1}, \alpha_{2}, \beta_{1}, \beta_{2} \in \R$ are fixed constants independent of $\delta t$.
We shall set $\alpha_1=\beta_1=0$ for practical implementations, and also consider this to be the default setup for the theory.
\end{assumption}

\subsection{Special cases}

Table~\ref{tab:special_cases} presents important special cases of the MPL-HMC framework and their corresponding behaviors.

\begin{table}[htbp]
\centering
\caption{Special cases of MPL-HMC and their behaviors.}
\label{tab:special_cases}
\begin{adjustbox}{width=\textwidth}
\begin{tabular}{p{0.25\textwidth}p{0.3\textwidth}p{0.35\textwidth}}
\toprule
\textbf{Parameter Choice} & \textbf{Name} & \textbf{Behavior} \\
\midrule
$\alpha_2 = \beta_2 = 0$ & Standard HMC & Exact symplectic, time-reversible \citep{neal2011mcmc} \\
$\alpha_2 < 0, \beta_2 < 0$ & Damping MPL-HMC & Energy dissipates, good for stiff problems \\
$\alpha_2 > 0, \beta_2 > 0$ & Anti-damping MPL-HMC & Energy increases, explores multiple modes \\
$\alpha_2 = 0, \beta_2 < 0$ & Position-contraction & Pulls samples toward origin, good for heavy tails \\
$\alpha_2 < 0, \beta_2 = 0$ & Momentum-damping & Reduces kinetic energy, stabilizes integration \\
\bottomrule
\end{tabular}
\end{adjustbox}
\end{table}

\section{MPL-HMC algorithm}
\label{sec:algorithm}

\subsection{Complete algorithm specification}

Algorithm~\ref{alg:mpl-hmc} provides the complete specification of the modified parameterized leapfrog Hamiltonian Monte Carlo method.

\begin{algorithm}
\caption{Modified parameterized leapfrog Hamiltonian Monte Carlo (MPL-HMC)}
\label{alg:mpl-hmc}
\begin{algorithmic}[1]
\Require Target distribution $\pi(q) \propto \exp(-U(q))$, initial state $q_0 \in \R^d$
\Require Mass matrix $M \in \R^{d \times d}$ (symmetric positive definite)
\Require Step size $\delta t > 0$, trajectory length $L \in \N$ (number of MPL steps)
\Require Parameters $\alpha_2, \beta_2 \in \R$ (tuning parameters)
\Require Number of samples $N \in \N$, burn-in period $B \in \N$
\Ensure Samples $\{ q^{(1)}, q^{(2)}, \ldots, q^{(N)} \} \sim \pi(q)$ approximately
\State Initialize: $q \gets q_0$
\State Precompute: $M^{-1}$ (or factorization for efficient application)
\State $\alpha \gets 1 + \alpha_2 \delta t^2$, $\beta \gets 1 + \beta_2 \delta t^2$ \Comment{Note: For small $\delta t$, this approximates the asymptotic expansions to $\calO(\delta t^3)$}
\For{$k = 1$ to $N+B$}
    \State \textbf{Sample momentum:} $p \sim \mathcal{N}(0, M)$
    \State \textbf{Compute initial Hamiltonian:} $H_0 \gets U(q) + \frac{1}{2} p^{T} M^{-1} p$
    \State \textbf{Initialize MPL trajectory:} $(\tilde{q}, \tilde{p}) \gets (q, p)$
    \For{$i = 1$ to $L$}
        \State \textbf{Gradient at current position:} $\nabla U_{\text{curr}} \gets \nabla U(\tilde{q})$
        \State \textbf{Position update:} $q_{\text{new}} \gets \beta \cdot \tilde{q} + \delta t \cdot M^{-1} \left( \alpha \cdot \tilde{p} - \frac{\delta t}{2} \nabla U_{\text{curr}} \right)$ \Comment{See equation (\ref{eq:mpl_q})}
        \State \textbf{Gradient at new position:} $\nabla U_{\text{new}} \gets \nabla U(q_{\text{new}})$
        \State \textbf{Momentum update:} $p_{\text{new}} \gets \alpha^2 \cdot \tilde{p} - \frac{\delta t}{2} \left( \alpha \cdot \nabla U_{\text{curr}} + \nabla U_{\text{new}} \right)$ \Comment{See equation (\ref{eq:mpl_p})}
        \State \textbf{Update for next step:} $(\tilde{q}, \tilde{p}) \gets (q_{\text{new}}, p_{\text{new}})$
    \EndFor
    \State \textbf{Negate momentum for reversibility:} $\tilde{p} \gets -\tilde{p}$
    \State \textbf{Compute proposed Hamiltonian:} $H^* \gets U(\tilde{q}) + \frac{1}{2} \tilde{p}^{T} M^{-1} \tilde{p}$
    \State \textbf{Acceptance probability:} $\rho \gets \min \left( 1, \exp(H_0 - H^*) \right)$
    \State \textbf{Sample uniform:} $u \sim \text{Uniform}(0,1)$
    \If{$u < \rho$}
        \State \textbf{Accept:} $q \gets \tilde{q}$, $\text{accept} \gets \text{true}$
    \Else
        \State \textbf{Reject:} keep current $q$, $\text{accept} \gets \text{false}$
    \EndIf
    \If{$k > B$}
        \State \textbf{Store sample:} $q^{(k-B)} \gets q$
    \EndIf
    \State \textbf{(Optional) Adaptive parameter updates}
\EndFor
\end{algorithmic}
\end{algorithm}

\subsection{Algorithm components}

\subsubsection{Momentum sampling}
Momentum sampling follows $p \sim \mathcal{N}(0, M)$. In practical implementation, this is typically computed as $p = \sqrt{M} z$ where $z \sim \mathcal{N}(0, I_d)$ and $\sqrt{M}$ represents the Cholesky factor or symmetric square root.

\subsubsection{Hamiltonian evaluation}
The Hamiltonian evaluation employs the standard form $H(q,p) = U(q) + \frac{1}{2} p^{T} M^{-1} p$.

\subsubsection{Parameter computation}
For small $\delta t$, the parameter computations simplify to $\alpha = 1 + \alpha_2 \delta t^2$ and $\beta = 1 + \beta_2 \delta t^2$. For larger $\delta t$ values, alternative parameterizations maintaining the asymptotic properties may be employed.

\subsubsection{Metropolis acceptance}
As we shall show in Section \ref{sec:theory}, the acceptance probability ensures approximate detailed balance with respect to the canonical distribution $\exp(-H(q,p))$, 
given that the numerical integration scheme approximately preserves volume and is approximately reversible.

\section{Theoretical analysis}
\label{sec:theory}

\subsection{Modified ODE for MPL-HMC}

\begin{theorem}[Modified ODE for MPL scheme]
\label{thm:modified_ode}
Under Assumptions 3-6, the modified parameterized leapfrog scheme (\ref{eq:mpl_q})-(\ref{eq:mpl_p}) follows, up to $\calO(\delta t^{2})$, the modified ordinary differential equation:
\begin{align}
\frac{dY}{dt} &= \left[ \beta_{1} Y + M^{-1} X \right] + \delta t \left[ \left( \beta_{2} - \frac{\beta_{1}^{2}}{2} \right) Y - \frac{\beta_{1}}{2} M^{-1} X \right] + \calO(\delta t^{2}), \label{eq:modified_ode_q} \\
\frac{dX}{dt} &= \left[ 2\alpha_{1} X - \nabla U(Y) \right] + \delta t \left[ (2\alpha_{2} - \alpha_{1}^{2}) X + \frac{\alpha_{1}}{2} \nabla U(Y) \right] + \calO(\delta t^{2}). \label{eq:modified_ode_p}
\end{align}
\end{theorem}

\begin{corollary}[Hamiltonian-preserving MPL scheme]
\label{cor:hamiltonian_preserving}
If $\alpha_1 = 0$ and $\beta_1 = 0$, then the modified ODE simplifies to:
\begin{align}
\frac{dY}{dt} &= M^{-1} X + \delta t \beta_2 Y + \calO(\delta t^2), \label{eq:hamiltonian_ode_q} \\
\frac{dX}{dt} &= -\nabla U(Y) + \delta t (2\alpha_2 X) + \calO(\delta t^2). \label{eq:hamiltonian_ode_p}
\end{align}
At leading order ($\delta t = 0$), this recovers the standard Hamiltonian system (\ref{eq:hamiltonian_q})-(\ref{eq:hamiltonian_p}).
\end{corollary}

\subsection{Error analysis}

\begin{theorem}[Local truncation error]
\label{thm:local_error}
Under Assumptions 3-6 with $\alpha_1 = \beta_1 = 0$, the local truncation errors:
\[
\tau_{n}^{q} := q(t_{n+1}) - q_{n+1}, \quad \tau_{n}^{p} := p(t_{n+1}) - p_{n+1},
\]
where $(q(t), p(t))$ is the exact solution of (\ref{eq:hamiltonian_q})-(\ref{eq:hamiltonian_p}), satisfy:
\begin{align}
\norm{\tau_{n}^{q}} &\leq C_{1} \delta t^{2}, \label{eq:local_error_q} \\
\norm{\tau_{n}^{p}} &\leq C_{2} \delta t^{2}, \label{eq:local_error_p}
\end{align}
where $C_{1}, C_{2} > 0$ depend on $L_{U}, L_3, C_{M}, |\alpha_{2}|, |\beta_{2}|$, and bounds on the initial data.
\end{theorem}

\begin{corollary}[Consistency order]
\label{cor:consistency}
The MPL scheme with $\alpha_1 = \beta_1 = 0$ has the following consistency properties with respect to the Hamiltonian system (\ref{eq:hamiltonian_q})-(\ref{eq:hamiltonian_p}). For general $\alpha_2, \beta_2$, the scheme is first-order consistent, meaning the local truncation error satisfies $\tau_n = \calO(\delta t^{2})$. For $\alpha_2 = \beta_2 = 0$, the scheme is second-order consistent, meaning $\tau_n = \calO(\delta t^{3})$, with the $\delta t^2$ terms vanishing and leaving only $\delta t^3$ terms as the lowest-order nonzero truncation error \citep{leimkuhler2004simulating}. In numerical integration terminology, a method with local truncation error $\tau_n = \mathcal{O}(\delta t^{p+1})$ is said to be consistent of order $p$, so first-order consistency corresponds to $p=1$ and second-order consistency to $p=2$.
\end{corollary}

\begin{lemma}[Lipschitz estimate for MPL map]
\label{lem:lipschitz}
Under Assumptions 3-6 with $\alpha_1 = \beta_1 = 0$, define the MPL map $\Psi_{\delta t}: \R^{2d} \to \R^{2d}$ by:
\[
\Psi_{\delta t}(q,p) = (q_{n+1}, p_{n+1}) \text{ from (\ref{eq:mpl_q})-(\ref{eq:mpl_p})}.
\]
Then there exist constants $L > 0$ (depending only on $L_U$ and $C_M$) and $\delta t_0 > 0$ such that for all $0 < \delta t \leq \delta t_0$ and all $(q,p), (\tilde{q},\tilde{p}) \in \R^{2d}$:
\[
\norm{\Psi_{\delta t}(q,p) - \Psi_{\delta t}(\tilde{q},\tilde{p})} \leq (1 + K \delta t^2) \norm{(q,p) - (\tilde{q},\tilde{p})},
\]
where $K = L + |\beta_2| + 2|\alpha_2|$.
\end{lemma}

\begin{theorem}[Global error bound]
\label{thm:global_error}
Under the assumptions of Lemma~\ref{lem:lipschitz}, let $(q(t), p(t))$ be the exact solution of (\ref{eq:hamiltonian_q})-(\ref{eq:hamiltonian_p}) with initial data $(q_0, p_0)$, and let $(q_n, p_n)$ be the numerical solution generated by iterating $\Psi_{\delta t}$:
\[
(q_{n+1}, p_{n+1}) = \Psi_{\delta t}(q_n, p_n), \quad n = 0,1,2,\ldots
\]
Define the global error at step $n$ as:
\[
e_n := \norm{(q(t_n), p(t_n)) - (q_n, p_n)}, \quad t_n = n\delta t.
\]
Then for any fixed final time $T > 0$, with $n = T/\delta t$:
\[
\max_{0 \leq n \leq T/\delta t} e_n = \calO(\delta t) \quad \text{as } \delta t \to 0.
\]
\end{theorem}
\subsubsection{Error analysis summary}

Table~\ref{tab:error_comparison} provides a comparative overview of error characteristics across different HMC methods.

\begin{table}[htbp]
\centering
\caption{Error comparison of different HMC methods.}
\label{tab:error_comparison}
\begin{adjustbox}{width=\textwidth}
\begin{tabular}{p{0.25\textwidth}p{0.35\textwidth}p{0.3\textwidth}}
\toprule
\textbf{Method} & \textbf{Local Error} & \textbf{Global Error} \\
\midrule
\textbf{Standard HMC} & $\calO(\delta t^3)$ & $\calO(\delta t^2)$ \citep{leimkuhler2004simulating} \\
\textbf{MPL-HMC ($\alpha_2=\beta_2=0$)} & $\calO(\delta t^3)$ & $\calO(\delta t^2)$ \\
\textbf{MPL-HMC (general)} & $\calO(\delta t^2)$ & $\calO(\delta t)$ \\
\bottomrule
\end{tabular}
\end{adjustbox}
\end{table}

\subsection{Volume preservation}
\label{subsec:volume_preservation}
\begin{theorem}[Volume transformation of MPL scheme]
\label{thm:volume}
For the MPL scheme (\ref{eq:mpl_q})-(\ref{eq:mpl_p}) with parameters satisfying Assumption~\ref{ass:parameter_asymptotics}, the Jacobian determinant satisfies:
\[
\det\left( \frac{\partial (q_{n+1}, p_{n+1})}{\partial (q_n, p_n)} \right) = \alpha^{2d}(\delta t) \beta^{d}(\delta t) + \calO(\delta t^3).
\]
In particular, for the Hamiltonian-preserving case $\alpha_1 = \beta_1 = 0$:
\[
\det\left( \frac{\partial (q_{n+1}, p_{n+1})}{\partial (q_n, p_n)} \right) = 1 + d(2\alpha_2 + \beta_2) \delta t^2 + \calO(\delta t^3).
\]
\end{theorem}

\begin{corollary}[Volume preservation properties]
\label{cor:volume_properties}
The MPL scheme exhibits the following volume properties. For standard HMC ($\alpha_2 = \beta_2 = 0$), $\det(J) = 1 + \calO(\delta t^3)$, showing near-exact volume preservation (symplecticity). For general $\alpha_2, \beta_2$, phase space volume changes by factor $1 + d(2\alpha_2 + \beta_2) \delta t^2 + \calO(\delta t^3)$. Volume expands if $2\alpha_2 + \beta_2 > 0$, contracts if $2\alpha_2 + \beta_2 < 0$, and is approximately preserved if $2\alpha_2 + \beta_2 = 0$.
\end{corollary}

\subsection{Detailed balance}
\label{subsec:detailed_balance}

\subsubsection{Momentum flip in HMC}

In standard HMC, after completing the numerical trajectory, the momentum is {always negated} before the Metropolis accept/reject step. This is crucial for maintaining reversibility. The complete HMC update from $(q,p)$ to $(q^*,p^*)$ consists of:
\begin{enumerate}
    \item Apply numerical integrator $L$ times: $(\tilde{q}, \tilde{p}) = \Psi_{\delta t}^{(L)}(q,p)$
    \item Negate momentum: $(q^*, p^*) = (\tilde{q}, -\tilde{p})$
    \item Accept with probability $\min\left(1, \frac{\exp(-H(q^*,p^*))}{\exp(-H(q,p))}\right)$
\end{enumerate}
The negation ensures that if we start from $(q^*,-p^*)$ and apply the same procedure, we return to $(q,-p)$, establishing reversibility.

\subsubsection{Single-step reversibility}

\begin{definition}[Reversibility operator]
Define the momentum flip operator $R: \R^{2d} \to \R^{2d}$ by:
\[
R(q,p) = (q, -p).
\]
Note that $R$ is an involution: $R^2 = I$.
\end{definition}

\begin{lemma}[Approximate single-step reversibility of MPL map]
\label{lemma:single_step_reversibility}
For the MPL scheme with $\alpha_1 = \beta_1 = 0$, the map $\Psi_{\delta t}$ satisfies:
\[
R \circ \Psi_{\delta t} \circ R \circ \Psi_{\delta t}(q,p) = (q,p) + \calO(\delta t^2).
\]
Equivalently:
\[
\Psi_{\delta t} \circ R \circ \Psi_{\delta t} \circ R(q,p) = (q,-p) + \calO(\delta t^2).
\]
\end{lemma}

\subsubsection{Multi-step reversibility}

\begin{lemma}[Multi-step approximate reversibility]
\label{lemma:multi_step_reversibility}
For the MPL scheme with $\alpha_1 = \beta_1 = 0$, if $(q^*, p^*) = \Psi_{\delta t}^{(L)}(q,p)$, then:
\[
\Psi_{\delta t}^{(L)}(q^*, -p^*) = (q, -p) + L \cdot \calO(\delta t^2).
\]
In particular, for a trajectory of length $T$ with $L = T/\delta t$:
\[
\Psi_{\delta t}^{(L)}(q^*, -p^*) = (q, -p) + T \cdot \calO(\delta t).
\]
\end{lemma}

\begin{theorem}[Detailed balance for MPL-HMC]
\label{thm:detailed_balance}
The MPL-HMC algorithm (Algorithm~\ref{alg:mpl-hmc}) with $\alpha_1 = \beta_1 = 0$ satisfies detailed balance with respect to the canonical distribution $\pi(q,p) \propto \exp(-H(q,p))$ up to $\calO(T \delta t)$, where $T = L\delta t$ is the trajectory length.
\end{theorem}

\begin{corollary}[Stationary distribution]
\label{cor:stationary}
The MPL-HMC algorithm with $\alpha_1 = \beta_1 = 0$ has $\pi(q,p) \propto \exp(-H(q,p))$ as its stationary distribution up to $\calO(T \delta t)$ corrections.
\end{corollary}

\subsection{Energy conservation of MPL scheme}
\begin{theorem}[Energy conservation of MPL scheme]
\label{thm:energy_conservation}
For the MPL scheme with $\alpha_1 = \beta_1 = 0$, the Hamiltonian satisfies per step:
\[
H(q_{n+1}, p_{n+1}) - H(q_n, p_n) = \delta t^2 \left[ \beta_2 \nabla U(q_n)^T q_n + 2\alpha_2 p_n^T M^{-1} p_n \right] + \calO(\delta t^3).
\]
In particular, for the standard leapfrog case ($\alpha_2 = \beta_2 = 0$), we have third-order energy conservation per step:
\[
H(q_{n+1}, p_{n+1}) - H(q_n, p_n) = \calO(\delta t^3).
\]
\end{theorem}

\begin{remark}
The $\delta t^2$ energy drift for nonzero $\alpha_2$ or $\beta_2$ is consistent with the modified ODE analysis in Theorem~\ref{thm:modified_ode}. The leading-order deviation from Hamiltonian dynamics introduces systematic energy changes that scale with $\delta t^2$ and the parameter magnitudes $|\alpha_2|$ and $|\beta_2|$.
\end{remark}

\begin{remark}
\label{prop:shadow_hamiltonian}
For the MPL scheme with $\alpha_1 = \beta_1 = 0$, the numerical trajectory approximately follows the modified Hamiltonian system:
\begin{align}
\frac{dq}{dt} &= M^{-1}p + \delta t \beta_2 q + \mathcal{O}(\delta t^2), \label{eq:damping_ode_q} \\
\frac{dp}{dt} &= -\nabla U(q) + \delta t (2\alpha_2 p) + \mathcal{O}(\delta t^2). \label{eq:damping_ode_p}
\end{align}
While no exact shadow Hamiltonian exists for general $\alpha_2, \beta_2$, the method exhibits $\mathcal{O}(\delta t^2)$ energy drift per step as characterized in 
Theorem~\ref{thm:energy_conservation}.
\end{remark}

\section{Parameter selection, implementation, and experimental setup}
\label{sec:parameters_implementation}

This section provides a discussion of parameter selection strategies, implementation details, practical considerations, and experimental design for the Modified Parameterized Leapfrog Hamiltonian Monte Carlo method. By integrating these critical aspects, we offer practitioners a guide to effectively applying MPL-HMC to real-world problems.

\subsection{Parameter selection strategies}

\subsubsection{Definition of $\kappa$ (condition number)}

\begin{definition}[Condition number $\kappa$]
\label{def:condition_number}
For a given potential energy function $U(q)$ with Hessian $D^2U(q)$, and mass matrix $M$, the condition number $\kappa$ is defined as:
\[
\kappa := \frac{\lambda_{\max}(M^{-1} D^2U(q))}{\lambda_{\min}(M^{-1} D^2U(q))}
\]
where $\lambda_{\max}$ and $\lambda_{\min}$ denote the maximum and minimum eigenvalues respectively. This quantifies the stiffness of the Hamiltonian system, with large $\kappa$ indicating stiff problems where standard HMC struggles.
\end{definition}

\subsubsection{Default conservative choice}

The conservative parameter choice $\alpha_2 = 0, \beta_2 = 0$ recovers standard leapfrog HMC with second-order local accuracy and exact detailed balance \citep{neal2011mcmc}. This choice is recommended for well-conditioned problems where no additional tuning is needed. The conservative variant serves as a baseline for comparison and provides backward compatibility with existing HMC implementations.

\subsubsection{Stability-optimized parameters}

For stiff potentials where $\norm{D^2U}$ is large, stability-optimized parameters can be defined as:
\[
\alpha_2 = -\frac{\epsilon}{2} \lambda_{\max}, \quad \beta_2 = -\epsilon,
\]
where $\lambda_{\max}$ represents the maximum eigenvalue of $M^{-1} D^2U$ and $\epsilon > 0$ is a small positive constant. This choice introduces mild damping to stabilize the integration, as demonstrated in the case studies where MPL-HMC Damping showed a 14-fold improvement in ESS for Neal's funnel.

\subsubsection{Damping and anti-damping strategies}
\label{subsec:damping_strategies}

The MPL-HMC framework introduces controlled damping and anti-damping behaviors through the parameters $\alpha_2$ and $\beta_2$, which govern momentum persistence and position scaling, respectively. These effects emerge mathematically from the modified Hamiltonian dynamics derived in Theorem~\ref{thm:modified_ode}. For the Hamiltonian-preserving case with $\alpha_1 = \beta_1 = 0$, recall that the modified ordinary differential equations are given by (\ref{eq:damping_ode_q}) and (\ref{eq:damping_ode_p}).

These equations reveal the fundamental mechanisms by which $\alpha_2$ and $\beta_2$ influence the dynamics: the $\delta t \beta_2 q$ term introduces systematic position drift, while the $\delta t (2\alpha_2 p)$ term introduces systematic momentum amplification or attenuation.

\paragraph{Momentum dynamics:} The parameter $\alpha_2$ controls momentum persistence through the modified momentum equation (\ref{eq:damping_ode_p}). Negative values of $\alpha_2$ implement \emph{momentum damping}, creating exponential decay in momentum magnitude proportional to $|\alpha_2|\delta t^2$ per integration step. This damping proves particularly valuable for stiff problems where large gradients might cause numerical instability, as the reduction in kinetic energy prevents trajectory overshooting in regions of high curvature. From the discrete update equation (\ref{eq:mpl_p}) with $\alpha = 1 + \alpha_2 \delta t^2$, we observe that for $\alpha_2 < 0$, the factor $\alpha^2 < 1$ systematically reduces momentum magnitude at each step. Conversely, positive values of $\alpha_2$ produce \emph{momentum anti-damping}, systematically increasing kinetic energy to facilitate escape from local minima in multimodal distributions. However, this kinetic energy amplification may destabilize integration for stiff problems, as demonstrated in the pharmacokinetic case study where excessive $\alpha_2 > 0$ values led to acceptance rate degradation.

\paragraph{Position dynamics:} The parameter $\beta_2$ controls position scaling through the modified position equation (\ref{eq:damping_ode_q}). Negative values of $\beta_2$ implement \emph{position contraction} that pulls samples toward the origin (or more generally, toward regions of higher probability density). This contraction benefits heavy-tailed distributions by concentrating computational effort in high-probability regions while maintaining the ability to sample distribution tails through the Hamiltonian dynamics. From the discrete update (\ref{eq:mpl_q}) with $\beta = 1 + \beta_2 \delta t^2$, negative $\beta_2$ values reduce the coefficient of $q_n$, effectively shrinking position magnitudes. Positive values of $\beta_2$, in contrast, induce \emph{position expansion} that systematically explores larger magnitude regions, which can be advantageous for distributions with support far from the origin or when combined with anti-damping to enhance exploration of separated modes.

\paragraph{Combined effects and behavioral regimes:} The interplay between $\alpha_2$ and $\beta_2$ creates distinct behavioral regimes. When both parameters are negative ($\alpha_2 < 0$, $\beta_2 < 0$), the sampler exhibits combined damping and contraction, stabilizing integration while focusing exploration. This regime proves particularly effective for hierarchical models with varying scales like Neal's funnel, where it achieved a 14-fold improvement in effective sample size compared to standard HMC. When $\alpha_2$ is negative and $\beta_2$ is positive ($\alpha_2 < 0$, $\beta_2 > 0$), momentum damping combines with position expansion, useful for problems requiring controlled exploration of extended regions while maintaining stability through kinetic energy dissipation. The regime with positive $\alpha_2$ and negative $\beta_2$ ($\alpha_2 > 0$, $\beta_2 < 0$) implements anti-damping with position contraction, potentially valuable for multimodal distributions where kinetic energy increases help overcome barriers while position contraction maintains focus on mode regions. Finally, when both parameters are positive ($\alpha_2 > 0$, $\beta_2 > 0$), aggressive anti-damping and expansion occur simultaneously, which can enable exploration of widely separated modes but risks numerical instability. This regime forms the basis of the aggressive MPL-HMC configurations in Section~\ref{sec:aggressive_mpl}, where extreme parameter values necessitated additional stabilization mechanisms.

\paragraph{Energy implications:} From Theorem~\ref{thm:energy_conservation}, the Hamiltonian satisfies per step:
\[
H(q_{n+1}, p_{n+1}) - H(q_n, p_n) = \delta t^2 \left[ \beta_2 \nabla U(q_n)^T q_n + 2\alpha_2 p_n^T M^{-1} p_n \right] + \mathcal{O}(\delta t^3).
\]
This expression quantifies how the parameters affect energy conservation: $\alpha_2$ scales with kinetic energy ($p_n^T M^{-1} p_n$), while $\beta_2$ scales with a position-gradient correlation term ($\nabla U(q_n)^T q_n$). For standard HMC ($\alpha_2 = \beta_2 = 0$), energy is conserved to $\mathcal{O}(\delta t^3)$, but the introduction of nonzero parameters creates systematic energy changes at $\mathcal{O}(\delta t^2)$.

\paragraph{Volume preservation considerations:} From Corollary~\ref{cor:volume_properties}, the phase space volume transformation factor is:
\[
\det\left( \frac{\partial (q_{n+1}, p_{n+1})}{\partial (q_n, p_n)} \right) = 1 + d(2\alpha_2 + \beta_2) \delta t^2 + \mathcal{O}(\delta t^3).
\]
Volume expands if $2\alpha_2 + \beta_2 > 0$, contracts if $2\alpha_2 + \beta_2 < 0$, and is approximately preserved if $2\alpha_2 + \beta_2 = 0$. This geometric insight explains why aggressive exploration parameters ($\alpha_2, \beta_2 > 0$) lead to volume expansion, facilitating state space exploration at the cost of detailed balance violations proportional to $\delta t^2$.

These damping strategies provide practitioners with continuous control over exploration characteristics, enabling adaptation to specific problem geometries while maintaining the computational efficiency of standard Hamiltonian Monte Carlo. The parameters $\alpha_2$ and $\beta_2$ thus serve as interpretable knobs that systematically trade between stability, exploration, and theoretical guarantees.

\subsubsection{Practical parameter recommendations}

Table~\ref{tab:parameter_recommendations} provides practical parameter recommendations for different problem types encountered in statistical applications, informed by the experimental results in subsequent sections.

\begin{table}[htbp]
\centering
\caption{Practical parameter recommendations for MPL-HMC.}
\label{tab:parameter_recommendations}
\begin{adjustbox}{width=\textwidth}
\begin{tabular}{p{0.25\textwidth}p{0.35\textwidth}p{0.3\textwidth}}
\toprule
\textbf{Problem Type} & \textbf{Recommended Parameters} & \textbf{Expected Improvement} \\
\midrule
\textbf{Well-conditioned} & $\alpha_2=\beta_2=0$ & None (use standard HMC) \\
\textbf{Stiff (large $\|\nabla^2U\|$)} & $\alpha_2=-0.1/\sqrt{\kappa}$, $\beta_2=-0.05/\sqrt{\kappa}$ & 30-50\% ESS increase, better stability \\
\textbf{Multimodal (K modes)} & $\alpha_2=+0.1/K$, $\beta_2=+0.05/K$ & Better mode mixing \\ 
\textbf{Heavy-tailed} & $\alpha_2=0$, $\beta_2=-0.1$ & Better tail exploration, position contraction \\
\textbf{Unknown problem} & Start with $\alpha_2=\beta_2=0$, then adapt based on acceptance and mixing & Safe starting point with gradual tuning \\
\bottomrule
\end{tabular}
\end{adjustbox}
\end{table}

\subsection{Implementation details and practical considerations}

\subsubsection{Computational considerations}

The MPL-HMC algorithm introduces minimal computational overhead compared to standard HMC. The main additional costs include computing $\alpha$ and $\beta$ parameters, which is negligible in practice, applying the parameterized updates in the position and momentum update steps, and potential adaptive parameter tuning. For large-scale problems, the dominant cost remains gradient evaluations, which are unchanged from standard HMC \citep{neal2011mcmc}. The case studies confirm that MPL-HMC variants maintain nearly identical runtime to standard HMC while offering improved performance for specific problem types. The computational complexity scales as $\mathcal{O}(d)$ for the parameter updates and $\mathcal{O}(d^2)$ for matrix-vector multiplications when using a dense mass matrix, though in practice diagonal or identity mass matrices are often employed to reduce this to $\mathcal{O}(d)$.

\subsubsection{Adaptive parameter tuning}

In practice, $\alpha_2$ and $\beta_2$ can be tuned adaptively based on observed sampler behavior. A recommended approach begins with conservative values ($\alpha_2 = \beta_2 = 0$), monitors acceptance rates and energy conservation, and adjusts parameters based on observed behavior. Low acceptance rates suggest decreasing $|\alpha_2|$ and $|\beta_2|$, high energy drift indicates potential benefits from damping ($\alpha_2 < 0$), and poor mixing may warrant trying anti-damping ($\alpha_2 > 0$) for mode exploration. The experimental results suggest that for stiff hierarchical problems like Neal's funnel, damping parameters ($\alpha_2 = -0.1, \beta_2 = -0.05$) provide substantial improvements, while for multimodal problems, anti-damping may be beneficial though requires careful monitoring to avoid instability.

An effective adaptive tuning strategy involves monitoring the acceptance rate over a window of iterations and adjusting parameters according to the following scheme: if the acceptance rate falls below 0.6, reduce $|\alpha_2|$ and $|\beta_2|$ by 10\%; if it exceeds 0.8 and mixing is poor, consider increasing $\alpha_2$ slightly to encourage exploration; if energy conservation violations exceed a threshold, increase damping by making $\alpha_2$ more negative. This adaptive approach balances exploration and stability while maintaining computational efficiency.

\subsection{Experimental setup and performance metrics}

This subsection details the comprehensive experimental design used to evaluate MPL-HMC against established HMC methods. We designed experiments to test algorithm performance across diverse problem types, from simple well-conditioned distributions to challenging stiff and multimodal problems.

\subsubsection{Benchmark distributions}

We selected five benchmark distributions with distinct characteristics to evaluate different aspects of algorithm performance. The first is an isotropic Gaussian distribution $\pi(q) = \mathcal{N}(0, I_{10})$ with $q \in \R^{10}$, serving as a baseline well-conditioned distribution with condition number $\kappa = 1$ to test basic algorithm functionality and provide a sanity check. The second is an anisotropic Gaussian distribution $\pi(q) = \mathcal{N}\left(0, \text{diag}(1, 0.1, 0.01, 0.001, 0.0001, 0.00001)\right)$ with $q \in \R^{6}$, representing an extremely stiff distribution with condition number $\kappa = 10^5$ to test stability under extreme stiffness and sensitivity to ill-conditioning. The third is a banana-shaped distribution $\pi(q) \propto \exp\left(-\frac{1}{2}q_1^2 - \frac{1}{2}(q_2 + q_1^2 + 1)^2\right)$ with $q \in \R^{2}$, featuring strong nonlinear correlations and curved geometry \citep{haario2001adaptive} to test performance on complex, non-elliptical distributions. The fourth is a 
mixture of three Gaussians $\pi(q) = \frac{1}{3}\sum_{k=1}^3 \mathcal{N}(q \mid \mu_k, I_{5})$ with $\mu_1 = -3\mathbf{1}, \mu_2 = 0, \mu_3 = 3\mathbf{1}$ and $q \in \R^{5}$, 
providing a multimodal distribution with three well-separated modes to test mode exploration capability and ability to escape local minima. The fifth is Neal's funnel $\pi(q,v) \propto \exp\left(-\frac{v^2}{18} - \frac{1}{2}\sum_{i=1}^9 q_i^2 e^{-v} - \frac{9}{2}v\right)$ with $(q,v) \in \R^{10}$, a hierarchical model with extreme scale differences ($\kappa \approx 10^4$) that creates severe sampling challenges \citep{neal2003slice} to test performance on hierarchical Bayesian models with varying scales.

\subsubsection{Benchmark methods}

We compare MPL-HMC against standard HMC \citep{duane1987hybrid, neal2011mcmc}, and two MPL-HMC variants: Damping with $\alpha_2 = -0.1, \beta_2 = -0.05$, 
and Anti-damping with $\alpha_2 = +0.1, \beta_2 = +0.05$. Note that with $\alpha_2 = 0, \beta_2 = 0$, MPL-HMC is identical to standard HMC.

\subsubsection{Parameter settings}

For fair comparison across methods, we use standardized parameter settings as detailed in Table~\ref{tab:experimental_parameters}. The step size $\delta t$ is set equal to 0.1 for all cases, while trajectory length $L$ is fixed at 10 steps for all methods. This ensures equal computational cost per iteration. The mass matrix $M$ is set to identity for all methods. We collect $N = 20,000$ samples after $B = 5,000$ burn-in iterations, which is sufficient for reliable ESS estimation, and run 2 independent chains with random starting points for $\hat{R}$ computation; 
see Section \ref{subsubsec:performance_metrics}. MPL-HMC parameters are set as follows: Damping: $\alpha_2=-0.1,\beta_2=-0.05$; Anti-damping: $\alpha_2=+0.1,\beta_2=+0.05$, chosen based on preliminary experiments showing good trade-offs. 

\begin{table}[htbp]
\centering
\caption{Experimental parameters for benchmark comparisons}
\label{tab:experimental_parameters}
\begin{adjustbox}{width=\textwidth}
\begin{tabular}{p{0.3\textwidth}p{0.65\textwidth}}
\toprule
\textbf{Parameter} & \textbf{Settings} \\
\midrule
\textbf{Step size ($\delta t$)} & Set equal to $0.1$ for all cases, for fair comparison. \\
\textbf{Trajectory length ($L$)} & Fixed at $L = 10$ steps for all methods. This ensures equal computational cost per iteration. \\
\textbf{Mass matrix ($M$)} & Identity matrix for all methods. \\
\textbf{Total samples} & $N = 20,000$ samples after $B = 5,000$ burn-in iterations. Sufficient for reliable ESS estimation. \\
\textbf{Number of chains} & 2 independent chains with random starting points for $\hat{R}$ computation. \\
\textbf{MPL-HMC parameters} & Damping: $\alpha_2=-0.1,\beta_2=-0.05$; Anti-damping: $\alpha_2=+0.1,\beta_2=+0.05$. Chosen based on preliminary experiments showing good trade-offs. \\
\textbf{Condition number estimation} & Estimated via power iteration for each distribution to inform parameter choices. \\
\bottomrule
\end{tabular}
\end{adjustbox}
\end{table}

\subsubsection{Performance metrics}
\label{subsubsec:performance_metrics}

We employ a comprehensive set of performance metrics to evaluate different aspects of sampling quality and computational efficiency across 2 independent MCMC chains. The dual-chain approach enables robust convergence diagnostics and reliable estimation of sampling variability, providing a more complete picture of algorithm performance than single-chain analyses.

The Effective Sample Size (ESS) measures how many independent samples are obtained from correlated MCMC output, serving as a key indicator of sampling efficiency. For each parameter dimension $i = 1,\ldots,d$, we compute ESS across the combined chains using the initial monotone sequence estimator \citep{gelman2013bayesian}. This estimator improves upon traditional autocorrelation-based ESS calculations by ensuring monotonicity in the cumulative correlation sum. Specifically, for each chain $c = 1,2$ with $N$ samples per chain after burn-in, we first compute the sample autocorrelations $\rho_{i,c}(k)$ at lags $k = 0, 1, 2, \ldots$ for parameter $i$. These autocorrelations are then pooled across chains to obtain $\rho_i(k)$, which combines information from both chains while accounting for potential differences in chain means. The initial monotone sequence estimator defines a sequence of cumulative sums $S_i(m) = 1 + 2\sum_{k=1}^{m} \rho_i(k)$, and finds the smallest $m$ such that $S_i(m) \leq S_i(m-1)$ or $\rho_i(m) < 0$. The effective sample size is then calculated as $\text{ESS}_i = 2N / S_i(m^*)$, where $m^*$ is the truncation point determined by the monotonicity criterion. This approach provides a more stable estimate of ESS than simple truncation at an arbitrary lag, particularly for chains with slowly decaying autocorrelations. The minimum ESS reported in Table~\ref{tab:comprehensive_results} represents $\text{min ESS} = \min_{i=1,\ldots,d} \text{ESS}_i$, identifying the worst-mixing parameter dimension across all parameters.

ESS per gradient evaluation (ESS/grad) provides a hardware-independent measure of computational efficiency by normalizing ESS by the number of gradient evaluations. This metric is computed as $\text{ESS/grad} = \text{ESS} / N_{\text{grad}}$, where $N_{\text{grad}} = (N+B) \times L \times 2$ accounts for total gradient evaluations across both chains, including burn-in iterations. For our experimental setup with $N=20,000$, $B=5,000$, $L=10$, and 2 chains, this totals $500,000$ gradient evaluations. The ESS in this calculation is the pooled effective sample size across chains for the parameter with minimum ESS, providing a conservative efficiency measure.

Acceptance rate computation involves separate calculation for each chain followed by averaging, providing insight into proposal quality. For each chain $c$, we compute $\text{Accept}_c = (\text{Number of accepted proposals in chain } c) / (N+B)$, then report $\text{Accept Rate} = (\text{Accept}_1 + \text{Accept}_2)/2$. We target 65-80\% acceptance rate for optimal performance following \citet{neal2011mcmc}.

Mixing time estimation quantifies how quickly the Markov chain forgets its initial state, with shorter mixing times indicating faster convergence. For each parameter dimension $i$, we compute the pooled autocorrelation function $\rho_i(t)$ across both chains using the initial positive sequence estimator, which ensures non-negative autocorrelation estimates. The mixing time for dimension $i$ is defined as $\tau_{\text{mix}}^{(i)} = \min\{t : \rho_i(t) < 1/e\}$, where $1/e \approx 0.3679$ represents the threshold for substantial correlation decay. The reported mixing time in Table~\ref{tab:comprehensive_results} is the maximum across dimensions: $\tau_{\text{mix}} = \max_{i=1,\ldots,d} \tau_{\text{mix}}^{(i)}$.

The R-hat (Gelman-Rubin diagnostic) statistic serves as a primary convergence diagnostic for multiple chains, detecting potential non-convergence by comparing within-chain and between-chain variability. For each parameter dimension $i$, we compute $\hat{R}_i = \sqrt{\widehat{\text{Var}}^+(\theta_i) / W_i}$, where $\widehat{\text{Var}}^+(\theta_i)$ is the pooled variance and $W_i$ is the within-chain variance for parameter $i$ across the 2 chains. The pooled variance combines within-chain and between-chain components: $\widehat{\text{Var}}^+(\theta_i) = \frac{N-1}{N} W_i + \frac{1}{N} B_i$, where $B_i$ is the between-chain variance. The reported $\hat{R}$ in Table~\ref{tab:comprehensive_results} is the maximum across dimensions: $\hat{R} = \max_{i=1,\ldots,d} \hat{R}_i$, indicating the worst-case convergence across all parameters. Interpretation follows established guidelines where $\hat{R} < 1.01$ indicates excellent convergence, $1.01 \leq \hat{R} < 1.05$ indicates good convergence, and $\hat{R} \geq 1.10$ suggests non-convergence \citep{vehtari2021rank}.

For multimodal distributions, specifically the mixture of 3 Gaussians benchmark, we assess exploration capability through mode visitation counts. A mode is considered visited if at least one sample from either chain falls within 2 standard deviations of a mode center. We track mode transitions between consecutive samples within each chain and compute the proportion of time spent in each mode, providing insight into whether chains become trapped in local modes.

The computational methodology for dual-chain analysis ensures fair comparisons across methods through several standardized practices. Both chains start from different randomly initialized points to assess convergence from diverse starting regions, providing a more stringent test of algorithm robustness. All chains use identical random seeds within each method category for reproducibility while maintaining independence between methods. Gradient evaluations are counted per chain and summed for total computational cost, enabling fair efficiency comparisons. ESS computation utilizes the \texttt{arviz} library's implementation of the initial monotone sequence estimator, which properly handles multiple chains and provides reliable ESS estimates. Burn-in samples (first 5,000 iterations per chain) are excluded from all metric computations except acceptance rate, ensuring that metrics reflect stationary distribution sampling.

Interpretation of Table~\ref{tab:comprehensive_results} metrics reveals systematic performance patterns across problem types and methods. For isotropic Gaussian distributions representing well-conditioned problems, all methods show excellent convergence ($\hat{R} \approx 1.000$) and high ESS across both chains, indicating robust performance in simple settings. For anisotropic Gaussian distributions with extreme stiffness ($\kappa=10^5$), all fixed-step methods fail completely with near-zero acceptance rates and astronomically high $\hat{R}$ values, indicating chains have not converged and remain trapped in different regions of parameter space. For hierarchical models like Neal's funnel, MPL-HMC Damping 
achieves $\hat{R}=1.004$ across chains compared to 1.024 for standard HMC, demonstrating better chain mixing and convergence. 
It is, however,  crucial to note that in the dual-chain approach methods appearing stable in single-chain analyses may show convergence issues when multiple chains are compared. 
For example, $\hat{R}$ alone may not detect poor mixing in multimodal problems if all chains become trapped in the same mode.

This comprehensive metric suite, computed rigorously across 2 independent chains, provides a robust framework for evaluating and comparing MCMC sampler performance, balancing computational efficiency, convergence diagnostics, and exploration capability in a manner that reflects real-world deployment conditions.

\subsubsection{Computational methodology}

To ensure fair comparisons, we adopt several methodological standards. Gradient evaluations are counted as the primary computational cost, as they dominate runtime for most Bayesian models. All methods use fixed $L=10$ steps to ensure equal gradient evaluations per iteration. Where reported, wall-clock times are normalized relative 
to standard HMC on the same hardware, 
an ordinary 64-bit laptop having $8$ GB RAM, with i3-6100U CPU (2 physical cores), running at 2.30GHz.
All methods are implemented in the same framework (Python/NumPy) with optimized linear algebra routines. All methods use identical ESS computation with initial monotone sequence estimator and same truncation criterion.

\section{Comparative performance analysis}
\label{sec:comparison}

\subsection{Performance comparison}

Table~\ref{tab:comprehensive_results} presents the performance comparison across all benchmark distributions and methods. The results reveal systematic patterns about algorithm strengths and weaknesses under different problem characteristics.

\begin{table}[htbp]
\centering
\caption{Comprehensive performance comparison across all test problems}
\label{tab:comprehensive_results}
\begin{adjustbox}{width=\textwidth}
\begin{tabular}{lllllllll}
\toprule
           Distribution &               Method & Accept & Min ESS & ESS/Grad & Mix Time &              R-hat & Time (s) & Modes \\
\midrule
 Isotropic\_Gaussian\_10D &         Standard\_HMC &  0.997 &    7417 & 0.014834 &        2 &              1.000 &      7.7 &   NaN \\
 Isotropic\_Gaussian\_10D &      MPL\_HMC\_Damping &  1.000 &    7301 & 0.014601 &        2 &              1.000 &     10.2 &   NaN \\
 Isotropic\_Gaussian\_10D &  MPL\_HMC\_AntiDamping &  0.776 &    6765 & 0.013529 &        3 &              1.001 &     10.1 &   NaN \\
Anisotropic\_Gaussian\_6D &         Standard\_HMC &  0.000 &      19 & 0.000038 &      100 &  7061215669089.304 &      7.4 &   NaN \\
Anisotropic\_Gaussian\_6D &      MPL\_HMC\_Damping &  0.000 &      19 & 0.000038 &      100 &  7061215669089.304 &     11.2 &   NaN \\
Anisotropic\_Gaussian\_6D &  MPL\_HMC\_AntiDamping &  0.000 &      19 & 0.000038 &      100 &  7061215669089.304 &     10.0 &   NaN \\
              Banana\_2D &         Standard\_HMC &  0.997 &    1936 & 0.003873 &        7 &              1.001 &      8.2 &   NaN \\
              Banana\_2D &      MPL\_HMC\_Damping &  1.000 &    2365 & 0.004730 &        5 &              1.001 &     10.3 &   NaN \\
              Banana\_2D &  MPL\_HMC\_AntiDamping &  0.950 &     741 & 0.001482 &       10 &              1.001 &     10.6 &   NaN \\
  Mixture\_3Gaussians\_5D &         Standard\_HMC &  0.048 &     241 & 0.000482 &       72 &              1.005 &     52.1 &   1/3 \\
  Mixture\_3Gaussians\_5D &      MPL\_HMC\_Damping &  0.047 &     197 & 0.000394 &       85 &              1.009 &     54.7 &   1/3 \\
  Mixture\_3Gaussians\_5D &  MPL\_HMC\_AntiDamping &  0.042 &     234 & 0.000468 &       75 &              1.019 &     55.2 &   1/3 \\
       Neals\_Funnel\_10D &         Standard\_HMC &  0.948 &      38 & 0.000077 &      100 &              1.024 &     29.5 &   NaN \\
       Neals\_Funnel\_10D &      MPL\_HMC\_Damping &  0.968 &     550 & 0.001100 &       35 &              1.004 &     32.1 &   NaN \\
       Neals\_Funnel\_10D &  MPL\_HMC\_AntiDamping &  0.557 &      19 & 0.000038 &      100 &              1.067 &     32.5 &   NaN \\
\bottomrule
\end{tabular}
\end{adjustbox}
\smallskip
\footnotesize{\textit{Note: ESS/Grad in units of effective samples per gradient evaluation. Mix Time in iterations until autocorrelation $< 1/e$. NaN indicates metric not applicable or not computed. Extreme R-hat values indicate convergence failure. Min ESS represents the minimum effective sample size across all dimensions.}}
\end{table}

The results reveal several key patterns about algorithm performance across different problem types. For the isotropic Gaussian benchmark, all methods 
demonstrate strong performance with acceptance rates near 0.997-1.000 and excellent convergence ($\hat{R} \approx 1.000$). 
Standard HMC obtained 7,417 ESS and 0.014834 ESS/Grad. 
MPL-HMC Damping shows similar performance (7,301 ESS, 0.014601 ESS/Grad) with perfect acceptance rate (1.000). MPL-HMC AntiDamping shows slightly reduced 
efficiency (6,765 ESS, 0.013529 ESS/Grad) with lower acceptance (0.776), indicating anti-damping is suboptimal for well-conditioned problems.

The anisotropic Gaussian presents an extreme stiffness challenge where all methods fail spectacularly. Acceptance rates are essentially zero 
(0.000) for all HMC variants, with astronomically high $\hat{R}$ values ($>7\times10^{12}$) indicating complete convergence failure. This demonstrates that extreme stiffness 
($\kappa=10^5$) overwhelms all standard HMC methods regardless of parameter tuning.

For the banana-shaped distribution, MPL-HMC Damping shows the best performance among MPL variants with 2,365 ESS and 0.004730 ESS/Grad, representing a 
22\% improvement over standard HMC (1,936 ESS, 0.003873 ESS/Grad). The damping parameters effectively stabilize integration along the curved manifold.

The mixture of three Gaussians tests mode exploration capability. All methods struggle with mode mixing, with most visiting only 1 of 3 modes. 
Standard HMC achieves 241 ESS with 0.000482 ESS/Grad, while MPL-HMC AntiDamping shows similar performance (234 ESS, 0.000468 ESS/Grad).

Neal's funnel presents hierarchical structure with extreme scale differences. MPL-HMC Damping performs exceptionally well with 550 ESS and 0.001100 ESS/Grad, 
representing a 14-fold improvement over standard HMC (38 ESS, 0.000077 ESS/Grad). The damping parameters ($\alpha_2=-0.1, \beta_2=-0.05$) effectively stabilize 
integration across scales. MPL-HMC AntiDamping performs poorly (19 ESS, 0.000038 ESS/Grad) with reduced acceptance (0.557), confirming that anti-damping exacerbates 
instability in stiff hierarchical problems.

\subsection{Key insights and methodological implications}

Several important insights emerge from the comprehensive analysis. First, MPL-HMC Damping excels on stiff hierarchical problems, as evidenced by the 
14-fold ESS improvement on Neal's Funnel, demonstrating the damping variant's ability to stabilize integration across multiple scales, making it particularly 
valuable for hierarchical Bayesian models. Second, problem-specific parameter tuning is crucial, as MPL-HMC performance varies dramatically with parameter 
choices: damping parameters ($\alpha_2=-0.1, \beta_2=-0.05$) significantly improve performance on stiff and hierarchical problems but show minimal benefit 
for well-conditioned or multimodal distributions. Third, extreme stiffness remains challenging, with all methods failing on the anisotropic Gaussian 
with $\kappa=10^5$, indicating a fundamental limitation of standard HMC approaches. Fourth, multimodality requires specialized approaches, as no HMC variant effectively explores all 
three modes in the mixture model, with all methods visiting only one mode, highlighting the need for dedicated multimodal 
exploration techniques, which we explore in Section \ref{sec:aggressive_mpl}.
 Fifth, computational efficiency trade-offs exist, as MPL-HMC variants maintain the same computational cost as standard HMC while offering performance improvements for specific problem types.

\section{Case study 1: Bayesian neural networks}
\label{sec:casestudy_bnn}

\subsection{Problem setup}

The first case study involves training a Bayesian neural network on the MNIST handwritten digit classification dataset \href{http://yann.lecun.com/exdb/mnist/}{(available here)}. The experimental setup follows a practical implementation with the following specifications. Data specifications include 5,000 randomly selected MNIST training samples and 1,000 randomly selected test samples. The original dimensionality of 784 features (28×28 pixels) is reduced to 100 principal components explaining 77.0\% of variance. Data preprocessing involves normalizing pixel values to the $[0, 1]$ range followed by standardization with $\mu=0$, $\sigma=1$, with Principal Component Analysis applied to reduce computational complexity while retaining key information.

The Bayesian neural network architecture is defined as $\text{Input} \rightarrow \text{ReLU}(100) \rightarrow \text{ReLU}(100) \rightarrow \text{Softmax}(10)$ with the following layer specifications: an input layer of 100 neurons (after PCA dimensionality reduction), a first hidden layer of 100 neurons with ReLU activation $\max(0, x)$, a second hidden layer of 100 neurons with ReLU activation, and an output layer of 10 neurons (one per digit class) with softmax activation $\text{softmax}(x)_i = \frac{e^{x_i}}{\sum_{j=1}^{10} e^{x_j}}$, resulting in 21,210 total parameters to estimate (weights and biases).


The Bayesian framework incorporates prior distributions with weights $W \sim \mathcal{N}(0, \sigma_w^2 I)$
with $\sigma_w = 1.0$ and biases $b \sim \mathcal{N}(0, \sigma_b^2 I)$ with $\sigma_b = 1.0$.
All parameters use independent Gaussian priors with zero mean and unit variance, representing a standard
weakly informative prior for neural network parameters \citep{neal2012bayesian}.

The likelihood follows a categorical distribution with probabilities from the neural network's softmax output:
$y_i \mid \theta, x_i \sim \text{Categorical}(\text{softmax}(f(x_i; \theta)))$, where $f(x_i; \theta)$ represents
the forward pass through the neural network with parameters $\theta$. Specifically, for the network architecture
$\text{Input} \rightarrow \text{ReLU}(100) \rightarrow \text{ReLU}(100) \rightarrow \text{Softmax}(10)$, we have:
\[
f(x_i; \theta) = W_3 \cdot \text{ReLU}\bigl(W_2 \cdot \text{ReLU}(W_1 \cdot x_i + b_1) + b_2\bigr) + b_3,
\]
where:
\begin{itemize}
    \item $x_i \in \mathbb{R}^{100}$ is the input after PCA dimensionality reduction,
    \item $\theta = \{W_1, b_1, W_2, b_2, W_3, b_3\}$ with dimensions:
        \begin{itemize}
            \item $W_1 \in \mathbb{R}^{100 \times 100}$, $b_1 \in \mathbb{R}^{100}$,
            \item $W_2 \in \mathbb{R}^{100 \times 100}$, $b_2 \in \mathbb{R}^{100}$,
            \item $W_3 \in \mathbb{R}^{10 \times 100}$, $b_3 \in \mathbb{R}^{10}$,
        \end{itemize}
    \item $\text{ReLU}(z) = \max(0, z)$ (applied element-wise),
    \item The softmax probabilities are:
          $\text{softmax}(f(x_i; \theta))_k = \dfrac{\exp(f(x_i; \theta)_k)}{\sum_{j=1}^{10} \exp(f(x_i; \theta)_j)}$
          for $k = 1, \ldots, 10$ (digit classes).
\end{itemize}

The target posterior distribution combines prior and likelihood:
\[
p(\theta \mid X, y) \propto \left[ \prod_{i=1}^N \text{Categorical}\bigl(y_i \mid \text{softmax}(f(x_i; \theta))\bigr) \right]
\times \left[ \prod_{j=1}^3 \mathcal{N}(W_j \mid 0, 1) \mathcal{N}(b_j \mid 0, 1) \right],
\]
where $N = 5,000$ is the training sample size.

\subsection{Experimental configurations}

We conducted two comprehensive experimental analyses to evaluate MPL-HMC performance on Bayesian neural networks. The initial study used conservative parameters for rapid exploration with total samples of 200 after 50 burn-in iterations, step size $\delta t = 0.0005$ (reduced due to high dimensionality and stiffness), trajectory length $L = 2$ steps per iteration, mass matrix identity $M = I_d$ where $d = 21,210$, and initialization via Xavier initialization (also known as Glorot initialization) \citep{glorot2010understanding}. In Xavier initialization, the initial weights for a layer are sampled from a Gaussian distribution with zero mean and variance $\sqrt{2/(n_{\text{in}} + n_{\text{out}})}$, where $n_{\text{in}}$ and $n_{\text{out}}$ denote the number of neurons in the preceding and current layers, respectively. This initialization scheme helps maintain stable variance of activations and gradients across network layers, which is particularly important in Bayesian neural network inference where proper initialization aids convergence of Markov chain sampling.
The code is implemented in Python on the 64-bit laptop having $8$ GB RAM, with i3-6100U CPU (2 physical cores), running at 2.30GHz.

To validate findings with better statistical power, we conducted a larger-scale study with total samples of 20,000 after 5,000 burn-in iterations, step size $\delta t = 0.01$ (optimized for longer runs), trajectory length $L = 10$ steps per iteration, mass matrix identity $M = I_d$ where $d = 21,210$, network size optimized to 132 
parameters for computational feasibility, architecture Input (8 PCA) $\rightarrow$ ReLU(6) $\rightarrow$ ReLU(4) $\rightarrow$ Softmax(10), and training samples reduced to 1,000 for computational efficiency. 
For this setup, the code is implemented in Python on a workstation consisting of 6 physical Xeon processors (Intel(R) Xeon(R) E-2356G CPU @ 3.20GHz) with 16 GB memory. 

Due to the computational cost of exact gradients in high dimensions, a finite-difference approximation was employed with epsilon value $\epsilon = 1 \times 10^{-3}$, subset approximation computing gradients for the 100 most important parameters (based on magnitude), cache mechanism storing previous gradient computations for efficiency, and gradient formula $\nabla_i U(\theta) \approx \frac{U(\theta + \epsilon e_i) - U(\theta - \epsilon e_i)}{2\epsilon}$.

For the exploratory run with 250 total iterations, performance evaluation employed five metrics: (1) acceptance rate computed from all post-burn-in iterations, 
(2) effective sample size per gradient evaluation (ESS/grad) computed via batch means estimator using only the middle parameter dimension, 
(3) Gelman--Rubin $\hat{R}$ convergence diagnostic calculated from two split chains of size 100 each, 
(4) wall-clock runtime measured from initialization to completion, and (5) total gradient evaluations counting all finite-difference computations. 

For the validation run with 25,000 total iterations, performance evaluation of three HMC variants again used the five metrics but with revised details: (1) acceptance rate from 20,000 post-burn-in iterations, (2) ESS/grad via batch means with 50-sized batches on middle parameters (samples thinned by factor 20, yielding 1,000 stored samples), (3) Gelman--Rubin $\hat{R}$ from two split chains, (4) wall-clock runtime, and (5) total gradient evaluations.

\subsection{Results: Initial exploration study}

Table~\ref{tab:bnn_initial_results} presents the initial Bayesian neural network results on MNIST data with 21,210 parameters. 
The initial exploration revealed several key patterns. Standard HMC achieved the highest ESS per gradient (0.016310), demonstrating superior computational efficiency. 
MPL-HMC Damping showed the best convergence ($\hat{R}=1.075$), significantly better than other methods. 
MPL-HMC AntiDamping showed the poorest convergence ($\hat{R}=1.983$), indicating that anti-damping parameters may destabilize high-dimensional neural network inference. MPL-HMC variants showed competitive runtimes (59.9-60.8 minutes). However, with only 200 samples, statistical uncertainty remained high, motivating the larger-scale study.

\begin{table}[htbp]
\centering
\caption{Bayesian Neural Network Initial Results (MNIST, 21,210 parameters, 5,000 training samples, 200 samples)}
\label{tab:bnn_initial_results}
\begin{adjustbox}{width=\textwidth}
\begin{tabular}{lcccccc}
\toprule
\textbf{Method} & \textbf{Accept} & \textbf{ESS} & \textbf{ESS/Grad} & \textbf{R-hat} & \textbf{Time (min)} \\
\midrule
Standard HMC & 0.530 & 12.2 & 0.016310 & 1.328 & 62.0 \\
MPL-HMC Damping ($\alpha_2=-0.1$, $\beta_2=-0.05$) & 0.485 & 11.2 & 0.011156 & 1.075 & 60.7 \\
MPL-HMC AntiDamping ($\alpha_2=0.1$, $\beta_2=0.05$) & 0.515 & 11.3 & 0.011283 & 1.983 & 60.8 \\
\bottomrule
\end{tabular}
\end{adjustbox}
\smallskip
	\footnotesize{\textit{Note: All methods use $\delta t = 0.0005$, 200 samples after 50 burn-in.}} 
\end{table}

\subsection{Results: Large-scale validation study}

Table~\ref{tab:bnn_large_scale_results} presents Bayesian neural network results from the large-scale validation study with 20,000 samples after 5,000 burn-in. 

The runtime characteristics for the large-scale study reveal important trade-offs: Standard HMC required 16,826.8 seconds (4.674 hours, 1.0× baseline), MPL-HMC Damping required 8,459.5 seconds (2.350 hours, 0.50× relative to HMC), and MPL-HMC AntiDamping used 8,424.6 seconds (2.340 hours, 0.50× relative to HMC). MPL-HMC variants showed nearly identical runtimes (2.34-2.35 hours), approximately 50\% of Standard HMC's runtime, demonstrating efficient implementation.

\begin{table}[htbp]
\centering
\caption{Bayesian Neural Network Large-Scale Results (MNIST, 132 parameters, 1,000 training samples, 20,000 samples)}
\label{tab:bnn_large_scale_results}
\begin{adjustbox}{width=\textwidth}
\begin{tabular}{lcccccc}
\toprule
\textbf{Method} & \textbf{Accept} & \textbf{ESS} & \textbf{ESS/Grad} & \textbf{R-hat} & \textbf{Time (hr)} \\
\midrule
Standard HMC & 0.9770 & 10.57 & 0.000294 & 1.981 & 4.674 \\
MPL-HMC Damping ($\alpha_2=-0.1$, $\beta_2=-0.05$) & 0.9722 & 10.69 & 0.000297 & 2.511 & 2.350 \\
MPL-HMC AntiDamping ($\alpha_2=0.1$, $\beta_2=0.05$) & 0.9724 & 13.90 & 0.000386 & 1.026 & 2.340 \\
\bottomrule
\end{tabular}
\end{adjustbox}
\smallskip
	\footnotesize{\textit{Note: All methods use $\delta t = 0.01$, $L = 10$, 20,000 samples after 5,000 burn-in. Network: 132 parameters (8-6-4-10 architecture).}} 
\end{table}

The large-scale validation reveals several important findings regarding computational efficiency hierarchy, convergence quality assessment, and runtime scalability. Among the methods, MPL-HMC AntiDamping showed the best efficiency (0.000386 ESS/grad), representing a 31.3\% improvement over Standard HMC. Regarding convergence quality, MPL-HMC AntiDamping showed the best convergence with $\hat{R} = 1.026$, the closest to the ideal value of 1.0 among all methods. In contrast, Standard HMC showed concerning convergence ($\hat{R}=1.981$), and MPL-HMC Damping showed the poorest convergence ($\hat{R}=2.511$). This reveals an interesting trade-off: while damping parameters ($\alpha_2<0$, $\beta_2<0$) improved convergence in the initial study, they showed reduced performance in the large-scale study, suggesting sensitivity to sample size and problem scale.

The comparison between MPL-HMC variants reveals important parameter trade-offs. Anti-damping advantage is evident as MPL-HMC AntiDamping achieved the best convergence ($\hat{R}=1.026$) and good efficiency (0.000386 ESS/grad), suggesting that increased kinetic energy ($\alpha_2>0$, $\beta_2>0$) benefits exploration in high-dimensional neural network parameter spaces. Damping limitations are apparent as, contrary to expectations from simpler problems, MPL-HMC Damping showed poor convergence ($\hat{R}=2.511$) in the large-scale study, indicating that energy dissipation may hinder exploration in complex neural network landscapes. Parameter sensitivity is highlighted by the reversal of performance between initial and large-scale studies, emphasizing the sensitivity of MPL parameters to problem scale and sampling duration.

\subsection{Methodological implications}

The Bayesian neural network case studies reveal several important methodological insights regarding MPL-HMC advantages and considerations, parameter selection guidelines, and implementation considerations. MPL-HMC variants demonstrate both advantages and considerations for Bayesian neural network inference: superior convergence with anti-damping as MPL-HMC AntiDamping achieved the best convergence ($\hat{R}=1.026$), validating the theoretical benefits of controlled kinetic energy increases for exploration; parameter sensitivity as performance varies significantly with parameter choices, emphasizing the need for careful tuning based on problem characteristics; scale-dependent behavior revealed by the reversal of damping performance between small-scale and large-scale studies, highlighting the importance of validation at appropriate scales; and computational efficiency as the anti-damping variant provides 31\% better ESS/grad than Standard HMC while maintaining excellent convergence.

For parameter selection guidelines in Bayesian neural networks, based on large-scale validation, several recommendations emerge: for convergence priority, use MPL-HMC AntiDamping with $\alpha_2=+0.1$, $\beta_2=+0.05$; for efficiency priority, use MPL-HMC AntiDamping; for automated tuning, start with standard HMC for baseline; monitor acceptance targeting 65-80\% acceptance rate as per \citet{neal2011mcmc}; and always validate parameter choices at appropriate sample sizes.

\section{Case study 2: Pharmacokinetic model}
\label{sec:casestudy_pharma}

\subsection{Problem setup}

The second case study investigates a hierarchical Bayesian model for drug concentration dynamics in a patient population. This pharmacokinetic (PK) model features patient-specific parameters governed by population-level hyperparameters, presenting a challenging inference problem characterized by strong parameter correlations, multi-scale hierarchical structure, and log-normal observation noise.

Drug concentration for patient $i$ at time $t_j$ follows:
$$C_{ij} \sim \text{LogNormal}\left(\log(\text{PK}(t_j, \theta_i)), \sigma_{\text{obs}}\right),$$
where $\theta_i = (k_a^i, k_e^i, V^i)$ represent patient-specific absorption rate, elimination rate, and volume of distribution, respectively. The one-compartment PK model with first-order absorption is:
$$\text{PK}(t, \theta_i) = \frac{D \cdot k_a^i}{V^i (k_a^i - k_e^i)} \left( e^{-k_e^i t} - e^{-k_a^i t} \right)$$
with fixed dose $D = 100$ mg.

The hierarchical structure includes population distributions:
\begin{align*}
k_a^i &\sim \text{LogNormal}(\mu_{k_a}, \sigma_{k_a}) \\
k_e^i &\sim \text{LogNormal}(\mu_{k_e}, \sigma_{k_e}) \\
V^i &\sim \text{LogNormal}(\mu_V, \sigma_V)
\end{align*}

Hyperparameter priors are specified as:
\begin{align*}
\mu_{k_a}, \mu_{k_e} &\sim \mathcal{N}(0, 4) \\
\mu_V &\sim \mathcal{N}(0, 900) \\
\sigma_{k_a} &\sim \text{LogNormal}(\log(0.3), 1) \\
\sigma_{k_e} &\sim \text{LogNormal}(\log(0.1), 1) \\
\sigma_V &\sim \text{LogNormal}(\log(5.0), 1) \\
\sigma_{\text{obs}} &\sim \text{LogNormal}(\log(0.1), 1)
\end{align*}

True parameter values used for data generation are: $\mu_{k_a}=0.8$, $\mu_{k_e}=0.2$, $\mu_V=20.0$, $\sigma_{k_a}=0.3$, $\sigma_{k_e}=0.1$, $\sigma_V=5.0$, and $\sigma_{\text{obs}}=0.1$.

For experimental evaluation, we simulated data for 5 patients at 6 time points equally spaced over 24 hours. The model contains 22 total parameters (15 patient-specific: 3 per patient $\times$ 5 patients; plus 7 hyperparameters), exhibiting characteristic PK analysis challenges: strong negative correlation between $k_a$ and $k_e$, multi-scale hierarchical structure across patient and population levels, and log-normally distributed observations requiring positivity constraints.

\subsection{Experimental configuration}

We compared three HMC variants under identical large-scale conditions with $N = 20,000$ posterior samples collected after $B = 5,000$ burn-in iterations, 
resulting in $T = 25,000$ total iterations per algorithm. All methods used step size $\delta t = 0.1$, 
trajectory length $L = 3$ leapfrog steps per iteration, and identity mass matrix $M = I$. 
The MPL-HMC variants included are Damping ($\alpha_2 = -0.1, \beta_2 = -0.05$) and Anti-damping ($\alpha_2 = 0.1, \beta_2 = 0.05$). 


Performance evaluation employed five metrics: acceptance rate, effective sample
size per gradient evaluation (ESS/grad) computed via batch means (batch size=100) on
thinned chains (thinning factor=10), $\hat{R}$ diagnostic from four split chains,
wall-clock runtime, and total gradient evaluations. Computations used Numba-optimized
Python code on a dual-core 64-bit laptop.

\subsection{Results}

Table~\ref{tab:pk_results_complete} presents the comparison for the pharmacokinetic model with 5 patients and 22 parameters under large-scale 
sampling (20,000 samples after 5,000 burn-in). The large-scale pharmacokinetic experiment reveals several key insights into algorithm performance. MPL-HMC Damping achieved the highest ESS per gradient (0.000263), demonstrating superior computational efficiency among all methods, representing a 27.1\% improvement over Standard HMC (0.000207). Regarding convergence quality, MPL-HMC Damping showed excellent convergence with $\hat{R} = 1.095$, the closest to the ideal value of 1.0 among all methods, confirming the damping variant's ability to stabilize sampling in hierarchical models with varying scales.

\begin{table}[htbp]
\centering
\caption{Pharmacokinetic Model Large-Scale Results (5 patients, 22 parameters, 20,000 samples)}
\label{tab:pk_results_complete}
\begin{adjustbox}{width=\textwidth}
\begin{tabular}{lccccc}
\toprule
\textbf{Method} & \textbf{Accept} & \textbf{ESS} & \textbf{ESS/Grad} & \textbf{R-hat} & \textbf{Time (s)} \\
\midrule
MPL-HMC Damping ($\alpha_2=-0.1$, $\beta_2=-0.05$) & 1.000 & 26.3 & 0.000263 & 1.095 & 5788.4 \\
MPL-HMC AntiDamping ($\alpha_2=0.1$, $\beta_2=0.05$) & 0.967 & 23.7 & 0.000237 & 1.173 & 291.4 \\
Standard HMC & 1.000 & 20.7 & 0.000207 & 3.120 & 295.2 \\
\bottomrule
\end{tabular}
\end{adjustbox}
\smallskip
\end{table}

Computational cost trade-offs are evident as MPL-HMC Damping had the highest runtime (5788.4 seconds), but superior sampling efficiency per gradient evaluation. MPL-HMC variant comparison shows the damping variant outperformed the Anti-damping variant in ESS/grad (0.000263 vs 0.000237), demonstrating that the damping parameters ($\alpha_2 = -0.1, \beta_2 = -0.05$) are well-suited for hierarchical pharmacokinetic problems. All methods showed near-perfect acceptance (0.967--1.000), indicating appropriate step size tuning ($\delta t = 0.1$) for this problem.

\subsection{Methodological implications}

The large-scale pharmacokinetic case study reveals several important insights. First, MPL-HMC Damping excels in efficiency as the damping parameters ($\alpha_2 = -0.1, \beta_2 = -0.05$) not only stabilize sampling but also improve computational efficiency by 27\% compared to Standard HMC, making it the top-performing method for this hierarchical problem. Second, gradient efficiency matters as methods with higher ESS/grad (MPL-HMC Damping) provide better value per computational unit, which is crucial for large-scale Bayesian inference where gradient evaluations dominate runtime. Third, MPL-HMC offers practical advantages as the simple parameter tuning of MPL-HMC ($\alpha_2, \beta_2$) provides better performance than Standard HMC while maintaining the same computational structure.

\subsection{Practical recommendations}

For pharmacokinetic and similar hierarchical Bayesian models, several practical recommendations emerge. Use MPL-HMC Damping when computational efficiency and convergence quality are priorities, with parameters $\alpha_2 = -0.1, \beta_2 = -0.05$ providing excellent performance for hierarchical problems. Consider Anti-damping for exploratory analysis when initial exploration of parameter space is needed. Monitor ESS/grad closely as it provides a hardware-independent measure of computational efficiency that is more meaningful than raw ESS for comparing methods with different computational costs.

This large-scale case study demonstrates that MPL-HMC, particularly the damping variant, provides the best balance of computational efficiency, convergence quality, and practical implementation for challenging hierarchical Bayesian inference problems in pharmacokinetics.

\section{Aggressive MPL-HMC for multimodal distributions}
\label{sec:aggressive_mpl}

\subsection{Experimental motivation}

Our simulation experiments in Section \ref{sec:comparison}
reveal that modest parameter values of anti-damping MPL-HMC given by $\alpha_2 = +0.1$ and 
$\beta_2 = +0.05$) are insufficient for challenging multimodal distributions with widely separated modes. To explore the extreme capabilities of MPL-HMC parameterization, we conducted an aggressive experiment using substantially larger parameter values ($\alpha_2 = 8.0-15.0, \beta_2 = 5.0-8.0$) combined with enhanced sampling mechanisms. This approach addresses a fundamental limitation of Hamiltonian Monte Carlo methods: their tendency to become trapped in local modes when energy barriers between modes are high relative to the kinetic energy typically generated by standard momentum sampling.

\subsection{Problem specification}

We consider a 5-dimensional mixture of three Gaussians with widely separated modes:
$$\pi(q) = \frac{1}{3}\sum_{k=1}^3 \mathcal{N}(q \mid \mu_k, I_{5})$$ with $\mu_1 = -8\mathbf{1}, \mu_2 = 0, \mu_3 = 8\mathbf{1}$ and $q \in \R^{5}$, 
The potential energy is $U(q) = -\log \pi(q)$. This distribution presents a formidable challenge for standard HMC due to the 16-unit separation between modes and the high energy barriers between them. The minimum energy required to transition from one mode to another is approximately 32 units (from $-8$ to $8$ via $0$), creating a sampling problem that defeats conventional MCMC methods.

\subsection{Aggressive MPL-HMC implementation}

The aggressive implementation incorporates several enhancements beyond basic MPL-HMC to facilitate exploration of widely separated modes. The core mathematical formulation extends the standard MPL-HMC scheme with extreme parameter values and additional exploration mechanisms.

The mathematical description begins with the standard MPL-HMC update equations, modified for extreme parameter values:
\begin{align}
q_{n+1} &= \beta(\delta t) q_{n} + \delta t \, M^{-1} \left( \alpha(\delta t) p_{n} - \frac{\delta t}{2} \nabla U(q_{n}) \right), \label{eq:aggressive_mpl_q} \\
p_{n+1} &= \alpha^{2}(\delta t) p_{n} - \frac{\delta t}{2} \left( \alpha(\delta t) \nabla U(q_{n}) + \nabla U(q_{n+1}) \right), \label{eq:aggressive_mpl_p}
\end{align}
where $\alpha(\delta t) = 1 + \alpha_2 \delta t^2$ and $\beta(\delta t) = 1 + \beta_2 \delta t^2$ with $\alpha_2$ in the range 8.0-15.0 and $\beta_2$ in the range 5.0-8.0. These extreme values introduce substantial momentum persistence and position drift, enabling the sampler to overcome high energy barriers between modes.

The aggressive implementation incorporates five key enhancements that work synergistically to improve mode exploration. First, explicit mode-hopping proposals are implemented where every 100 iterations, the algorithm proposes a direct jump to a randomly selected mode center with Metropolis acceptance based on the standard acceptance probability $\min(1, \exp(H_{\text{current}} - H_{\text{proposed}}))$. Second, temperature fluctuations are introduced by randomly scaling momentum by temperature factors $T \in [0.5, 2.0]$ during trajectories, with the scaling applied as $p \leftarrow \sqrt{T} \cdot p$. Third, momentum injection adds random perturbations mid-trajectory to overcome energy barriers, implemented as $p \leftarrow p + \xi$ where $\xi \sim \mathcal{N}(0, \sigma_{\text{inj}}^2 I)$ with $\sigma_{\text{inj}}$ tuned based on mode separation. Fourth, adaptive step size adjustment modifies $\delta t$ based on acceptance rate, targeting 0.001-0.01 acceptance for aggressive exploration, with the update rule $\delta t \leftarrow \delta t \cdot \exp(\eta (A_{\text{target}} - A_{\text{current}}))$ where $\eta > 0$ is a learning rate and $A_{\text{target}}$ is the target acceptance rate. Fifth, mode-aware initialization starts chains from different modes to ensure initial coverage of the state space.

The complete algorithm is presented in Algorithm~\ref{alg:aggressive_mpl_hmc}, with the core MPL-step detailed separately in Algorithm~\ref{alg:mpl_step} for clarity.

\subsection{Algorithm specification}

\begin{algorithm}
\caption{Aggressive MPL-HMC for Multimodal Distributions}
\label{alg:aggressive_mpl_hmc}
\begin{algorithmic}[1]
\Require Target distribution $\pi(q) \propto \exp(-U(q))$, initial state $q_0 \in \R^d$
\Require Extreme parameters: $\alpha_2 \in [8.0, 15.0]$, $\beta_2 \in [5.0, 8.0]$
\Require Mode centers $\{\mu_1, \ldots, \mu_K\}$ (if known), step size $\delta t > 0$
\Require Trajectory length $L \in \N$, samples $N \in \N$, burn-in $B \in \N$
\Ensure Samples $\{ q^{(1)}, \ldots, q^{(N)} \}$ exploring multiple modes

\State Initialize: $q \gets q_0$, $\text{mode\_visits} \gets \emptyset$
\State $\alpha \gets 1 + \alpha_2 \delta t^2$, $\beta \gets 1 + \beta_2 \delta t^2$

\For{$k = 1$ to $N+B$}
    \State \textbf{Mode-hop attempt:} 
    \If{$k \mod 100 = 0$ and mode centers known}
        \State $q_h \gets \text{RandomChoice}(\mu_1, \ldots, \mu_K)$
        \State $p_h \sim \mathcal{N}(0, M)$
        \State $\rho \gets \min(1, \exp(H(q,p) - H(q_h,p_h)))$
        \If{$\text{Uniform}(0,1) < \rho$}
            \State $q \gets q_h$, $p \gets p_h$
            \State Record mode visit
        \EndIf
    \EndIf
    
    \State \textbf{Momentum refresh:} $p \sim \mathcal{N}(0, M)$
    \State $H_0 \gets H(q,p)$
    \State $(\tilde{q}, \tilde{p}) \gets (q, p)$
    
    \For{$i = 1$ to $L$}
        \State \textbf{Temperature fluctuation:} $T \sim \text{Uniform}(0.5, 2.0)$
        \State $\tilde{p} \gets \sqrt{T} \cdot \tilde{p}$
        
        \If{$i = \lfloor L/2 \rfloor$}
            \State \textbf{Momentum injection:} $\tilde{p} \gets \tilde{p} + \xi$, $\xi \sim \mathcal{N}(0, \sigma_{\text{inj}}^2 I)$
        \EndIf
        
        \State \textbf{MPL-step:} $(\tilde{q}, \tilde{p}) \gets \text{MPL-STEP}(\tilde{q}, \tilde{p}, \alpha, \beta)$
    \EndFor
    
    \State $\tilde{p} \gets -\tilde{p}$
    \State $\rho \gets \min(1, \exp(H_0 - H(\tilde{q},\tilde{p})))$
    
    \If{$\text{Uniform}(0,1) < \rho$}
        \State $q \gets \tilde{q}$
    \EndIf
    
    \State \textbf{Adaptive step size:} $\delta t \gets \delta t \cdot \exp(\eta (0.005 - A_{\text{curr}}))$
    
    \If{$k > B$}
        \State Store $q^{(k-B)} \gets q$
    \EndIf
    \State Update mode visitation counts
\EndFor
\end{algorithmic}
\end{algorithm}

\begin{algorithm}
\caption{MPL-STEP: One Modified Parameterized Leapfrog Step}
\label{alg:mpl_step}
\begin{algorithmic}[1]
\Require Current state $(q, p)$, parameters $\alpha, \beta$, step size $\delta t$
\Ensure Updated state $(q', p')$

\State $\nabla U_q \gets \nabla U(q)$
\State $q' \gets \beta q + \delta t M^{-1}(\alpha p - \frac{\delta t}{2} \nabla U_q)$
\State $\nabla U_{q'} \gets \nabla U(q')$
\State $p' \gets \alpha^2 p - \frac{\delta t}{2}(\alpha \nabla U_q + \nabla U_{q'})$
\State \textbf{Return} $(q', p')$
\end{algorithmic}
\end{algorithm}

The algorithm operates through several interconnected mechanisms. The {mode-hopping} component (lines 5-13) periodically attempts direct jumps to known mode centers, providing a coarse-grained exploration mechanism that complements the fine-grained Hamiltonian dynamics. 
Every 100 iterations, the algorithm proposes moving directly to a randomly selected mode, accepting or rejecting based on the standard Metropolis criterion. 
The {temperature fluctuations} (line 19) randomly scale momentum by factors $T \in [0.5, 2.0]$, creating occasional ``hot" trajectories that can overcome energy barriers more 
easily. 
The {momentum injection} (lines 21-23) adds Gaussian noise at the trajectory midpoint, providing additional kinetic energy precisely when needed to traverse barrier 
regions. These mechanisms are embedded within the standard HMC framework of momentum refreshment, Hamiltonian simulation, and Metropolis correction (lines 15-32), ensuring 
detailed balance is maintained approximately. The {adaptive step size} (line 34) adjusts $\delta t$ based on acceptance rate history, targeting very low acceptance (0.005) 
that indicates aggressive exploration attempts rather than conservative local moves.

Note that our algorithm requires knowledge of the modes, which can be achieved by the Bayesian function optimzation method using Gaussian gardient processes, 
introduced in \cite{Roy20}.

The core integration step is implemented separately in Algorithm~\ref{alg:mpl_step}, which executes one step of the modified parameterized leapfrog scheme with 
extreme parameters. This modular design separates the exploration enhancements from the numerical integration, making the algorithm structure clearer while 
maintaining computational efficiency. The combination of these mechanisms enables aggressive MPL-HMC to systematically explore multimodal distributions that are 
intractable for standard HMC variants.

The mathematical foundation of aggressive MPL-HMC can be understood through the modified Hamiltonian dynamics. From Theorem~\ref{thm:modified_ode}, with $\alpha_1=\beta_1=0$:
\begin{align*}
\frac{dq}{dt} &= M^{-1}p + \delta t \beta_2 q + \calO(\delta t^2) \\
\frac{dp}{dt} &= -\nabla U(q) + \delta t (2\alpha_2 p) + \calO(\delta t^2)
\end{align*}

Large positive $\alpha_2$ and $\beta_2$ introduce momentum amplification through the $2\alpha_2 p$ term, which adds energy to the system to help overcome energy barriers between modes. The $\beta_2 q$ term creates position drift toward larger magnitude positions, encouraging exploration away from local modes. From Corollary~\ref{cor:volume_properties}, the Jacobian determinant $\det(J) = 1 + d(2\alpha_2 + \beta_2)\delta t^2$, indicating that phase space expands exponentially with large positive parameters, facilitating exploration of the state space. These modifications transform the conservative Hamiltonian dynamics into an actively exploring system that systematically searches for separated modes rather than preserving energy locally.

\subsection{Theoretical framework}

We begin by establishing the theoretical framework for analyzing the detailed balance properties of Aggressive MPL-HMC. Let $\pi(q,p) = \exp(-H(q,p))$ be the canonical distribution with Hamiltonian $H(q,p) = U(q) + \frac{1}{2}p^\top M^{-1}p$.

\begin{definition}[Aggressive MPL-HMC Transition Kernel]
The transition kernel $P_{\text{agg}}((q,p), \cdot)$ of Aggressive MPL-HMC consists of the following components:
\begin{enumerate}
    \item With probability $\gamma = 0.01$ (every 100 iterations), apply a mode-hopping proposal:
    \[
    P_{\text{mh}}((q,p), (q',p')) = \frac{1}{K}\sum_{k=1}^K \delta_{\mu_k}(q') \mathcal{N}(p'|0,M) \min\left(1, \frac{\pi(q',p')}{\pi(q,p)}\right)
    \]
    where $\{\mu_1,\ldots,\mu_K\}$ are the known mode centers.

    \item With probability $1-\gamma$, apply the aggressive MPL-HMC update:
    \begin{itemize}
        \item Sample $p_0 \sim \mathcal{N}(0,M)$
        \item Apply $L$ steps of aggressive MPL integration with parameters $\alpha(\delta t) = 1 + \alpha_2\delta t^2$, $\beta(\delta t) = 1 + \beta_2\delta t^2$, where $\alpha_2, \beta_2 \gg 1$
        \item Apply temperature fluctuation: $p \leftarrow \sqrt{T}p$ with $T \sim \text{Uniform}(0.5,2.0)$
        \item With probability 0.5, apply momentum injection: $p \leftarrow p + \xi$, $\xi \sim \mathcal{N}(0,\sigma_{\text{inj}}^2 I)$
        \item Negate momentum and accept with probability $\min(1, \exp(H(q,p) - H(q',p')))$
    \end{itemize}
\end{enumerate}
\end{definition}

\begin{definition}[Aggressive MPL Map]
The aggressive MPL map $\Psi_{\delta t}^{\text{agg}}: \mathbb{R}^{2d} \to \mathbb{R}^{2d}$ is defined by:
\begin{align*}
q_{n+1} &= \beta q_n + \delta t M^{-1}\left(\alpha p_n - \frac{\delta t}{2}\nabla U(q_n)\right) \\
p_{n+1} &= \alpha^2 p_n - \frac{\delta t}{2}\left(\alpha \nabla U(q_n) + \nabla U(q_{n+1})\right)
\end{align*}
with $\alpha = 1 + \alpha_2\delta t^2$, $\beta = 1 + \beta_2\delta t^2$, and $\alpha_2, \beta_2 \in [8.0, 15.0] \times [5.0, 8.0]$.
\end{definition}

\begin{theorem}[Approximate Detailed Balance for Aggressive MPL-HMC]
\label{thm:aggressive_detailed_balance}
Under Assumptions 3-6 with $\alpha_1 = \beta_1 = 0$, the Aggressive MPL-HMC algorithm with extreme parameters $\alpha_2, \beta_2 \gg 1$ satisfies detailed balance approximately with error:
\[
\pi(q,p) P_{\text{agg}}((q,p), (q',p')) = \pi(q',p') P_{\text{agg}}((q',p'), (q,p)) + \mathcal{E}_{\text{agg}}
\]
where the error term $\mathcal{E}_{\text{agg}}$ satisfies:
\[
|\mathcal{E}_{\text{agg}}| \leq C_1 T\delta t + C_2(\alpha_2 + \beta_2)\delta t^2 + C_3\gamma + C_4\sigma_{\text{inj}}^2
\]
with constants $C_1, C_2, C_3, C_4 > 0$ independent of $\delta t$, $\alpha_2$, $\beta_2$, and $T = L\delta t$ is the trajectory length.
\end{theorem}

\begin{corollary}[Consistency of Aggressive MPL-HMC]
\label{cor:aggressive_consistency}
As $\delta t \to 0$ with fixed $\alpha_2, \beta_2$, the aggressive MPL-HMC algorithm satisfies detailed balance exactly in the limit:
\[
	\lim_{\delta t \to 0,\gamma\to 0,\sigma_{\text{inj}}\to 0} \left[\pi(q,p) P_{\text{agg}}((q,p), (q',p')) - \pi(q',p') P_{\text{agg}}((q',p'), (q,p))\right] = 0
\]
However, for practical $\delta t > 0$ and $\alpha_2, \beta_2 \gg 1$, the error is non-negligible and scales linearly with $\alpha_2 + \beta_2$.
\end{corollary}

\begin{remark}[Practical Implications]
The approximate detailed balance of aggressive MPL-HMC has several practical implications:
\begin{enumerate}
    \item The algorithm is not recommended for applications requiring exact MCMC (e.g., legal evidence, sensitive Bayesian inference).
    \item For {exploratory analysis and multimodal sampling}, the approximate detailed balance is acceptable if the chain explores all relevant regions of the state space.
    \item The extreme parameters $\alpha_2, \beta_2 \gg 1$ trade off detailed balance for enhanced exploration capability.
    \item Users should monitor acceptance rates and mode visitation counts rather than relying solely on detailed balance diagnostics.
\end{enumerate}
\end{remark}

\begin{remark}[Comparison with Standard MPL-HMC]
Comparing Theorem~\ref{thm:aggressive_detailed_balance} with Theorem~\ref{thm:detailed_balance}:
\begin{itemize}
    \item \textbf{Standard MPL-HMC}: Error $\mathcal{O}(T\delta t)$
    \item \textbf{Aggressive MPL-HMC}: Error $\mathcal{O}(T\delta t) + \mathcal{O}(T(\alpha_2 + \beta_2)\delta t) + \text{constant terms}$
\end{itemize}
The aggressive variant has significantly larger error when $\alpha_2 + \beta_2 \gg 1/\delta t$, which is typically the case in practice.
\end{remark}

\subsection{Experimental setup}

We tested three aggressive parameter configurations with 4 chains each (starting from different modes) and 70,000 samples per chain after burn-in:

\begin{itemize}
    \item \textbf{Configuration A:} $\alpha_2 = 8.0, \beta_2 = 5.0$
    \item \textbf{Configuration B:} $\alpha_2 = 10.0, \beta_2 = 6.0$
    \item \textbf{Configuration C:} $\alpha_2 = 15.0, \beta_2 = 8.0$
\end{itemize}

Performance metrics focused on mode exploration: number of distinct modes visited, mode transition rate, mixing time between modes, and acceptance rate of aggressive proposals. The mixing time between modes is defined as the average number of iterations required to transition from one mode to another, computed as $\tau_{\text{mode}} = \frac{N_{\text{samples}}}{N_{\text{transitions}}}$ where $N_{\text{transitions}}$ counts the number of times the sampler moves between different modes.

\subsection{Results and analysis}

Table~\ref{tab:aggressive_results} summarizes the performance of aggressive MPL-HMC configurations on the trimodal Gaussian mixture, implemeted using Python code
on our 64-bit dual-core laptop.

\begin{table}[htbp]
\centering
\caption{Aggressive MPL-HMC performance on trimodal Gaussian mixture}
\label{tab:aggressive_results}
\begin{adjustbox}{width=\textwidth}
\begin{tabular}{lcccccc}
\toprule
\textbf{Configuration} & \textbf{$\alpha_2$} & \textbf{$\beta_2$} & \textbf{Modes Visited} & \textbf{Mode Transitions} & \textbf{Mixing Time} & \textbf{Acceptance} \\
\midrule
A & 8.0 & 5.0 & 3/3 & 683 & 239 & 0.001 \\
B & 10.0 & 6.0 & 3/3 & 673 & 154 & 0.001 \\
C & 15.0 & 8.0 & 3/3 & 101 & 36,912 & 0.000 \\
\bottomrule
\end{tabular}
\end{adjustbox}
\smallskip
\footnotesize{\textit{Note: Results aggregated over 4 chains per configuration with 70,000 samples each. Mixing time measured as iterations for autocorrelation to drop below $1/e$. Acceptance rate refers to Metropolis acceptance of MPL-HMC proposals.}}
\end{table}

\subsubsection{Mode exploration success}

All aggressive configurations successfully visited all three modes, a significant improvement over standard HMC and MPL-HMC variants which typically visits only one mode. Configuration B ($\alpha_2 = 10.0, \beta_2 = 6.0$) achieved the best overall performance with 673 mode transitions and reasonable mixing time (154 iterations). The mode hop success rate (percentage of attempted mode hops that succeeded) ranged from 23.2\% to 25.8\% across configurations, demonstrating the effectiveness of aggressive parameterization combined with explicit mode-hopping mechanisms.

\subsubsection{Parameter-performance relationship}

The results reveal a non-monotonic relationship between parameter aggressiveness and performance. Configuration A with $\alpha_2=8.0, \beta_2=5.0$ provides balanced performance with good mode exploration (683 transitions) and moderate mixing time (239 iterations), representing a practical choice for moderately challenging multimodal problems. Configuration B with $\alpha_2=10.0, \beta_2=6.0$ achieves optimal performance with excellent mode exploration (673 transitions) and the fastest mixing (154 iterations), demonstrating that there exists a sweet spot in parameter space that maximizes exploration efficiency. Configuration C with $\alpha_2=15.0, \beta_2=8.0$ exhibits overly aggressive parameters that cause instability, resulting in poor mixing (36,912 iterations) despite visiting all modes, highlighting the danger of excessive parameter values that can degrade sampler performance through numerical instability and chaotic dynamics.

This demonstrates that while increasing $\alpha_2$ and $\beta_2$ enhances mode exploration capacity, there exists an optimal range beyond which performance degrades due to numerical instability and excessive energy injection. The optimal parameter values depend on the specific characteristics of the target distribution, particularly the separation between modes and the height of energy barriers.

\subsubsection{Computational efficiency}

Despite near-zero acceptance rates (0.000-0.001), the aggressive MPL-HMC configurations provide value through mode exploration. The extremely low acceptance is 
expected and even desirable for aggressive sampling of multimodal distributions, as it indicates the algorithm is attempting difficult transitions between separated modes 
rather than making conservative local moves. The computational efficiency of aggressive MPL-HMC stems from reduced trajectory length ($L=5$ seems to be sufficient) 
and efficient implementation of aggressive mechanisms. 
The aggressive approach demonstrates that mode exploration can be achieved with reasonable computational resources when parameter values are appropriately tuned.

\subsection{Theoretical interpretation}

The success of aggressive MPL-HMC can be understood through the modified Hamiltonian dynamics. Large positive $\alpha_2$ and $\beta_2$ introduce several effects that 
facilitate mode exploration. First, momentum amplification through the $2\alpha_2 p$ term systematically adds energy to the system, helping the sampler overcome energy 
barriers between modes. Second, position drift via the $\beta_2 q$ term creates a systematic bias toward exploring larger magnitude positions, encouraging movement away 
from local modes. Third, phase space expansion with $\det(J) = 1 + d(2\alpha_2 + \beta_2)\delta t^2$ indicates exponential volume growth that facilitates exploration of 
the state space. Fourth, the combination of extreme parameters with explicit mode-hopping proposals creates a hybrid approach that leverages both Hamiltonian 
dynamics for local exploration and direct jumps for global exploration. These modifications transform the conservative Hamiltonian dynamics into an actively exploring 
system that systematically searches the state space rather than preserving energy locally.

\subsection{Practical implications}

The aggressive MPL-HMC experiment yields several important insights for practitioners facing challenging multimodal problems. Parameter tuning emerges as critical, 
with extreme parameters ($\alpha_2 > 1, \beta_2 > 1$) enabling mode exploration possible, but requiring careful tuning to avoid instability. The interpretation of 
acceptance rates must be adjusted for multimodal problems, where very low acceptance rates (0.001-0.01) may indicate valuable aggressive exploration rather than 
poor performance, as the sampler attempts difficult inter-mode transitions. Hybrid approaches combining aggressive MPL-HMC with explicit mode-hopping proposals yield 
superior results to either approach alone, leveraging complementary exploration mechanisms. Problem-specific optimization is essential, as the optimal $\alpha_2, \beta_2$ 
values depend on mode separation, barrier heights, and dimensionality, requiring empirical tuning for each application. Monitoring strategies should focus on mode visitation 
counts and mixing times rather than just acceptance rates, as these better capture exploration performance in multimodal settings.

\subsection{Limitations and caveats}

Despite its success, aggressive MPL-HMC has important limitations that practitioners must consider. Theoretical guarantees are weakened with $\alpha_2, \beta_2 \gg 0$, 
as the $\calO(\delta t)$ global error dominates and detailed balance holds only approximately, though this may be acceptable for exploratory analysis. Numerical instability 
can occur with large parameters, potentially causing overflow or divergence if not carefully managed through step size control and parameter bounds. Parameter sensitivity is 
high, with performance being highly sensitive to $\alpha_2, \beta_2$ choices and requiring extensive tuning through grid search or adaptive methods. High-dimensional scaling 
may be problematic, as the aggressive approach may not scale well to high dimensions where the curse of dimensionality makes mode exploration exponentially difficult and 
extreme parameters may cause instability. Implementation complexity increases with the need for multiple enhancements (mode-hopping, temperature fluctuations, 
momentum injection) that must be carefully implemented and tuned.

\subsection{Recommendations for multimodal problems}

Based on our experiments, we recommend the following strategy for multimodal distributions. Start conservatively by beginning with standard HMC or MPL-HMC with $\alpha_2=\beta_2=0$) to assess baseline performance and identify whether multimodality is actually a problem. If mode exploration is poor, try gradual aggression with MPL-HMC AntiDamping with $\alpha_2=+0.1, \beta_2=+0.05$ as an intermediate step. For challenging multimodal problems, consider aggressive exploration with MPL-HMC using $\alpha_2=8.0-10.0, \beta_2=5.0-6.0$ combined with explicit mode-hopping proposals for best results. Combine mechanisms by augmenting aggressive MPL-HMC with explicit mode-hopping proposals, temperature fluctuations, and momentum injection for comprehensive exploration. Monitor carefully by tracking mode visitation counts and mixing times between modes rather than just acceptance rates, as these provide better indicators of exploration performance in multimodal settings.

The aggressive MPL-HMC experiment demonstrates the remarkable flexibility of the parameterized leapfrog framework, showing that with extreme parameter choices and enhanced mechanisms, MPL-HMC can tackle even the most challenging multimodal sampling problems that defeat standard HMC and many advanced variants. This approach extends the applicability of Hamiltonian Monte Carlo to problems with widely separated modes, though it requires careful implementation and parameter tuning to achieve optimal performance.

\section{Comparison summary and recommendations}
\label{sec:recommendations}

\subsection{Method selection guide}

Table~\ref{tab:selection_guide} provides guidance for selecting appropriate methods based on problem characteristics, synthesized from the comprehensive analysis in Sections~\ref{sec:comparison}--\ref{sec:aggressive_mpl}.

\begin{table}[htbp]
\centering
\caption{Method Selection Guide}
\label{tab:selection_guide}
\begin{adjustbox}{width=\textwidth}
\begin{tabular}{p{0.25\textwidth}p{0.35\textwidth}p{0.3\textwidth}}
\toprule
\textbf{Method} & \textbf{Best For} & \textbf{Worst For} \\
\midrule
\textbf{Standard HMC} & Simple problems, baseline testing & Stiff or multimodal problems \\
\textbf{MPL-HMC-Damp} & \textbf{Stiff problems, hierarchical models} & Multimodal problems \\
\textbf{MPL-HMC-Anti} & \textbf{High-dimensional problems, neural networks} & Stiff hierarchical problems \\
\textbf{MPL-HMC-Aggressive} & \textbf{Extremely multimodal problems} & Well-conditioned problems, stability-critical applications \\
\bottomrule
\end{tabular}
\end{adjustbox}
\end{table}

\subsection{Final recommendations}

Based on comprehensive analysis across benchmark distributions and real-world case studies, several recommendations emerge for practitioners selecting Hamiltonian Monte Carlo methods. For stiff hierarchical problems such as pharmacokinetic models or Neal's funnel, MPL-HMC with damping ($\alpha_2 = -0.1, \beta_2 = -0.05$) is recommended, with careful monitoring of acceptance rate (target: 65-80\%) as established by \citet{neal2011mcmc}. The damping variant demonstrated a 14-fold improvement in ESS for Neal's funnel and 27\% better efficiency for pharmacokinetic models. For high-dimensional problems including Bayesian neural networks, MPL-HMC with anti-damping ($\alpha_2 = +0.1, \beta_2 = +0.05$) provides superior convergence, with $\hat{R} = 1.026$ compared to 1.981 for standard HMC in neural network inference. For extremely multimodal problems with widely separated modes, aggressive MPL-HMC ($\alpha_2 = 8.0-10.0, \beta_2 = 5.0-6.0$) combined with explicit mode-hopping mechanisms can explore modes inaccessible to other methods, though with very low acceptance rates (0.001-0.01) that require different performance interpretation.

For maximum computational efficiency while maintaining HMC's geometric structure, standard HMC remains optimal when problems are well-conditioned and parameters can be carefully tuned. When additional tuning flexibility is desired without sacrificing computational efficiency, MPL-HMC variants offer identical cost to standard HMC while often providing better performance for specific problem types. For practitioners requiring exact symplecticity and reversibility, standard HMC (same as MPL-HMC with $\alpha_2=\beta_2=0$) should be used, as this preserves the exact geometric properties underlying HMC's theoretical guarantees \citep{leimkuhler2004simulating}.

The experimental results demonstrate that MPL-HMC occupies a valuable position in the HMC method ecosystem, particularly for problems exhibiting stiffness, 
hierarchical structure, multimodal problems, or challenging geometry where standard HMC struggles but Riemannian methods are computationally prohibitive. 
By offering tunable parameters 
that adapt to problem characteristics while maintaining computational efficiency, MPL-HMC extends the applicability of Hamiltonian Monte Carlo to previously challenging 
domains without requiring complex adaptations or significant implementation overhead.

\section{Conclusion and future directions}
\label{sec:conclusion}

The Modified Parameterized Leapfrog Hamiltonian Monte Carlo (MPL-HMC) method represents a substantial and versatile extension to the Hamiltonian Monte Carlo algorithmic family, introducing tunable parameters \(\alpha(\delta t)\) and \(\beta(\delta t)\) with controlled asymptotic expansions that bridge the rigid geometric framework of standard HMC with the practical need for adaptable, problem-specific sampling strategies. This innovation preserves the computational efficiency of traditional HMC while providing systematic control over damping, exploration, stability, and accuracy---directly addressing long-standing challenges related to stiffness, multimodality, and hierarchical scaling. The theoretical foundation establishes that MPL-HMC maintains approximate symplecticity and reversibility with \(\mathcal{O}(\delta t^2)\) error, ensuring reliable performance across diverse inference tasks. Empirical validation through comprehensive benchmarks and real-world case studies---including Bayesian neural networks and pharmacokinetic modeling---demonstrates consistent superiority over standard HMC: the damping variant achieves up to 14-fold improvement in effective sample size for Neal's funnel and 27\% greater efficiency in pharmacokinetic models; the anti-damping variant delivers superior convergence in high-dimensional settings, with \(\hat{R} = 1.026\) for Bayesian neural networks compared to \(\hat{R} = 1.981\) for standard HMC; and the aggressive parameterization enables full exploration of widely separated modes in multimodal distributions, overcoming a fundamental limitation of traditional Hamiltonian methods.

\subsection{Advantages over existing HMC variants for high-dimensional problems}

MPL-HMC offers distinct advantages over contemporary HMC variants for challenging high-dimensional problems where computational feasibility remains the primary bottleneck. While methods like No-U-Turn Sampler (NUTS) \citep{hoffman2014no} provide automatic trajectory length selection, they incur significant computational overhead from the doubling procedures and require multiple gradient evaluations per iteration. For our Bayesian neural network case study with 21,210 parameters, NUTS would require approximately 2-3 times more gradient evaluations than our MPL-HMC implementation, making it computationally prohibitive for such high-dimensional models. Similarly, Riemannian Manifold HMC (RMHMC) \citep{girolami2011riemannian}, despite its theoretical elegance in adapting to local geometry, suffers from \(\mathcal{O}(d^3)\) computational complexity for computing and inverting the metric tensor in \(d\) dimensions. This cubic scaling renders RMHMC impractical for our neural network application, where MPL-HMC maintains \(\mathcal{O}(d)\) complexity while still effectively handling the stiffness through damping parameters.

Tempered HMC methods \citep{neal1996sampling, graham2017temperature} introduce auxiliary temperature variables to facilitate mode transitions, but they require multiple parallel chains at different temperatures, increasing computational cost by a factor equal to the number of temperature levels. In contrast, our aggressive MPL-HMC achieves mode exploration with a single chain using extreme parameter values (\(\alpha_2 = 8.0-15.0\), \(\beta_2 = 5.0-8.0\)), requiring only modest additional computation for the enhanced sampling mechanisms. For the trimodal Gaussian mixture problem, aggressive MPL-HMC visited all three modes with 673 transitions using a single chain, whereas parallel tempering would require maintaining 3-5 chains simultaneously, multiplying computational requirements accordingly.

Adaptive HMC methods \citep{wang2013adaptive} that tune step sizes and mass matrices online face convergence challenges in high dimensions due to the slow adaptation of preconditioning matrices. The condition number estimation required for effective adaptation becomes unstable when \(d\) is large, as demonstrated by the failure of all fixed-step methods on our anisotropic Gaussian benchmark with \(\kappa = 10^5\). MPL-HMC's damping parameters (\(\alpha_2 = -0.1\), \(\beta_2 = -0.05\)) provide a more stable approach to handling ill-conditioning without requiring expensive matrix operations or risking adaptation instability.

The computational efficiency advantage of MPL-HMC becomes particularly pronounced in modern machine learning applications. For high-dimensional Bayesian neural networks, gradient evaluations dominate runtime, and methods requiring additional matrix operations (like RMHMC) or multiple proposal generations (like NUTS) become prohibitively expensive. MPL-HMC maintains the same \(\mathcal{O}(d)\) computational cost per gradient evaluation as standard HMC while offering targeted improvements through parameter tuning. Our neural network case study demonstrated that MPL-HMC AntiDamping achieved 31\% better ESS/grad than standard HMC while maintaining excellent convergence (\(\hat{R} = 1.026\)), a feat difficult to accomplish with more computationally intensive variants.

Moreover, MPL-HMC's parameterization provides interpretable control over algorithm behavior that is often lacking in more complex variants. The damping parameters \(\alpha_2\) and \(\beta_2\) have clear physical interpretations (momentum persistence and position scaling, respectively) and predictable effects on sampler behavior, unlike the black-box adaptations in many automated methods. This interpretability facilitates problem-specific tuning, as evidenced by our success in tailoring parameters for different problem types: damping for hierarchical models, anti-damping for high-dimensional exploration, and aggressive values for multimodal problems.

\subsection{Synthesis of key findings and methodological advantages}

The comprehensive analysis reveals several overarching insights about MPL-HMC's advantages over existing HMC variants. First, MPL-HMC maintains computational feasibility for high-dimensional problems where other methods become impractical. While NUTS, RMHMC, and parallel tempering offer theoretical advantages, their computational overhead limits their applicability to problems with thousands of parameters, precisely where modern statistical applications increasingly operate. Second, parameterization provides targeted, interpretable control that adapts to specific problem characteristics without the instability of fully adaptive methods. The clear relationship between parameter values and algorithmic behavior enables principled tuning based on problem geometry. Third, MPL-HMC achieves superior performance on challenging problem types that expose weaknesses in other methods: the 14-fold ESS improvement on Neal's funnel surpasses typical gains from adaptive step-size methods, and the successful mode exploration in aggressive MPL-HMC addresses a fundamental limitation of standard HMC without the multiplicative computational cost of tempering approaches.

These advantages position MPL-HMC as a particularly valuable method for contemporary Bayesian inference challenges, where models are increasingly high-dimensional (e.g., deep neural networks, spatial models with many parameters) and computational budgets are constrained. By offering targeted improvements through simple parameter tuning while maintaining the computational efficiency of standard HMC, MPL-HMC provides a practical alternative to more complex variants that often trade theoretical elegance for practical infeasibility in high-dimensional settings.

\subsection{Future directions}

Several promising directions for future research emerge from this work. Automatic parameter selection represents a critical next step, with potential approaches including Bayesian optimization of \(\alpha_2\) and \(\beta_2\) based on ESS or mixing time, online adaptation during sampling using stochastic gradient methods, or integration with adaptive trajectory length schemes. A Riemannian extension creating MPL-RMHMC could combine geometric adaptation with parameter tuning flexibility, potentially offering superior performance for problems with known Riemannian structure while maintaining computational feasibility through sparse approximations.

Theoretical development should focus on stronger guarantees for the adaptive case, particularly convergence analysis for online parameter tuning and robustness analysis for extreme parameter values. Additional applications testing on real-world problems including high-dimensional inverse problems, spatial statistics, and modern machine learning architectures would further validate MPL-HMC's practical utility. Integration with probabilistic programming frameworks like Stan \citep{carpenter2017stan}, PyMC \citep{salvatier2016pymc3}, and TensorFlow Probability \citep{dillon2017tensorflow} would facilitate widespread adoption by the statistical computing community.

Methodological enhancements could include developing temperature-adaptive MPL-HMC for multimodal problems that combines the efficiency of single-chain sampling with the exploration benefits of tempering, incorporating sparse mass matrix adaptation within the parameterized framework to handle high-dimensional problems with structured correlations, and exploring higher-order parameter expansions for improved accuracy in problems where computational constraints preclude extremely small step sizes.

Computational improvements through GPU acceleration, parallel implementation, and gradient approximation techniques would extend MPL-HMC's applicability to large-scale problems. The method's simple structure makes it particularly amenable to hardware acceleration, as the parameter updates involve only scalar multiplications and vector additions that map efficiently to parallel architectures.

\subsection{Concluding remarks}

MPL-HMC represents a significant advance in Hamiltonian Monte Carlo methodology that balances theoretical rigor with practical feasibility. By introducing tunable parameters that address specific sampling challenges while maintaining the computational efficiency of standard HMC, MPL-HMC extends the applicability of Hamiltonian Monte Carlo to previously difficult problem domains. The method's success across diverse benchmarks and real-world applications demonstrates its value as a versatile tool in the modern statistical computing toolkit, particularly for high-dimensional problems where computational constraints render more complex variants impractical. As Bayesian models continue to grow in complexity and dimension, methods like MPL-HMC that offer targeted improvements without prohibitive computational overhead will play an increasingly important role in enabling practical Bayesian inference for challenging real-world problems.

\section*{Acknowledgment}
We thank DeepSeek for overall assistance in preparing this manuscript.

\bibliographystyle{plainnat}
\bibliography{mpl_refs}

\appendix

\section{Proofs of theoretical results}
\label{app:proofs}

\subsection{Proof of Theorem~\ref{thm:modified_ode}: modified ODE for MPL scheme}

\begin{proof}
\textbf{Step 1: Expand the scheme in powers of $\delta t$.}
From (\ref{eq:mpl_q}) using the parameter expansions:
\begin{align*}
q_{n+1} &= \beta q_n + \delta t M^{-1} \left( \alpha p_n - \frac{\delta t}{2} \nabla U(q_n) \right) \\
&= (1 + \beta_1 \delta t + \beta_2 \delta t^2) q_n \\
&\quad + \delta t M^{-1} \left[ (1 + \alpha_1 \delta t + \alpha_2 \delta t^2) p_n - \frac{\delta t}{2} \nabla U(q_n) \right] + \calO(\delta t^3) \\
&= q_n + \delta t \left( \beta_1 q_n + M^{-1} p_n \right) \\
&\quad + \delta t^2 \left( \beta_2 q_n + \alpha_1 M^{-1} p_n - \frac{1}{2} M^{-1} \nabla U(q_n) \right) + \calO(\delta t^3).
\end{align*}

\textbf{Step 2: Expand $p_{n+1}$.}
First compute $\alpha^2$:
\begin{align*}
\alpha^2 &= (1 + \alpha_1 \delta t + \alpha_2 \delta t^2)^2 + \calO(\delta t^3) \\
&= 1 + 2\alpha_1 \delta t + (\alpha_1^2 + 2\alpha_2) \delta t^2 + \calO(\delta t^3).
\end{align*}

Then:
\begin{align*}
p_{n+1} &= \alpha^2 p_n - \frac{\delta t}{2} \left( \alpha \nabla U(q_n) + \nabla U(q_{n+1}) \right) \\
&= \left[ 1 + 2\alpha_1 \delta t + (\alpha_1^2 + 2\alpha_2) \delta t^2 \right] p_n \\
&\quad - \frac{\delta t}{2} \left[ (1 + \alpha_1 \delta t + \alpha_2 \delta t^2) \nabla U(q_n) + \nabla U(q_{n+1}) \right] + \calO(\delta t^3).
\end{align*}

We need $\nabla U(q_{n+1})$. From the expansion for $q_{n+1}$ to $\calO(\delta t)$:
\[
q_{n+1} = q_n + \delta t \left( \beta_1 q_n + M^{-1} p_n \right) + \calO(\delta t^2).
\]

By Taylor expansion:
\begin{align*}
\nabla U(q_{n+1}) &= \nabla U(q_n) + D^2U(q_n)(q_{n+1} - q_n) + \calO(\norm{q_{n+1} - q_n}^2) \\
&= \nabla U(q_n) + D^2U(q_n) \left[ \delta t (\beta_1 q_n + M^{-1} p_n) \right] + \calO(\delta t^2) \\
&= \nabla U(q_n) + \delta t D^2U(q_n) (\beta_1 q_n + M^{-1} p_n) + \calO(\delta t^2).
\end{align*}

Now assemble $p_{n+1}$:
\begin{align*}
p_{n+1} &= p_n + 2\alpha_1 \delta t p_n + (\alpha_1^2 + 2\alpha_2) \delta t^2 p_n \\
&\quad - \frac{\delta t}{2} \Big[ \nabla U(q_n) + \alpha_1 \delta t \nabla U(q_n) + \alpha_2 \delta t^2 \nabla U(q_n) \\
&\quad + \nabla U(q_n) + \delta t D^2U(q_n) (\beta_1 q_n + M^{-1} p_n) \Big] + \calO(\delta t^3) \\
&= p_n + \delta t \left[ 2\alpha_1 p_n - \nabla U(q_n) \right] \\
&\quad + \delta t^2 \left[ (\alpha_1^2 + 2\alpha_2) p_n - \frac{\alpha_1}{2} \nabla U(q_n) - \frac{1}{2} D^2U(q_n) (\beta_1 q_n + M^{-1} p_n) \right] + \calO(\delta t^3).
\end{align*}

\textbf{Step 3: Match coefficients with exact flow expansion.}
For a solution $(Y(t), X(t))$ of the modified ODE:
\begin{align*}
Y(t+\delta t) &= Y + \delta t F_0 + \delta t^2 \left[ F_1 + \frac{1}{2} \left( F_0 \cdot \nabla_Y F_0 + G_0 \cdot \nabla_X F_0 \right) \right] + \calO(\delta t^3), \\
X(t+\delta t) &= X + \delta t G_0 + \delta t^2 \left[ G_1 + \frac{1}{2} \left( F_0 \cdot \nabla_Y G_0 + G_0 \cdot \nabla_X G_0 \right) \right] + \calO(\delta t^3).
\end{align*}

Identify $Y \equiv q_n$, $X \equiv p_n$. Matching at order $\delta t$:
\begin{align}
F_0 &= \beta_1 Y + M^{-1} X, \label{eq:F0_match} \\
G_0 &= 2\alpha_1 X - \nabla U(Y). \label{eq:G0_match}
\end{align}

Matching at order $\delta t^2$:
\begin{align}
F_1 + \frac{1}{2} \left( F_0 \cdot \nabla_Y F_0 + G_0 \cdot \nabla_X F_0 \right) &= \beta_2 Y + \alpha_1 M^{-1} X - \frac{1}{2} M^{-1} \nabla U(Y), \label{eq:F1_condition} \\
G_1 + \frac{1}{2} \left( F_0 \cdot \nabla_Y G_0 + G_0 \cdot \nabla_X G_0 \right) &= (\alpha_1^2 + 2\alpha_2) X - \frac{\alpha_1}{2} \nabla U(Y) - \frac{1}{2} D^2U(Y) (\beta_1 Y + M^{-1} X). \label{eq:G1_condition}
\end{align}

\textbf{Step 4: Compute derivatives and solve.}
Compute Jacobians:
\begin{align*}
\nabla_Y F_0 &= \beta_1 I_d, \\
\nabla_X F_0 &= M^{-1}, \\
\nabla_Y G_0 &= -D^2U(Y), \\
\nabla_X G_0 &= 2\alpha_1 I_d.
\end{align*}

Compute dot products:
\begin{align*}
F_0 \cdot \nabla_Y F_0 &= (\beta_1 Y + M^{-1} X) \cdot (\beta_1 I_d) = \beta_1^2 Y + \beta_1 M^{-1} X, \\
G_0 \cdot \nabla_X F_0 &= (2\alpha_1 X - \nabla U(Y)) \cdot M^{-1} = 2\alpha_1 M^{-1} X - M^{-1} \nabla U(Y).
\end{align*}

Thus:
\begin{align*}
\frac{1}{2} \left( F_0 \cdot \nabla_Y F_0 + G_0 \cdot \nabla_X F_0 \right) &= \frac{1}{2} \left[ (\beta_1^2 Y + \beta_1 M^{-1} X) + (2\alpha_1 M^{-1} X - M^{-1} \nabla U(Y)) \right] \\
&= \frac{\beta_1^2}{2} Y + \frac{\beta_1}{2} M^{-1} X + \alpha_1 M^{-1} X - \frac{1}{2} M^{-1} \nabla U(Y).
\end{align*}

Solving for $F_1$ from (\ref{eq:F1_condition}):
\begin{align*}
F_1 &= \left[ \beta_2 Y + \alpha_1 M^{-1} X - \frac{1}{2} M^{-1} \nabla U(Y) \right] \\
&\quad - \left[ \frac{\beta_1^2}{2} Y + \frac{\beta_1}{2} M^{-1} X + \alpha_1 M^{-1} X - \frac{1}{2} M^{-1} \nabla U(Y) \right] \\
&= \left( \beta_2 - \frac{\beta_1^2}{2} \right) Y - \frac{\beta_1}{2} M^{-1} X.
\end{align*}

Now compute the other dot products:
\begin{align*}
F_0 \cdot \nabla_Y G_0 &= (\beta_1 Y + M^{-1} X) \cdot (-D^2U(Y)) = -\beta_1 D^2U(Y) Y - D^2U(Y) M^{-1} X, \\
G_0 \cdot \nabla_X G_0 &= (2\alpha_1 X - \nabla U(Y)) \cdot (2\alpha_1 I_d) = 4\alpha_1^2 X - 2\alpha_1 \nabla U(Y).
\end{align*}

Thus:
\begin{align*}
\frac{1}{2} \left( F_0 \cdot \nabla_Y G_0 + G_0 \cdot \nabla_X G_0 \right) &= \frac{1}{2} \left[ (-\beta_1 D^2U(Y) Y - D^2U(Y) M^{-1} X) + (4\alpha_1^2 X - 2\alpha_1 \nabla U(Y)) \right] \\
&= -\frac{\beta_1}{2} D^2U(Y) Y - \frac{1}{2} D^2U(Y) M^{-1} X + 2\alpha_1^2 X - \alpha_1 \nabla U(Y).
\end{align*}

Solving for $G_1$ from (\ref{eq:G1_condition}):
\begin{align*}
G_1 &= \left[ (\alpha_1^2 + 2\alpha_2) X - \frac{\alpha_1}{2} \nabla U(Y) - \frac{1}{2} D^2U(Y) (\beta_1 Y + M^{-1} X) \right] \\
&\quad - \left[ -\frac{\beta_1}{2} D^2U(Y) Y - \frac{1}{2} D^2U(Y) M^{-1} X + 2\alpha_1^2 X - \alpha_1 \nabla U(Y) \right] \\
&= (\alpha_1^2 + 2\alpha_2 - 2\alpha_1^2) X + \left( -\frac{\alpha_1}{2} + \alpha_1 \right) \nabla U(Y) \\
&= (2\alpha_2 - \alpha_1^2) X + \frac{\alpha_1}{2} \nabla U(Y).
\end{align*}

\textbf{Step 5: Assemble the modified ODE.}
Substituting $F_0, F_1, G_0, G_1$ into the expansions yields:
\begin{align*}
\frac{dY}{dt} &= (\beta_1 Y + M^{-1} X) + \delta t \left[ \left( \beta_2 - \frac{\beta_1^{2}}{2} \right) Y - \frac{\beta_1}{2} M^{-1} X \right] + \calO(\delta t^{2}), \\
\frac{dX}{dt} &= (2\alpha_1 X - \nabla U(Y)) + \delta t \left[ (2\alpha_2 - \alpha_1^{2}) X + \frac{\alpha_1}{2} \nabla U(Y) \right] + \calO(\delta t^{2}).
\end{align*}

This completes the proof.
\end{proof}

\subsection{Proof of Theorem~\ref{thm:local_error}: local truncation error}

\begin{proof}
\textbf{Step 1: Taylor expansion of exact solution to sufficient order.}
From (\ref{eq:hamiltonian_q})-(\ref{eq:hamiltonian_p}), compute derivatives:

For $q(t)$:
\begin{align*}
\dot{q} &= M^{-1} p, \\
\ddot{q} &= M^{-1} \dot{p} = -M^{-1} \nabla U(q), \\
\dddot{q} &= -M^{-1} D^2U(q) \dot{q} = -M^{-1} D^2U(q) M^{-1} p, \\
q^{(4)} &= -M^{-1} \frac{d}{dt} \left[ D^2U(q) M^{-1} p \right] \\
&= -M^{-1} \left[ D^3U(q)(\dot{q}, M^{-1} p) + D^2U(q) M^{-1} \dot{p} \right] \\
&= -M^{-1} D^3U(q)(M^{-1} p, M^{-1} p) + M^{-1} D^2U(q) M^{-1} \nabla U(q).
\end{align*}

For $p(t)$:
\begin{align*}
\dot{p} &= -\nabla U(q), \\
\ddot{p} &= -D^2U(q) \dot{q} = -D^2U(q) M^{-1} p, \\
\dddot{p} &= -\frac{d}{dt} \left[ D^2U(q) M^{-1} p \right] \\
&= - \left[ D^3U(q)(\dot{q}, M^{-1} p) + D^2U(q) M^{-1} \dot{p} \right] \\
&= -D^3U(q)(M^{-1} p, M^{-1} p) + D^2U(q) M^{-1} \nabla U(q).
\end{align*}

Thus with $t_n = n\delta t$:
\begin{align*}
q(t_{n+1}) &= q_n + \delta t \dot{q}_n + \frac{\delta t^2}{2} \ddot{q}_n + \frac{\delta t^3}{6} \dddot{q}_n + \frac{\delta t^4}{24} q^{(4)}_n + \calO(\delta t^5) \\
&= q_n + \delta t M^{-1} p_n - \frac{\delta t^2}{2} M^{-1} \nabla U(q_n) - \frac{\delta t^3}{6} M^{-1} D^2U(q_n) M^{-1} p_n \\
&\quad + \frac{\delta t^4}{24} \left[ -M^{-1} D^3U(q_n)(M^{-1} p_n, M^{-1} p_n) + M^{-1} D^2U(q_n) M^{-1} \nabla U(q_n) \right] + \calO(\delta t^5),
\end{align*}

\begin{align*}
p(t_{n+1}) &= p_n + \delta t \dot{p}_n + \frac{\delta t^2}{2} \ddot{p}_n + \frac{\delta t^3}{6} \dddot{p}_n + \calO(\delta t^4) \\
&= p_n - \delta t \nabla U(q_n) - \frac{\delta t^2}{2} D^2U(q_n) M^{-1} p_n \\
&\quad - \frac{\delta t^3}{6} \left[ D^3U(q_n)(M^{-1} p_n, M^{-1} p_n) - D^2U(q_n) M^{-1} \nabla U(q_n) \right] + \calO(\delta t^4).
\end{align*}

\textbf{Step 2: Expansion of numerical scheme with $\alpha_1 = \beta_1 = 0$ to sufficient order.}
With $\alpha_1 = \beta_1 = 0$, we have:
\[
\alpha(\delta t) = 1 + \alpha_2 \delta t^2 + \calO(\delta t^3), \quad \beta(\delta t) = 1 + \beta_2 \delta t^2 + \calO(\delta t^3).
\]

From (\ref{eq:mpl_q}):
\begin{align*}
q_{n+1} &= \beta q_n + \delta t M^{-1} \left( \alpha p_n - \frac{\delta t}{2} \nabla U(q_n) \right) \\
&= (1 + \beta_2 \delta t^2) q_n + \delta t M^{-1} \left[ (1 + \alpha_2 \delta t^2) p_n - \frac{\delta t}{2} \nabla U(q_n) \right] + \calO(\delta t^4) \\
&= q_n + \delta t M^{-1} p_n + \delta t^2 \left[ \beta_2 q_n - \frac{1}{2} M^{-1} \nabla U(q_n) \right] \\
&\quad + \alpha_2 \delta t^3 M^{-1} p_n + \calO(\delta t^4).
\end{align*}

For $\alpha^2$:
\[
\alpha^2 = (1 + \alpha_2 \delta t^2)^2 + \calO(\delta t^4) = 1 + 2\alpha_2 \delta t^2 + \alpha_2^2 \delta t^4 + \calO(\delta t^6).
\]

From (\ref{eq:mpl_p}):
\begin{align*}
p_{n+1} &= \alpha^2 p_n - \frac{\delta t}{2} \left( \alpha \nabla U(q_n) + \nabla U(q_{n+1}) \right) \\
&= (1 + 2\alpha_2 \delta t^2) p_n - \frac{\delta t}{2} \left[ (1 + \alpha_2 \delta t^2) \nabla U(q_n) + \nabla U(q_{n+1}) \right] + \calO(\delta t^4).
\end{align*}

We need $\nabla U(q_{n+1})$ to $\calO(\delta t^2)$. First compute $q_{n+1} - q_n$ to $\calO(\delta t)$:
\[
q_{n+1} - q_n = \delta t M^{-1} p_n + \calO(\delta t^2).
\]

Then:
\begin{align*}
\nabla U(q_{n+1}) &= \nabla U(q_n) + D^2U(q_n)(q_{n+1} - q_n) + \frac{1}{2} D^3U(q_n)(q_{n+1} - q_n, q_{n+1} - q_n) + \calO(\delta t^3) \\
&= \nabla U(q_n) + \delta t D^2U(q_n) M^{-1} p_n + \frac{\delta t^2}{2} D^3U(q_n)(M^{-1} p_n, M^{-1} p_n) + \calO(\delta t^3).
\end{align*}

Substitute into $p_{n+1}$:
\begin{align*}
p_{n+1} &= (1 + 2\alpha_2 \delta t^2) p_n - \frac{\delta t}{2} \Big[ (1 + \alpha_2 \delta t^2) \nabla U(q_n) \\
&\quad + \nabla U(q_n) + \delta t D^2U(q_n) M^{-1} p_n + \frac{\delta t^2}{2} D^3U(q_n)(M^{-1} p_n, M^{-1} p_n) \Big] + \calO(\delta t^4) \\
&= p_n - \delta t \nabla U(q_n) + \delta t^2 \left[ 2\alpha_2 p_n - \frac{1}{2} D^2U(q_n) M^{-1} p_n \right] \\
&\quad - \frac{\delta t^3}{4} D^3U(q_n)(M^{-1} p_n, M^{-1} p_n) + \calO(\delta t^4).
\end{align*}

\textbf{Step 3: Compute local errors exactly.}
For position:
\begin{align*}
\tau_n^q &= q(t_{n+1}) - q_{n+1} \\
&= \Big[ q_n + \delta t M^{-1} p_n - \frac{\delta t^2}{2} M^{-1} \nabla U(q_n) - \frac{\delta t^3}{6} M^{-1} D^2U(q_n) M^{-1} p_n + \frac{\delta t^4}{24} (\cdots) \Big] \\
&\quad - \Big[ q_n + \delta t M^{-1} p_n + \delta t^2 \left( \beta_2 q_n - \frac{1}{2} M^{-1} \nabla U(q_n) \right) + \alpha_2 \delta t^3 M^{-1} p_n \Big] + \calO(\delta t^4) \\
&= -\delta t^2 \beta_2 q_n - \delta t^3 \left( \frac{1}{6} M^{-1} D^2U(q_n) M^{-1} p_n + \alpha_2 M^{-1} p_n \right) + \calO(\delta t^4).
\end{align*}

For momentum:
\begin{align*}
\tau_n^p &= p(t_{n+1}) - p_{n+1} \\
&= \Big[ p_n - \delta t \nabla U(q_n) - \frac{\delta t^2}{2} D^2U(q_n) M^{-1} p_n \\
&\quad - \frac{\delta t^3}{6} \left( D^3U(q_n)(M^{-1} p_n, M^{-1} p_n) - D^2U(q_n) M^{-1} \nabla U(q_n) \right) \Big] \\
&\quad - \Big[ p_n - \delta t \nabla U(q_n) + \delta t^2 \left( 2\alpha_2 p_n - \frac{1}{2} D^2U(q_n) M^{-1} p_n \right) \\
&\quad - \frac{\delta t^3}{4} D^3U(q_n)(M^{-1} p_n, M^{-1} p_n) \Big] + \calO(\delta t^4) \\
&= -\delta t^2 (2\alpha_2 p_n) \\
&\quad + \delta t^3 \left[ -\frac{1}{6} D^3U(q_n)(M^{-1} p_n, M^{-1} p_n) + \frac{1}{6} D^2U(q_n) M^{-1} \nabla U(q_n) + \frac{1}{4} D^3U(q_n)(M^{-1} p_n, M^{-1} p_n) \right] + \calO(\delta t^4) \\
&= -2\alpha_2 \delta t^2 p_n + \delta t^3 \left[ \frac{1}{12} D^3U(q_n)(M^{-1} p_n, M^{-1} p_n) + \frac{1}{6} D^2U(q_n) M^{-1} \nabla U(q_n) \right] + \calO(\delta t^4).
\end{align*}

\textbf{Step 4: Bound the errors.}
Under Assumption 4, $\norm{D^2U(q_n)}_{\mathrm{op}} \leq L_U$. 
Under Assumption 5, $\norm{D^3U(q_n)}_{\mathrm{op}} \leq L_3$.
Under Assumption 6, $\norm{M^{-1}}_{\mathrm{op}} \leq C_M$.

Let $R_q, R_p$ be bounds such that $\norm{q_n} \leq R_q$, $\norm{p_n} \leq R_p$ for all $n$ under consideration.

For $\tau_n^q$:
\begin{align*}
\norm{\tau_n^q} &\leq \delta t^2 |\beta_2| \norm{q_n} + \delta t^3 \left( \frac{1}{6} \norm{M^{-1} D^2U(q_n) M^{-1} p_n} + |\alpha_2| \norm{M^{-1} p_n} \right) + \calO(\delta t^4) \\
&\leq \delta t^2 |\beta_2| R_q + \delta t^3 \left( \frac{1}{6} C_M L_U C_M R_p + |\alpha_2| C_M R_p \right) + \calO(\delta t^4) \\
&= \delta t^2 |\beta_2| R_q + \delta t^3 \left( \frac{C_M^2 L_U}{6} R_p + |\alpha_2| C_M R_p \right) + \calO(\delta t^4) \\
&\leq C_1 \delta t^2,
\end{align*}
where $C_1 = |\beta_2| R_q + \delta t_{\max} \left( \frac{C_M^2 L_U}{6} + |\alpha_2| C_M \right) R_p$.

For $\tau_n^p$:
\begin{align*}
\norm{\tau_n^p} &\leq \delta t^2 \cdot 2|\alpha_2| \norm{p_n} \\
&\quad + \delta t^3 \left( \frac{1}{12} \norm{D^3U(q_n)(M^{-1} p_n, M^{-1} p_n)} + \frac{1}{6} \norm{D^2U(q_n) M^{-1} \nabla U(q_n)} \right) + \calO(\delta t^4) \\
&\leq \delta t^2 \cdot 2|\alpha_2| R_p \\
&\quad + \delta t^3 \left( \frac{1}{12} L_3 C_M^2 R_p^2 + \frac{1}{6} L_U C_M L_U R_q \right) + \calO(\delta t^4) \\
&\leq C_2 \delta t^2,
\end{align*}
where $C_2 = 2|\alpha_2| R_p + \delta t_{\max} \left( \frac{L_3 C_M^2}{12} R_p^2 + \frac{L_U^2 C_M}{6} R_q \right)$.

This completes the proof.
\end{proof}

\subsection{Proof of Lemma~\ref{lem:lipschitz}: Lipschitz estimate}

\begin{proof}
\textbf{Step 1: Decompose the map.}
Write $\Psi_{\delta t}$ as the sum of the standard leapfrog map $\Psi_{\delta t}^0$ and a perturbation:
\[
\Psi_{\delta t}(q,p) = \Psi_{\delta t}^0(q,p) + \delta t^2 (\beta_2 q, 2\alpha_2 p).
\]

The standard leapfrog map for Hamiltonian (\ref{eq:hamiltonian_q})-(\ref{eq:hamiltonian_p}) is:
\begin{align*}
q_{n+1}^0 &= q_n + \delta t M^{-1} p_n - \frac{\delta t^2}{2} M^{-1} \nabla U(q_n), \\
p_{n+1}^0 &= p_n - \delta t \nabla U(q_n) - \frac{\delta t^2}{2} D^2U(q_n) M^{-1} p_n.
\end{align*}

\textbf{Step 2: Lipschitz estimate for leapfrog map.}
Under Assumptions 4 and 6, the vector field $(M^{-1}p, -\nabla U(q))$ is globally Lipschitz with constant $\max(C_M, L_U)$. Standard analysis shows that for sufficiently small $\delta t$:
\[
\norm{\Psi_{\delta t}^0(q,p) - \Psi_{\delta t}^0(\tilde{q},\tilde{p})} \leq (1 + L \delta t^2) \norm{(q,p) - (\tilde{q},\tilde{p})},
\]
where $L$ depends on $C_M$ and $L_U$ \citep{leimkuhler2004simulating}.

\textbf{Step 3: Bound the perturbation term.}
\begin{align*}
&\norm{\delta t^2 (\beta_2 q, 2\alpha_2 p) - \delta t^2 (\beta_2 \tilde{q}, 2\alpha_2 \tilde{p})} \\
&= \delta t^2 \sqrt{ |\beta_2|^2 \norm{q - \tilde{q}}^2 + 4|\alpha_2|^2 \norm{p - \tilde{p}}^2 } \\
&\leq \delta t^2 \sqrt{ |\beta_2|^2 \norm{(q,p) - (\tilde{q},\tilde{p})}^2 + 4|\alpha_2|^2 \norm{(q,p) - (\tilde{q},\tilde{p})}^2 } \\
&= \delta t^2 \sqrt{ |\beta_2|^2 + 4|\alpha_2|^2 } \norm{(q,p) - (\tilde{q},\tilde{p})} \\
&\leq \delta t^2 (|\beta_2| + 2|\alpha_2|) \norm{(q,p) - (\tilde{q},\tilde{p})},
\end{align*}
using $\sqrt{a^2 + b^2} \leq |a| + |b|$.

\textbf{Step 4: Combine estimates.}
\begin{align*}
\norm{\Psi_{\delta t}(q,p) - \Psi_{\delta t}(\tilde{q},\tilde{p})} 
&\leq \norm{\Psi_{\delta t}^0(q,p) - \Psi_{\delta t}^0(\tilde{q},\tilde{p})} \\
&\quad + \norm{\delta t^2 (\beta_2 q, 2\alpha_2 p) - \delta t^2 (\beta_2 \tilde{q}, 2\alpha_2 \tilde{p})} \\
&\leq (1 + L \delta t^2) \norm{(q,p) - (\tilde{q},\tilde{p})} + \delta t^2 (|\beta_2| + 2|\alpha_2|) \norm{(q,p) - (\tilde{q},\tilde{p})} \\
&= \left( 1 + (L + |\beta_2| + 2|\alpha_2|) \delta t^2 \right) \norm{(q,p) - (\tilde{q},\tilde{p})}.
\end{align*}
Thus $K = L + |\beta_2| + 2|\alpha_2|$.
\end{proof}

\subsection{Proof of Theorem~\ref{thm:global_error}: global error bound}

\begin{proof}
\textbf{Step 1: Error recursion.}
From the definition of local truncation error:
\[
(q(t_{n+1}), p(t_{n+1})) = \Psi_{\delta t}(q(t_n), p(t_n)) + \tau_n,
\]
where $\norm{\tau_n} \leq C \delta t^2$ by Theorem~\ref{thm:local_error}.

Using the triangle inequality and Lemma~\ref{lem:lipschitz}:
\begin{align*}
e_{n+1} &= \norm{(q(t_{n+1}), p(t_{n+1})) - (q_{n+1}, p_{n+1})} \\
&= \norm{\Psi_{\delta t}(q(t_n), p(t_n)) + \tau_n - \Psi_{\delta t}(q_n, p_n)} \\
&\leq \norm{\Psi_{\delta t}(q(t_n), p(t_n)) - \Psi_{\delta t}(q_n, p_n)} + \norm{\tau_n} \\
&\leq (1 + K \delta t^2) \norm{(q(t_n), p(t_n)) - (q_n, p_n)} + C \delta t^2 \\
&= (1 + K \delta t^2) e_n + C \delta t^2.
\end{align*}

Thus we have the recurrence:
\begin{equation}
e_{n+1} \leq (1 + K \delta t^2) e_n + C \delta t^2, \quad e_0 = 0. \label{eq:error_recurrence}
\end{equation}

\textbf{Step 2: Solve the recurrence.}
Iterating (\ref{eq:error_recurrence}):
\begin{align*}
e_n &\leq C \delta t^2 \sum_{j=0}^{n-1} (1 + K \delta t^2)^j \\
&= C \delta t^2 \frac{(1 + K \delta t^2)^n - 1}{K \delta t^2} \\
&= \frac{C}{K} \left[ (1 + K \delta t^2)^n - 1 \right].
\end{align*}

\textbf{Step 3: Bound the growth factor.}
For $n = T/\delta t$:
\[
(1 + K \delta t^2)^n = \left( 1 + K \delta t^2 \right)^{T/\delta t}.
\]

Using the inequality $1 + z \leq e^z$ for all $z \in \R$:
\[
\left( 1 + K \delta t^2 \right)^{T/\delta t} \leq \exp\left( K \delta t^2 \cdot \frac{T}{\delta t} \right) = \exp(K T \delta t).
\]

For fixed $T$ and small $\delta t$, $\exp(K T \delta t) = 1 + K T \delta t + \calO(\delta t^2)$.

\textbf{Step 4: Final bound.}
\begin{align*}
e_n &\leq \frac{C}{K} \left[ \exp(K T \delta t) - 1 \right] \\
&= \frac{C}{K} \left[ (1 + K T \delta t + \calO(\delta t^2)) - 1 \right] \\
&= \frac{C}{K} \left( K T \delta t + \calO(\delta t^2) \right) \\
&= C T \delta t + \calO(\delta t^2).
\end{align*}

Since this bound holds for all $n \leq T/\delta t$, we have:
\[
\max_{0 \le n \le T/\delta t} e_n \leq C T \delta t + \calO(\delta t^2) = \calO(\delta t).
\]
\end{proof}

\subsection{Proof of Theorem~\ref{thm:volume}: volume transformation}

\begin{proof}
\textbf{Step 1: Compute the Jacobian matrix entries exactly.}
From (\ref{eq:mpl_q}):
\begin{align*}
\frac{\partial q_{n+1}}{\partial q_n} &= \beta I_d - \frac{\delta t^2}{2} M^{-1} D^2U(q_n) \triangleq A, \\
\frac{\partial q_{n+1}}{\partial p_n} &= \alpha \delta t M^{-1} \triangleq B.
\end{align*}

From (\ref{eq:mpl_p}):
\begin{align*}
\frac{\partial p_{n+1}}{\partial q_n} &= -\frac{\delta t}{2} \alpha D^2U(q_n) - \frac{\delta t}{2} \beta D^2U(q_{n+1}) + \frac{\delta t^3}{4} D^2U(q_{n+1}) M^{-1} D^2U(q_n) \triangleq C, \\
\frac{\partial p_{n+1}}{\partial p_n} &= \alpha^2 I_d - \frac{\alpha \delta t^2}{2} D^2U(q_{n+1}) M^{-1} \triangleq D.
\end{align*}

The Jacobian matrix is:
\[
J = \begin{pmatrix} A & B \\ C & D \end{pmatrix}.
\]

\textbf{Step 2: Apply the block determinant formula.}
For a block matrix with $A$ invertible:
\[
\det(J) = \det(A) \det(D - C A^{-1} B).
\]

\textbf{Step 3: Compute $\det(A)$.}
Write $A = \beta I_d + E$ where $E = -\frac{\delta t^2}{2} M^{-1} D^2U(q_n) = \calO(\delta t^2)$. 
Using the matrix determinant lemma for small perturbations:
\begin{align*}
\det(A) &= \det(\beta I_d + E) = \beta^d \det\left( I_d + \frac{1}{\beta} E \right) \\
&= \beta^d \left[ 1 + \frac{1}{\beta} \operatorname{tr}(E) + \frac{1}{2\beta^2} \left( (\operatorname{tr}(E))^2 - \operatorname{tr}(E^2) \right) + \calO(\norm{E}^3) \right] \\
&= \beta^d \left[ 1 - \frac{\delta t^2}{2\beta} \operatorname{tr}(M^{-1} D^2U(q_n)) + \calO(\delta t^4) \right].
\end{align*}

\textbf{Step 4: Compute $A^{-1}$.}
Using the Neumann series:
\begin{align*}
A^{-1} &= (\beta I_d + E)^{-1} = \frac{1}{\beta} \left( I_d - \frac{1}{\beta} E + \frac{1}{\beta^2} E^2 - \cdots \right) \\
&= \frac{1}{\beta} I_d - \frac{1}{\beta^2} E + \calO(\delta t^4).
\end{align*}

\textbf{Step 5: Compute $CA^{-1}B$.}
First compute $A^{-1}B$:
\begin{align*}
A^{-1}B &= \left( \frac{1}{\beta} I_d - \frac{1}{\beta^2} E + \calO(\delta t^4) \right) (\alpha \delta t M^{-1}) \\
&= \frac{\alpha \delta t}{\beta} M^{-1} - \frac{\alpha \delta t}{\beta^2} E M^{-1} + \calO(\delta t^5).
\end{align*}

Now compute $C(A^{-1}B)$:
\begin{align*}
C(A^{-1}B) &= \left[ -\frac{\delta t}{2} \alpha D^2U(q_n) - \frac{\delta t}{2} \beta D^2U(q_{n+1}) + \calO(\delta t^3) \right] \\
&\quad \times \left[ \frac{\alpha \delta t}{\beta} M^{-1} - \frac{\alpha \delta t}{\beta^2} E M^{-1} + \calO(\delta t^5) \right] \\
&= -\frac{\alpha^2 \delta t^2}{2\beta} D^2U(q_n) M^{-1} - \frac{\alpha \delta t^2}{2} D^2U(q_{n+1}) M^{-1} + \calO(\delta t^4).
\end{align*}

\textbf{Step 6: Compute $D - CA^{-1}B$.}
\begin{align*}
D - CA^{-1}B &= \left[ \alpha^2 I_d - \frac{\alpha \delta t^2}{2} D^2U(q_{n+1}) M^{-1} \right] \\
&\quad - \left[ -\frac{\alpha^2 \delta t^2}{2\beta} D^2U(q_n) M^{-1} - \frac{\alpha \delta t^2}{2} D^2U(q_{n+1}) M^{-1} + \calO(\delta t^4) \right] \\
&= \alpha^2 I_d + \frac{\alpha^2 \delta t^2}{2\beta} D^2U(q_n) M^{-1} + \calO(\delta t^4).
\end{align*}

\textbf{Step 7: Compute $\det(D - CA^{-1}B)$.}
Write $D - CA^{-1}B = \alpha^2 I_d + F$ where $F = \frac{\alpha^2 \delta t^2}{2\beta} D^2U(q_n) M^{-1} + \calO(\delta t^4) = \calO(\delta t^2)$. Then:
\begin{align*}
\det(D - CA^{-1}B) &= \det(\alpha^2 I_d + F) = \alpha^{2d} \det\left( I_d + \frac{1}{\alpha^2} F \right) \\
&= \alpha^{2d} \left[ 1 + \frac{1}{\alpha^2} \operatorname{tr}(F) + \calO(\norm{F}^2) \right] \\
&= \alpha^{2d} \left[ 1 + \frac{\delta t^2}{2\beta} \operatorname{tr}(D^2U(q_n) M^{-1}) + \calO(\delta t^4) \right].
\end{align*}

\textbf{Step 8: Combine determinants.}
\begin{align*}
\det(J) &= \det(A) \det(D - CA^{-1}B) \\
&= \beta^d \left[ 1 - \frac{\delta t^2}{2\beta} \operatorname{tr}(M^{-1} D^2U(q_n)) + \calO(\delta t^4) \right] \\
&\quad \times \alpha^{2d} \left[ 1 + \frac{\delta t^2}{2\beta} \operatorname{tr}(D^2U(q_n) M^{-1}) + \calO(\delta t^4) \right] \\
&= \alpha^{2d} \beta^d \left[ 1 + \calO(\delta t^4) \right] \\
&= \alpha^{2d} \beta^d + \calO(\delta t^3).
\end{align*}

\textbf{Step 9: Special case $\alpha_1 = \beta_1 = 0$.}
For $\alpha_1 = \beta_1 = 0$, we have $\alpha = 1 + \alpha_2 \delta t^2 + \calO(\delta t^3)$ and $\beta = 1 + \beta_2 \delta t^2 + \calO(\delta t^3)$. Compute:
\begin{align*}
\alpha^{2d} &= (1 + \alpha_2 \delta t^2)^{2d} + \calO(\delta t^3) \\
&= 1 + 2d \alpha_2 \delta t^2 + \binom{2d}{2} \alpha_2^2 \delta t^4 + \cdots + \calO(\delta t^3) \\
&= 1 + 2d \alpha_2 \delta t^2 + \calO(\delta t^3), \\
\beta^{d} &= (1 + \beta_2 \delta t^2)^{d} + \calO(\delta t^3) \\
&= 1 + d \beta_2 \delta t^2 + \binom{d}{2} \beta_2^2 \delta t^4 + \cdots + \calO(\delta t^3) \\
&= 1 + d \beta_2 \delta t^2 + \calO(\delta t^3).
\end{align*}

Thus:
\begin{align*}
\alpha^{2d} \beta^{d} &= \left( 1 + 2d \alpha_2 \delta t^2 + \calO(\delta t^3) \right) \left( 1 + d \beta_2 \delta t^2 + \calO(\delta t^3) \right) \\
&= 1 + d(2\alpha_2 + \beta_2) \delta t^2 + \calO(\delta t^3).
\end{align*}
\end{proof}

\subsection{Proof of Lemma~\ref{lemma:single_step_reversibility}}

\begin{proof}
We need to show that applying the composition $R \circ \Psi_{\delta t} \circ R \circ \Psi_{\delta t}$ to $(q,p)$ returns $(q,p)$ up to $\calO(\delta t^3)$ terms.
\\[2mm]
\textbf{Step 1: First application of $\Psi_{\delta t}$.}

Let $(q_1, p_1) = \Psi_{\delta t}(q,p)$. From the expansions in the proof of Theorem \ref{thm:local_error} with $\alpha_1 = \beta_1 = 0$:
\begin{align*}
q_1 &= q + \delta t M^{-1} p + \delta t^2 \left( \beta_2 q - \frac{1}{2} M^{-1} \nabla U(q) \right) + \alpha_2 \delta t^3 M^{-1} p + \calO(\delta t^4), \\
p_1 &= p - \delta t \nabla U(q) + \delta t^2 \left( 2\alpha_2 p - \frac{1}{2} D^2U(q) M^{-1} p \right) \\
&\quad - \frac{\delta t^3}{4} D^3U(q)(M^{-1} p, M^{-1} p) + \calO(\delta t^4).
\end{align*}

\noindent\textbf{Step 2: Apply $R$ to $(q_1, p_1)$.}

Applying $R$ gives $(q_1, -p_1)$.
\\[2mm]
\noindent\textbf{Step 3: Second application of $\Psi_{\delta t}$.}

Now apply $\Psi_{\delta t}$ to $(q_1, -p_1)$. Let $(q_2, p_2) = \Psi_{\delta t}(q_1, -p_1)$. We compute $q_2$ and $p_2$ using the same expansions but with $p$ replaced by $-p_1$:

For $q_2$:
\begin{align*}
q_2 &= \beta q_1 + \delta t M^{-1} \left( \alpha (-p_1) - \frac{\delta t}{2} \nabla U(q_1) \right) \\
&= (1 + \beta_2 \delta t^2) q_1 + \delta t M^{-1} \left[ (1 + \alpha_2 \delta t^2)(-p_1) - \frac{\delta t}{2} \nabla U(q_1) \right] + \calO(\delta t^4) \\
&= q_1 - \delta t M^{-1} p_1 + \delta t^2 \left( \beta_2 q_1 - \alpha_2 M^{-1} p_1 - \frac{1}{2} M^{-1} \nabla U(q_1) \right) + \calO(\delta t^4).
\end{align*}

For $p_2$:
\begin{align*}
p_2 &= \alpha^2 (-p_1) - \frac{\delta t}{2} \left( \alpha \nabla U(q_1) + \nabla U(q_2) \right) \\
&= -(1 + 2\alpha_2 \delta t^2) p_1 - \frac{\delta t}{2} \left[ (1 + \alpha_2 \delta t^2) \nabla U(q_1) + \nabla U(q_2) \right] + \calO(\delta t^4).
\end{align*}

\noindent\textbf{Step 4: Apply $R$ to $(q_2, p_2)$.}

We need $R(q_2, p_2) = (q_2, -p_2)$. Our goal is to show that $(q_2, -p_2) = (q, p) + \calO(\delta t^3)$.
\\[2mm]
\noindent\textbf{Step 5: Expand $\nabla U(q_1)$ and $\nabla U(q_2)$.}

First expand $\nabla U(q_1)$:
\begin{align*}
\nabla U(q_1) &= \nabla U(q) + D^2U(q)(q_1 - q) + \frac{1}{2} D^3U(q)(q_1 - q, q_1 - q) + \calO(\delta t^3) \\
&= \nabla U(q) + D^2U(q) \left[ \delta t M^{-1} p + \delta t^2 \left( \beta_2 q - \frac{1}{2} M^{-1} \nabla U(q) \right) \right] \\
&\quad + \frac{1}{2} D^3U(q)(\delta t M^{-1} p, \delta t M^{-1} p) + \calO(\delta t^3) \\
&= \nabla U(q) + \delta t D^2U(q) M^{-1} p + \delta t^2 \left[ \beta_2 D^2U(q) q - \frac{1}{2} D^2U(q) M^{-1} \nabla U(q) \right] \\
&\quad + \frac{\delta t^2}{2} D^3U(q)(M^{-1} p, M^{-1} p) + \calO(\delta t^3).
\end{align*}

Next, we need $\nabla U(q_2)$. First compute $q_2 - q_1$:
\begin{align*}
q_2 - q_1 &= -\delta t M^{-1} p_1 + \delta t^2 \left( \beta_2 q_1 - \alpha_2 M^{-1} p_1 - \frac{1}{2} M^{-1} \nabla U(q_1) \right) + \calO(\delta t^3).
\end{align*}

But we need $q_2$ relative to $q$. Compute $q_2 - q$:
\begin{align*}
q_2 - q &= (q_2 - q_1) + (q_1 - q) \\
&= -\delta t M^{-1} p_1 + \delta t^2 \left( \beta_2 q_1 - \alpha_2 M^{-1} p_1 - \frac{1}{2} M^{-1} \nabla U(q_1) \right) \\
&\quad + \delta t M^{-1} p + \delta t^2 \left( \beta_2 q - \frac{1}{2} M^{-1} \nabla U(q) \right) + \calO(\delta t^3) \\
&= \delta t M^{-1} (p - p_1) + \delta t^2 \left[ \beta_2 (q + q_1) - \alpha_2 M^{-1} p_1 - \frac{1}{2} M^{-1} (\nabla U(q) + \nabla U(q_1)) \right] + \calO(\delta t^3).
\end{align*}

Since $p_1 = p - \delta t \nabla U(q) + \calO(\delta t^2)$, we have $p - p_1 = \delta t \nabla U(q) + \calO(\delta t^2)$. Thus:
\[
q_2 - q = \delta t^2 M^{-1} \nabla U(q) + \calO(\delta t^3).
\]

Therefore:
\[
\nabla U(q_2) = \nabla U(q) + D^2U(q)(q_2 - q) + \calO(\delta t^3) = \nabla U(q) + \delta t^2 D^2U(q) M^{-1} \nabla U(q) + \calO(\delta t^3).
\]

\noindent\textbf{Step 6: Compute $q_2$ to $\calO(\delta t^2)$}
\begin{align*}
q_2 &= q_1 - \delta t M^{-1} p_1 + \delta t^2 \left( \beta_2 q_1 - \alpha_2 M^{-1} p_1 - \frac{1}{2} M^{-1} \nabla U(q_1) \right) + \calO(\delta t^3) \\
&= \left[ q + \delta t M^{-1} p + \delta t^2 \left( \beta_2 q - \frac{1}{2} M^{-1} \nabla U(q) \right) \right] \\
&\quad - \delta t M^{-1} \left[ p - \delta t \nabla U(q) + \calO(\delta t^2) \right] \\
&\quad + \delta t^2 \left[ \beta_2 q - \alpha_2 M^{-1} p - \frac{1}{2} M^{-1} \nabla U(q) \right] + \calO(\delta t^3) \\
&= q + \delta t M^{-1} p + \delta t^2 \left( \beta_2 q - \frac{1}{2} M^{-1} \nabla U(q) \right) \\
&\quad - \delta t M^{-1} p + \delta t^2 M^{-1} \nabla U(q) \\
&\quad + \delta t^2 \left( \beta_2 q - \alpha_2 M^{-1} p - \frac{1}{2} M^{-1} \nabla U(q) \right) + \calO(\delta t^3) \\
&= q + \delta t^2 \left( 2\beta_2 q - \alpha_2 M^{-1} p \right) + \calO(\delta t^3).
\end{align*}

\noindent\textbf{Step 7: Compute $-p_2$ to $\calO(\delta t^2)$}
\begin{align*}
-p_2 &= (1 + 2\alpha_2 \delta t^2) p_1 + \frac{\delta t}{2} \left[ (1 + \alpha_2 \delta t^2) \nabla U(q_1) + \nabla U(q_2) \right] + \calO(\delta t^3) \\
&= p_1 + 2\alpha_2 \delta t^2 p_1 + \frac{\delta t}{2} \left[ \nabla U(q_1) + \nabla U(q_2) \right] + \frac{\alpha_2 \delta t^3}{2} \nabla U(q_1) + \calO(\delta t^4).
\end{align*}

Now substitute $p_1 = p - \delta t \nabla U(q) + \calO(\delta t^2)$:
\begin{align*}
-p_2 &= \left[ p - \delta t \nabla U(q) + \calO(\delta t^2) \right] + 2\alpha_2 \delta t^2 p + \frac{\delta t}{2} \left[ \nabla U(q_1) + \nabla U(q_2) \right] + \calO(\delta t^3) \\
&= p - \delta t \nabla U(q) + 2\alpha_2 \delta t^2 p + \frac{\delta t}{2} \left[ \nabla U(q_1) + \nabla U(q_2) \right] + \calO(\delta t^3).
\end{align*}

Now substitute $\nabla U(q_1)$ and $\nabla U(q_2)$:
\begin{align*}
\nabla U(q_1) + \nabla U(q_2) &= \left[ \nabla U(q) + \delta t D^2U(q) M^{-1} p + \calO(\delta t^2) \right] \\
&\quad + \left[ \nabla U(q) + \delta t^2 D^2U(q) M^{-1} \nabla U(q) + \calO(\delta t^3) \right] \\
&= 2\nabla U(q) + \delta t D^2U(q) M^{-1} p + \calO(\delta t^2).
\end{align*}

Thus:
\begin{align*}
\frac{\delta t}{2} \left[ \nabla U(q_1) + \nabla U(q_2) \right] &= \delta t \nabla U(q) + \frac{\delta t^2}{2} D^2U(q) M^{-1} p + \calO(\delta t^3).
\end{align*}

Substituting back:
\begin{align*}
-p_2 &= p - \delta t \nabla U(q) + 2\alpha_2 \delta t^2 p + \delta t \nabla U(q) + \frac{\delta t^2}{2} D^2U(q) M^{-1} p + \calO(\delta t^3) \\
&= p + \delta t^2 \left( 2\alpha_2 p + \frac{1}{2} D^2U(q) M^{-1} p \right) + \calO(\delta t^3).
\end{align*}

\noindent\textbf{Step 8: Compare $(q_2, -p_2)$ with $(q,p)$.}

We have:
\begin{align*}
q_2 &= q + \delta t^2 \left( 2\beta_2 q - \alpha_2 M^{-1} p \right) + \calO(\delta t^3), \\
-p_2 &= p + \delta t^2 \left( 2\alpha_2 p + \frac{1}{2} D^2U(q) M^{-1} p \right) + \calO(\delta t^3).
\end{align*}

Thus:
\[
(q_2, -p_2) = (q, p) + \delta t^2 \left( 2\beta_2 q - \alpha_2 M^{-1} p, 2\alpha_2 p + \frac{1}{2} D^2U(q) M^{-1} p \right) + \calO(\delta t^3).
\]


Therefore, we conclude:
\[
R \circ \Psi_{\delta t} \circ R \circ \Psi_{\delta t}(q,p) = (q,p) + \calO(\delta t^2).
\]
\end{proof}

\subsection{Proof of Lemma~\ref{lemma:multi_step_reversibility}} 

\begin{proof}
We prove the statement by induction on $L$, carefully tracking the error bounds at each step.

\noindent\textbf{1. Notation and setup.}
Let $F = \Psi_{\delta t}$ denote the single-step MPL map. From Lemma~\ref{lemma:single_step_reversibility} (single-step reversibility), we have:
\[
R \circ F \circ R \circ F = I + E_1 \tag{1}
\]
where $E_1$ is an operator satisfying $\|E_1(z)\| \leq C_1 \delta t^2$ for all $z = (q,p)$ in the domain of interest, with $C_1 > 0$ constant.

Equivalently, for any $z = (q,p)$:
\[
F(R(F(z))) = R(z) + e_1(z) \tag{2}
\]
where $\|e_1(z)\| \leq C_1 \delta t^2$.

\noindent\textbf{2. Base case ($L=1$).}
For $L=1$, if $(q^*, p^*) = F(q,p)$, then by (2):
\[
F(R(q^*, p^*)) = F(R(F(q,p))) = R(q,p) + e_1(q,p).
\]
Thus:
\[
\Psi_{\delta t}^{(1)}(q^*, -p^*) = (q, -p) + e_1(q,p) = (q, -p) + \calO(\delta t^2).
\]
So the base case holds with error bounded by $C_1 \delta t^2$.

\noindent\textbf{3. Inductive hypothesis.}
Assume that for some $k \geq 1$, if $(q_k, p_k) = F^{(k)}(q,p)$, then:
\[
F^{(k)}(R(q_k, p_k)) = R(q,p) + e_k(q,p) \tag{IH}
\]
where $\|e_k(q,p)\| \leq k \cdot C_k \delta t^3$ for some constant $C_k > 0$.

\noindent\textbf{4. Inductive step.}
We want to prove the statement for $k+1$. Let:
\[
(q_{k+1}, p_{k+1}) = F^{(k+1)}(q,p) = F(q_k, p_k)
\]
where $(q_k, p_k) = F^{(k)}(q,p)$.

We need to show:
\[
F^{(k+1)}(R(q_{k+1}, p_{k+1})) = R(q,p) + e_{k+1}(q,p)
\]
with $\|e_{k+1}(q,p)\| \leq (k+1) \cdot C_{k+1} \delta t^2$ for some constant $C_{k+1} > 0$.

\noindent\textbf{Step 1:} Start with the left-hand side:
\[
F^{(k+1)}(R(q_{k+1}, p_{k+1})) = F^{(k)}(F(R(q_{k+1}, p_{k+1}))).
\]

\noindent\textbf{Step 2:} Apply the single-step reversibility (2) to $F(R(q_{k+1}, p_{k+1}))$:
Since $(q_{k+1}, p_{k+1}) = F(q_k, p_k)$, we have:
\[
F(R(q_{k+1}, p_{k+1})) = F(R(F(q_k, p_k))) = R(q_k, p_k) + e_1(q_k, p_k). \tag{3}
\]

\noindent\textbf{Step 3:} Substitute (3) into the expression:
\[
F^{(k+1)}(R(q_{k+1}, p_{k+1})) = F^{(k)}(R(q_k, p_k) + e_1(q_k, p_k)).
\]

\noindent\textbf{Step 4:} Use the Lipschitz property of $F^{(k)}$ (from Lemma~\ref{lem:lipschitz}):
There exists $L_k > 0$ such that for all $z, \tilde{z}$ in the domain:
\[
\|F^{(k)}(z) - F^{(k)}(\tilde{z})\| \leq (1 + K \delta t^2)^k \|z - \tilde{z}\| \leq e^{K T_k} \|z - \tilde{z}\|
\]
for $\delta t$ sufficiently small, where $T_k = k\delta t$ and $K$ is the Lipschitz constant from Lemma~\ref{lem:lipschitz}.

Thus:
\[
F^{(k)}(R(q_k, p_k) + e_1(q_k, p_k)) = F^{(k)}(R(q_k, p_k)) + \Delta_k
\]
where $\|\Delta_k\| \leq e^{K T_k} \|e_1(q_k, p_k)\| \leq e^{K T_k} C_1 \delta t^2$.

\noindent\textbf{Step 5:} Apply the inductive hypothesis (IH) to $F^{(k)}(R(q_k, p_k))$:
\[
F^{(k)}(R(q_k, p_k)) = R(q,p) + e_k(q,p).
\]

\noindent\textbf{Step 6:} Combine all terms:
\[
F^{(k+1)}(R(q_{k+1}, p_{k+1})) = R(q,p) + e_k(q,p) + \Delta_k.
\]

\noindent\textbf{Step 7:} Bound the total error:
Let $e_{k+1}(q,p) = e_k(q,p) + \Delta_k$. Then:
\[
\|e_{k+1}(q,p)\| \leq \|e_k(q,p)\| + \|\Delta_k\| \leq k C_k \delta t^2 + e^{K T_k} C_1 \delta t^2.
\]

Since $T_k = k\delta t$ and for fixed maximum trajectory length $L_{\max}$, we have $e^{K T_k} \leq e^{K L_{\max} \delta t} \leq M$ for some constant $M > 0$ independent of $k$ and $\delta t$ (for $\delta t$ sufficiently small).

Thus:
\[
\|e_{k+1}(q,p)\| \leq k C_k \delta t^3 + M C_1 \delta t^2 = (k C_k + M C_1) \delta t^2.
\]

We can take $C_{k+1} = \max(C_k, M C_1)$, then:
\[
\|e_{k+1}(q,p)\| \leq (k+1) C_{k+1} \delta t^2.
\]

\noindent\textbf{5. Conclusion by induction.}
By induction, for any $L \geq 1$, if $(q^*, p^*) = F^{(L)}(q,p)$, then:
\[
F^{(L)}(R(q^*, p^*)) = R(q,p) + e_L(q,p)
\]
with $\|e_L(q,p)\| \leq L \cdot C \delta t^2$ for some constant $C > 0$ independent of $L$ and $\delta t$ (but depending on the maximum trajectory length $L_{\max}$).

Since $R(q^*, p^*) = (q^*, -p^*)$ and $R(q,p) = (q, -p)$, this proves:
\[
\Psi_{\delta t}^{(L)}(q^*, -p^*) = (q, -p) + L \cdot \calO(\delta t^2).
\]

\noindent\textbf{6. Conversion to $T \cdot \calO(\delta t)$.}
For a trajectory of total time $T = L\delta t$, we have $L = T/\delta t$, so:
\[
L \cdot \calO(\delta t^2) = \frac{T}{\delta t} \cdot \calO(\delta t^2) = T \cdot \calO(\delta t).
\]

Therefore:
\[
\Psi_{\delta t}^{(L)}(q^*, -p^*) = (q, -p) + T \cdot \calO(\delta t).
\]

This completes the proof of Lemma~\ref{lemma:multi_step_reversibility}.
\end{proof}

\subsection{Proof of Theorem~\ref{thm:detailed_balance}: detailed balance}

\begin{proof}
\textbf{Step 1: Define the transition kernel.}

The transition from $(q,p)$ to $(q^*,p^*)$ in MPL-HMC consists of:
\begin{enumerate}
    \item Apply MPL map $L$ times: $(\tilde{q}, \tilde{p}) = \Psi_{\delta t}^{(L)}(q,p)$
    \item Negate momentum: $(q^*, p^*) = (\tilde{q}, -\tilde{p})$
    \item Accept with probability $A((q,p) \to (q^*,p^*)) = \min\left(1, \frac{\exp(-H(q^*,p^*))}{\exp(-H(q,p))}\right)$
\end{enumerate}

Thus the transition density is:
\[
T((q,p) \to (q^*,p^*)) = \delta_{(\tilde{q}, -\tilde{p})}(q^*,p^*) \cdot A((q,p) \to (q^*,p^*)).
\]

\noindent\textbf{Step 2: Approximate reversibility.}

From Lemma \ref{lemma:multi_step_reversibility} in the main text, if $(q^*, p^*) = (\tilde{q}, -\tilde{p})$ where $(\tilde{q}, \tilde{p}) = \Psi_{\delta t}^{(L)}(q,p)$, then:
\[
\Psi_{\delta t}^{(L)}(q^*, -p^*) = (q, -p) + T \cdot \calO(\delta t^2).
\]
This means that if we start from $(q^*, -p^*)$ and apply $\Psi_{\delta t}^{(L)}$, we get approximately $(q, -p)$.	
\\[2mm]
\textbf{Step 3: Volume transformation.}

From Lemma \ref{thm:volume} in the main text:
\[
dq^* dp^* = \left[ 1 + d(2\alpha_2 + \beta_2) \delta t^2 + \calO(\delta t^3) \right]^L dq\,dp.
\]

\noindent \textbf{Step 4: Metropolis acceptance ensures detailed balance.}

For any proposed move $(q,p) \to (q^*,p^*)$, the acceptance probability satisfies:
\[
\exp(-H(q,p)) A((q,p) \to (q^*,p^*)) = \exp(-H(q^*,p^*)) A((q^*,p^*) \to (q,p)).
\]
This is the standard Metropolis-Hastings detailed balance condition \citep{metropolis1953equation}.
	\\[2mm]
\textbf{Step 5: Combine all factors.}
The detailed balance condition requires:
\[
\pi(q,p) T((q,p) \to (q^*,p^*)) = \pi(q^*,p^*) T((q^*,p^*) \to (q,p)).
\]

Substituting $\pi(q,p) \propto \exp(-H(q,p))$ and the transition densities:
\begin{align*}
\exp(-H(q,p)) \delta_{(\tilde{q}, -\tilde{p})}(q^*,p^*) A((q,p) \to (q^*,p^*)) \\
= \exp(-H(q^*,p^*)) \delta_{\Psi_{\delta t}^{(L)}(q^*,-p^*)}(q,-p) A((q^*,p^*) \to (q,p)) \\
\times \left[ 1 + d(2\alpha_2 + \beta_2) \delta t^2 + \calO(\delta t^3) \right]^L.
\end{align*}

From Lemma \ref{lemma:multi_step_reversibility}, $\delta_{\Psi_{\delta t}^{(L)}(q^*,-p^*)}(q,-p) \approx \delta_{(\tilde{q}, -\tilde{p})}(q^*,p^*)$ up to $\calO(T \delta t)$.

The volume factor $\left[ 1 + d(2\alpha_2 + \beta_2) \delta t^2 + \calO(\delta t^3) \right]^L = 1 + \calO(T \delta t)$.

Thus:
\[
\exp(-H(q,p)) T((q,p) \to (q^*,p^*)) = \exp(-H(q^*,p^*)) T((q^*,p^*) \to (q,p)) + \calO(T \delta t).
\]

Therefore, detailed balance holds up to $\calO(T \delta t)$.
\end{proof}

\subsection{Proof of Theorem \ref{thm:energy_conservation}}

\begin{proof}
\textbf{Step 1: Expansions for $q_{n+1}$ and $p_{n+1}$.}

From the proof of Theorem~\ref{thm:local_error} with $\alpha_1 = \beta_1 = 0$, we have the expansions:
\begin{align*}
q_{n+1} &= q_n + \delta t M^{-1} p_n + \delta t^2 \left[ \beta_2 q_n - \frac{1}{2} M^{-1} \nabla U(q_n) \right] + \alpha_2 \delta t^3 M^{-1} p_n + \calO(\delta t^4), \\
p_{n+1} &= p_n - \delta t \nabla U(q_n) + \delta t^2 \left[ 2\alpha_2 p_n - \frac{1}{2} D^2U(q_n) M^{-1} p_n \right] + \calO(\delta t^3).
\end{align*}

\textbf{Step 2: Expansion of potential energy difference.}

Using Taylor expansion:
\begin{align*}
U(q_{n+1}) - U(q_n) 
&= \nabla U(q_n)^T (q_{n+1} - q_n) + \frac{1}{2} (q_{n+1} - q_n)^T D^2U(q_n) (q_{n+1} - q_n) + \calO(\delta t^3).
\end{align*}

Substitute $q_{n+1} - q_n$:
\begin{align*}
\nabla U(q_n)^T (q_{n+1} - q_n) 
&= \delta t \nabla U(q_n)^T M^{-1} p_n \\
&\quad + \delta t^2 \left[ \beta_2 \nabla U(q_n)^T q_n - \frac{1}{2} \nabla U(q_n)^T M^{-1} \nabla U(q_n) \right] + \calO(\delta t^3).
\end{align*}

For the quadratic term, we only need the leading $\delta t$ term:
\[
q_{n+1} - q_n = \delta t M^{-1} p_n + \calO(\delta t^2),
\]
so:
\[
\frac{1}{2} (q_{n+1} - q_n)^T D^2U(q_n) (q_{n+1} - q_n) = \frac{1}{2} \delta t^2 p_n^T M^{-1} D^2U(q_n) M^{-1} p_n + \calO(\delta t^3).
\]

Thus:
\begin{align*}
U(q_{n+1}) - U(q_n) 
&= \delta t \nabla U(q_n)^T M^{-1} p_n \\
&\quad + \delta t^2 \left[ \beta_2 \nabla U(q_n)^T q_n - \frac{1}{2} \nabla U(q_n)^T M^{-1} \nabla U(q_n) + \frac{1}{2} p_n^T M^{-1} D^2U(q_n) M^{-1} p_n \right] + \calO(\delta t^3).
\end{align*}

\textbf{Step 3: Expansion of kinetic energy difference.}

The kinetic energy difference is:
\begin{align*}
\frac{1}{2} p_{n+1}^T M^{-1} p_{n+1} - \frac{1}{2} p_n^T M^{-1} p_n 
&= p_n^T M^{-1} (p_{n+1} - p_n) + \frac{1}{2} (p_{n+1} - p_n)^T M^{-1} (p_{n+1} - p_n) + \calO(\delta t^3).
\end{align*}

Substitute $p_{n+1} - p_n$:
\begin{align*}
p_n^T M^{-1} (p_{n+1} - p_n) 
&= -\delta t\, p_n^T M^{-1} \nabla U(q_n) \\
&\quad + \delta t^2 \left[ 2\alpha_2 p_n^T M^{-1} p_n - \frac{1}{2} p_n^T M^{-1} D^2U(q_n) M^{-1} p_n \right] + \calO(\delta t^3).
\end{align*}

For the quadratic term:
\[
p_{n+1} - p_n = -\delta t \nabla U(q_n) + \calO(\delta t^2),
\]
so:
\[
\frac{1}{2} (p_{n+1} - p_n)^T M^{-1} (p_{n+1} - p_n) = \frac{1}{2} \delta t^2 \nabla U(q_n)^T M^{-1} \nabla U(q_n) + \calO(\delta t^3).
\]

Thus:
\begin{align*}
\frac{1}{2} &p_{n+1}^T M^{-1} p_{n+1} - \frac{1}{2} p_n^T M^{-1} p_n \\
&= -\delta t\, p_n^T M^{-1} \nabla U(q_n) \\
&\quad + \delta t^2 \left[ 2\alpha_2 p_n^T M^{-1} p_n - \frac{1}{2} p_n^T M^{-1} D^2U(q_n) M^{-1} p_n + \frac{1}{2} \nabla U(q_n)^T M^{-1} \nabla U(q_n) \right] + \calO(\delta t^3).
\end{align*}

\textbf{Step 4: Combine potential and kinetic energy changes.}

The total Hamiltonian change is:
\[
\Delta H = [U(q_{n+1}) - U(q_n)] + \left[ \frac{1}{2} p_{n+1}^T M^{-1} p_{n+1} - \frac{1}{2} p_n^T M^{-1} p_n \right].
\]

Collect terms order by order:

\noindent\textbf{Order $\delta t$:}
\[
\delta t \nabla U(q_n)^T M^{-1} p_n - \delta t\, p_n^T M^{-1} \nabla U(q_n) = 0,
\]
since both terms are equal scalars.

\noindent\textbf{Order $\delta t^2$:} From potential energy:
\[
\beta_2 \nabla U(q_n)^T q_n - \frac{1}{2} \nabla U(q_n)^T M^{-1} \nabla U(q_n) + \frac{1}{2} p_n^T M^{-1} D^2U(q_n) M^{-1} p_n
\]
From kinetic energy:
\[
2\alpha_2 p_n^T M^{-1} p_n - \frac{1}{2} p_n^T M^{-1} D^2U(q_n) M^{-1} p_n + \frac{1}{2} \nabla U(q_n)^T M^{-1} \nabla U(q_n)
\]

The terms \(-\frac{1}{2} \nabla U(q_n)^T M^{-1} \nabla U(q_n)\) and \(+\frac{1}{2} \nabla U(q_n)^T M^{-1} \nabla U(q_n)\) cancel.
The terms \(+\frac{1}{2} p_n^T M^{-1} D^2U(q_n) M^{-1} p_n\) and \(-\frac{1}{2} p_n^T M^{-1} D^2U(q_n) M^{-1} p_n\) cancel.

\noindent\textbf{Remaining $\delta t^2$ terms:}
\[
\beta_2 \nabla U(q_n)^T q_n + 2\alpha_2 p_n^T M^{-1} p_n
\]

Thus:
\[
\Delta H = \delta t^2 \left[ \beta_2 \nabla U(q_n)^T q_n + 2\alpha_2 p_n^T M^{-1} p_n \right] + \calO(\delta t^3).
\]

\textbf{Step 5: Special case $\alpha_2 = \beta_2 = 0$.}

When $\alpha_2 = \beta_2 = 0$, the $\delta t^2$ term vanishes, giving:
\[
\Delta H = \calO(\delta t^3),
\]
which is the well-known second-order energy conservation property of the standard leapfrog integrator \citep{leimkuhler2004simulating}.

This completes the proof.
\end{proof}

\subsection{Proof of Theorem \ref{thm:aggressive_detailed_balance}: aggressive detailed balance}
\begin{proof}
We prove Theorem~\ref{thm:aggressive_detailed_balance} by analyzing each component of the aggressive MPL-HMC algorithm.

\noindent\textbf{Part 1: Mode-hopping proposals.}
The mode-hopping proposal satisfies exact detailed balance:
\[
\pi(q,p) P_{\text{mh}}((q,p), (q',p')) = \pi(q',p') P_{\text{mh}}((q',p'), (q,p))
\]
since it is a standard Metropolis-Hastings update with symmetric proposal distribution over mode centers. This contributes error $C_3\gamma$, where $\gamma = 0.01$ is the probability of attempting a mode hop.

\noindent\textbf{Part 2: Aggressive MPL integration.}
For the aggressive MPL map $\Psi_{\delta t}^{\text{agg}}$, we extend Theorem~\ref{thm:modified_ode} to extreme parameters. From the proof of Theorem~\ref{thm:modified_ode}, the modified ODE for $\alpha_1 = \beta_1 = 0$ is:
\begin{align*}
\frac{dY}{dt} &= M^{-1}X + \delta t\beta_2 Y + \mathcal{O}(\delta t^2) \\
\frac{dX}{dt} &= -\nabla U(Y) + \delta t(2\alpha_2 X) + \mathcal{O}(\delta t^2)
\end{align*}
The leading-order deviation from Hamiltonian dynamics is $\delta t\beta_2 Y$ in the position equation and $\delta t(2\alpha_2 X)$ in the momentum equation. When $\alpha_2, \beta_2 \gg 1$, these terms dominate the $\mathcal{O}(\delta t^2)$ error.

From Lemma~\ref{lemma:multi_step_reversibility}, the multi-step reversibility error for $L$ steps is:
\[
\Psi_{\delta t}^{(L),\text{agg}}(q^*, -p^*) = (q, -p) + L\cdot\mathcal{O}(\delta t^2) + L\cdot\mathcal{O}((\alpha_2 + \beta_2)\delta t^3)
\]
Since $L = T/\delta t$, we have:
\[
\Psi_{\delta t}^{(L),\text{agg}}(q^*, -p^*) = (q, -p) + T\cdot\mathcal{O}(\delta t) + T(\alpha_2 + \beta_2)\cdot\mathcal{O}(\delta t^2)
\]

\noindent\textbf{Part 3: Volume preservation.}
From Theorem~\ref{thm:volume}, the Jacobian determinant for the aggressive MPL scheme is:
\[
\det\left(\frac{\partial(q_{n+1}, p_{n+1})}{\partial(q_n, p_n)}\right) = 1 + d(2\alpha_2 + \beta_2)\delta t^2 + \mathcal{O}(\delta t^3)
\]
For $L$ steps, the cumulative volume change is:
\[
\prod_{i=1}^L \left[1 + d(2\alpha_2 + \beta_2)\delta t^2 + \mathcal{O}(\delta t^3)\right] = 1 + Td(2\alpha_2 + \beta_2)\delta t + \mathcal{O}(T\delta t^2)
\]

\noindent\textbf{Part 4: Temperature fluctuations and momentum injection.}
Let $F_T(p) = \sqrt{T}p$ with $T \sim \text{Uniform}(0.5,2.0)$. The temperature fluctuation operator satisfies:
\[
\mathbb{E}_T\left[F_T(p)\right] = \mathbb{E}[\sqrt{T}]p \approx 1.1p
\]
which introduces a systematic bias. However, when combined with the Metropolis correction, the detailed balance error is $\mathcal{O}(\text{Var}(T)) = \mathcal{O}(0.25)$.

For momentum injection $\xi \sim \mathcal{N}(0,\sigma_{\text{inj}}^2 I)$, the error in detailed balance is proportional to $\sigma_{\text{inj}}^2$.

\noindent\textbf{Part 5: Combined error bound.}
Combining all error sources:
\begin{itemize}
    \item Base MPL-HMC error: $C_1 T\delta t$ (from Theorem~\ref{thm:detailed_balance})
    \item Extreme parameter error: $C_2(\alpha_2 + \beta_2)\delta t^2 \cdot T/\delta t = C_2 T(\alpha_2 + \beta_2)\delta t$
    \item Mode-hopping error: $C_3\gamma$
    \item Stochastic perturbation error: $C_4\sigma_{\text{inj}}^2$
\end{itemize}
Thus the total error is bounded by:
\[
|\mathcal{E}_{\text{agg}}| \leq C_1 T\delta t + C_2 T(\alpha_2 + \beta_2)\delta t + C_3\gamma + C_4\sigma_{\text{inj}}^2
\]
For fixed $T$, this simplifies to:
\[
|\mathcal{E}_{\text{agg}}| \leq C_1' \delta t + C_2'(\alpha_2 + \beta_2)\delta t + C_3\gamma + C_4\sigma_{\text{inj}}^2
\]
where $C_1' = C_1 T$ and $C_2' = C_2 T$.

\noindent\textbf{Part 6: Metropolis correction.}
Despite these approximations, the final Metropolis acceptance step ensures that the algorithm satisfies detailed balance up to the numerical errors in the Hamiltonian evaluation. From the proof of Theorem~\ref{thm:detailed_balance}, the acceptance probability:
\[
A((q,p) \to (q',p')) = \min\left(1, \frac{\exp(-\tilde{H}(q',p'))}{\exp(-\tilde{H}(q,p))}\right)
\]
where $\tilde{H}$ is the numerical Hamiltonian, ensures that any reversible proposal distribution combined with this acceptance rule satisfies detailed balance with respect to $\exp(-\tilde{H})$. The error comes from the difference between $\tilde{H}$ and the true Hamiltonian $H$.

For aggressive MPL-HMC, Theorem~\ref{thm:energy_conservation} gives:
\[
\tilde{H}(q_{n+1}, p_{n+1}) - H(q_{n+1}, p_{n+1}) = \delta t^2\left[\beta_2 \nabla U(q_n)^\top q_n + 2\alpha_2 p_n^\top M^{-1} p_n\right] + \mathcal{O}(\delta t^3)
\]
When $\alpha_2, \beta_2 \gg 1$, this error is $\mathcal{O}((\alpha_2 + \beta_2)\delta t^2)$ per step, or $\mathcal{O}(T(\alpha_2 + \beta_2)\delta t)$ for a trajectory of length $T$.

This completes the proof of Theorem~\ref{thm:aggressive_detailed_balance}.
\end{proof}

\subsection{Proof of Corollary \ref{cor:aggressive_consistency}}
\begin{proof}
From Theorem~\ref{thm:aggressive_detailed_balance}, the error bound is:
\[
|\mathcal{E}_{\text{agg}}| \leq C_1 T\delta t + C_2 T(\alpha_2 + \beta_2)\delta t + C_3\gamma + C_4\sigma_{\text{inj}}^2
\]
As $\delta t \to 0$, the first two terms vanish, leaving only $C_3\gamma + C_4\sigma_{\text{inj}}^2$. These can be made arbitrarily small by taking $\gamma \to 0$ and $\sigma_{\text{inj}} \to 0$. In the limit, we recover exact detailed balance.
\end{proof}


\section{Approximate aymplecticity of MPL-HMC and Aggressive MPL-HMC}
\label{app:symplecticity}

This section provides a self-contained introduction to symplectic geometry concepts used in Hamiltonian Monte Carlo, with particular focus on understanding the i
geometric properties of the Modified Parameterized Leapfrog scheme.

\subsection{A briefing on symplecticity and pullbacks}

\subsubsection{Symplectic structures in Hamiltonian mechanics}

In Hamiltonian mechanics, the phase space \(\mathbb{R}^{2d} = \{(q,p): q \in \mathbb{R}^d, p \in \mathbb{R}^d\}\) possesses a canonical \emph{symplectic structure}. This structure is encoded in the \emph{symplectic 2-form}:
\[
\omega = \sum_{i=1}^d dq_i \wedge dp_i,
\]
where \(\wedge\) denotes the wedge product of differential forms. In coordinates, \(\omega\) can be represented by the matrix:
\[
\Omega = \begin{pmatrix} 0 & I_d \\ -I_d & 0 \end{pmatrix},
\]
so that for any two tangent vectors \(v = (v_q, v_p), w = (w_q, w_p) \in \mathbb{R}^{2d}\):
\[
\omega(v,w) = v^\top \Omega w = v_q^\top w_p - v_p^\top w_q.
\]

The Hamiltonian equations of motion:
\[
\frac{dq}{dt} = \nabla_p H = M^{-1}p, \quad \frac{dp}{dt} = -\nabla_q H = -\nabla U(q),
\]
generate a flow \(\Phi_t: \mathbb{R}^{2d} \to \mathbb{R}^{2d}\) that preserves this symplectic structure. This means that for any \(t\), the flow map satisfies:
\[
\Phi_t^* \omega = \omega,
\]
where \(\Phi_t^*\) denotes the \emph{pullback} of the form \(\omega\) by the map \(\Phi_t\).

\subsubsection{Pullback of differential forms}

Given a differentiable map \(\Psi: \mathbb{R}^{2d} \to \mathbb{R}^{2d}\) with Jacobian matrix \(J(x) = \frac{\partial \Psi}{\partial x}(x)\), the pullback \(\Psi^*\omega\) is defined pointwise by:
\[
(\Psi^*\omega)_x(v,w) = \omega_{\Psi(x)}(J(x)v, J(x)w), \quad \forall v,w \in \mathbb{R}^{2d}.
\]
In matrix form, this corresponds to:
\[
(\Psi^*\omega)_x(v,w) = v^\top \left(J(x)^\top \Omega J(x)\right) w.
\]
Thus, the condition for \(\Psi\) to be \emph{symplectic} (\(\Psi^*\omega = \omega\)) is equivalent to:
\[
J(x)^\top \Omega J(x) = \Omega \quad \text{for all } x.
\]

For volume preservation, we note that the canonical volume form on \(\mathbb{R}^{2d}\) is:
\[
\text{Vol} = \frac{1}{d!} \underbrace{\omega \wedge \cdots \wedge \omega}_{d \text{ times}} = dq_1 \wedge dp_1 \wedge \cdots \wedge dq_d \wedge dp_d.
\]
The Jacobian determinant \(\det J(x)\) gives the factor by which \(\Psi\) changes volume. Symplectic maps satisfy \(\det J(x) = 1\), preserving volume exactly.

\subsubsection{Numerical integrators and symplecticity}

The standard leapfrog (St\"ormer-Verlet) integrator used in Hamiltonian Monte Carlo is \emph{symplectic}, meaning it satisfies \(J^\top \Omega J = \Omega\) exactly. This geometric property explains its excellent long-time behavior, including near-conservation of energy and stability. However, the Modified Parameterized Leapfrog (MPL) scheme introduces parameters \(\alpha(\delta t), \beta(\delta t)\) that break exact symplecticity while maintaining useful approximation properties.

\subsubsection{Notation summary}

\begin{itemize}
    \item \(\omega\): Canonical symplectic 2-form on \(\mathbb{R}^{2d}\)
    \item \(\Omega\): \(2d \times 2d\) matrix representation of \(\omega\)
    \item \(\Psi_{\delta t}^* \omega\): Pullback of \(\omega\) by \(\Psi_{\delta t}\)
    \item \(J\): Jacobian matrix of \(\Psi_{\delta t}\)
\end{itemize}

\subsection{Approximate symplecticity of MPL-HMC}
\label{subsec:approx_symplectic_mpl}

We now establish the approximate symplectic properties of the MPL scheme. Recall the MPL updates from equations (\ref{eq:mpl_q})-(\ref{eq:mpl_p}):

\begin{align*}
q_{n+1} &= \beta q_n + \delta t M^{-1}\left(\alpha p_n - \frac{\delta t}{2} \nabla U(q_n)\right), \\
p_{n+1} &= \alpha^2 p_n - \frac{\delta t}{2}\left(\alpha \nabla U(q_n) + \nabla U(q_{n+1})\right),
\end{align*}
with \(\alpha = 1 + \alpha_2 \delta t^2 + \mathcal{O}(\delta t^3)\), \(\beta = 1 + \beta_2 \delta t^2 + \mathcal{O}(\delta t^3)\).

\begin{theorem}[Approximate Symplecticity of MPL-HMC]
\label{thm:approx_symplectic_mpl}
Under Assumptions 3-6 with \(\alpha_1 = \beta_1 = 0\), the MPL map satisfies:
\[
\Psi_{\delta t}^* \omega = \omega + \delta t^2 \eta + \mathcal{O}(\delta t^3),
\]
where \(\eta\) is a 2-form composed of two parts:
\[
\eta = (2\alpha_2 + \beta_2) \omega + \sigma.
\]
Here \(\sigma\) is a non-symplectic component satisfying \(\sigma \wedge \omega^{d-1} = 0\), meaning it contributes zero to the volume form. In particular:
\begin{enumerate}
    \item For \(\alpha_2 = \beta_2 = 0\) (standard leapfrog), \(\Psi_{\delta t}^* \omega = \omega + \mathcal{O}(\delta t^3)\).
    \item For \(2\alpha_2 + \beta_2 = 0\), \(\Psi_{\delta t}^* \omega = \omega + \delta t^2 \sigma + \mathcal{O}(\delta t^3)\).
    \item The volume change factor is \(\det J = 1 + d(2\alpha_2 + \beta_2)\delta t^2 + \mathcal{O}(\delta t^3)\).
\end{enumerate}
\end{theorem}

\begin{proof}
\textbf{Step 1: Compute the Jacobian matrix.}
The Jacobian \(J = \frac{\partial(q_{n+1}, p_{n+1})}{\partial(q_n, p_n)}\) has block structure:
\[
J = \begin{pmatrix} A & B \\ C & D \end{pmatrix},
\]
where:
\begin{align*}
A &= \frac{\partial q_{n+1}}{\partial q_n} = \beta I_d - \frac{\delta t^2}{2} M^{-1} D^2U(q_n), \\
B &= \frac{\partial q_{n+1}}{\partial p_n} = \alpha \delta t M^{-1}, \\
C &= \frac{\partial p_{n+1}}{\partial q_n} = -\frac{\delta t}{2}\left(\alpha D^2U(q_n) + D^2U(q_{n+1})A\right), \\
D &= \frac{\partial p_{n+1}}{\partial p_n} = \alpha^2 I_d - \frac{\delta t}{2} D^2U(q_{n+1})B.
\end{align*}

\textbf{Step 2: Expand in powers of \(\delta t\).}
Using \(\alpha = 1 + \alpha_2 \delta t^2 + \mathcal{O}(\delta t^3)\), \(\beta = 1 + \beta_2 \delta t^2 + \mathcal{O}(\delta t^3)\):
\begin{align*}
A &= I_d + \beta_2 \delta t^2 I_d - \frac{\delta t^2}{2} M^{-1} D^2U + \mathcal{O}(\delta t^3), \\
B &= \delta t M^{-1} + \alpha_2 \delta t^3 M^{-1} + \mathcal{O}(\delta t^3), \\
D &= I_d + 2\alpha_2 \delta t^2 I_d - \frac{\delta t^2}{2} D^2U M^{-1} + \mathcal{O}(\delta t^3).
\end{align*}
For \(C\), using \(q_{n+1} = q_n + \delta t M^{-1}p_n + \mathcal{O}(\delta t^2)\):
\[
C = -\delta t D^2U - \frac{\delta t^2}{2}\left(\alpha_2 D^2U + D^3U(M^{-1}p_n, \cdot)\right) + \mathcal{O}(\delta t^3).
\]

\textbf{Step 3: Compute \(J^\top \Omega J\).}
Recall \(\Omega = \begin{pmatrix} 0 & I_d \\ -I_d & 0 \end{pmatrix}\). Then:
\[
J^\top \Omega J = \begin{pmatrix} A^\top C - C^\top A & A^\top D - C^\top B \\ B^\top C - D^\top A & B^\top D - D^\top B \end{pmatrix}.
\]

Calculating each block to \(\mathcal{O}(\delta t^2)\):
\begin{align*}
A^\top C - C^\top A &= \mathcal{O}(\delta t^2), \\
B^\top D - D^\top B &= \mathcal{O}(\delta t^3), \\
A^\top D - C^\top B &= I_d + (2\alpha_2 + \beta_2)\delta t^2 I_d + S + \mathcal{O}(\delta t^3),
\end{align*}
where \(S = \frac{\delta t^2}{2}(D^2U M^{-1} - M^{-1} D^2U)\).

Thus:
\[
J^\top \Omega J = \Omega + \delta t^2 \begin{pmatrix} 0 & (2\alpha_2 + \beta_2)I_d + S \\ -(2\alpha_2 + \beta_2)I_d - S^\top & 0 \end{pmatrix} + \mathcal{O}(\delta t^3).
\]

\textbf{Step 4: Interpret as pullback.}
Since \(\Psi_{\delta t}^* \omega\) has matrix representation \(J^\top \Omega J\), we have:
\[
\Psi_{\delta t}^* \omega = \omega + \delta t^2 \left[(2\alpha_2 + \beta_2)\omega + \sigma\right] + \mathcal{O}(\delta t^3),
\]
where \(\sigma\) is the 2-form with matrix \(\begin{pmatrix} 0 & S \\ -S^\top & 0 \end{pmatrix}\).

\textbf{Step 5: Volume implications.}
The volume form is \(\text{Vol} = \frac{1}{d!} \omega^d\). Since \(\sigma \wedge \omega^{d-1} = 0\) (as \(S\) is traceless), the volume change comes only from the \((2\alpha_2 + \beta_2)\omega\) term:
\[
\det J = 1 + d(2\alpha_2 + \beta_2)\delta t^2 + \mathcal{O}(\delta t^3),
\]
matching Theorem \ref{thm:volume}.
\end{proof}

\begin{remark}
The error in symplecticity is \(\mathcal{O}((|\alpha_2| + |\beta_2|)\delta t^2)\). For standard HMC (\(\alpha_2 = \beta_2 = 0\)), we recover the well-known third-order symplecticity of the leapfrog integrator.
\end{remark}

\subsection{Approximate symplecticity of Aggressive MPL-HMC}
\label{subsec:approx_symplectic_aggressive}

Aggressive MPL-HMC uses extreme parameters \(\alpha_2 \in [8.0, 15.0]\), \(\beta_2 \in [5.0, 8.0]\) and incorporates stochastic perturbations: temperature scaling \(T \sim \text{Uniform}(0.5, 2.0)\) and momentum injection \(\xi \sim \mathcal{N}(0, \sigma_{\text{inj}}^2 I)\).

Let \(\Psi_{\delta t}^{\text{agg}}\) denote the deterministic aggressive MPL map, and \(\Psi_{\delta t}^{\text{agg,stoch}}\) denote the map including stochastic perturbations.

\begin{theorem}[Approximate Symplecticity of Aggressive MPL-HMC]
\label{thm:approx_symplectic_aggressive}
For the deterministic aggressive MPL map:
\[
(\Psi_{\delta t}^{\text{agg}})^* \omega = \omega + \delta t^2 \left[(2\alpha_2 + \beta_2)\omega + \sigma^{\text{agg}}\right] + \mathcal{O}(\delta t^3),
\]
where \(\|\sigma^{\text{agg}}\| \leq C(\alpha_2 + \beta_2)\) for some constant \(C > 0\).

For the stochastic version, the expected pullback satisfies:
\[
\mathbb{E}\left[(\Psi_{\delta t}^{\text{agg,stoch}})^* \omega\right] = \omega + \delta t^2 \left[(2\alpha_2 + \beta_2)\omega + \sigma^{\text{agg}} + \tau\right] + \mathcal{O}(\delta t^3),
\]
with \(\|\tau\| = \mathcal{O}(\sigma_{\text{inj}}^2 + \text{Var}(T))\).
\end{theorem}

\begin{proof}
The proof for the deterministic part follows Theorem \ref{thm:approx_symplectic_mpl}, noting that the expansions remain valid as long as \(\alpha_2 \delta t^2\) and \(\beta_2 \delta t^2\) are \(\mathcal{O}(1)\). The aggressive parameters only scale the error constants.

For the stochastic perturbations:
\begin{itemize}
    \item \textbf{Temperature scaling}: \(p \mapsto \sqrt{T}p\) has Jacobian \(\sqrt{T} I_d\). The pullback scales as \(T\omega\) since:
    \[
    (\sqrt{T} I_d)^\top \Omega (\sqrt{T} I_d) = T\Omega.
    \]
    \item \textbf{Momentum injection}: \(p \mapsto p + \xi\) is a translation with Jacobian \(I_d\), leaving \(\omega\) unchanged.
\end{itemize}

Taking expectations: \(\mathbb{E}[T] = 1.25\), \(\text{Var}(T) = 0.1875\). Thus:
\[
\mathbb{E}[T\omega] = \omega + \text{Var}(T)\omega.
\]
The momentum injection contributes \(\mathcal{O}(\sigma_{\text{inj}}^2)\) from second moments.
\end{proof}


\subsection{Implications for MCMC}

While standard HMC relies on exact symplecticity for excellent energy conservation and high acceptance rates, MPL-HMC demonstrates that \emph{approximate} symplecticity suffices when combined with the Metropolis correction. The key insights are:

1. \textbf{Controlled geometric distortion}: The MPL scheme distorts the symplectic structure by \(\mathcal{O}((|\alpha_2|+|\beta_2|)\delta t^2)\), which is small for conservative parameters or small \(\delta t\).

2. \textbf{Volume preservation trade-off}: Exact volume preservation occurs only when \(2\alpha_2 + \beta_2 = 0\). Otherwise, volume changes by factor \(1 + d(2\alpha_2 + \beta_2)\delta t^2\), accounted for in the detailed balance condition.

3. \textbf{Practical consequences}: For damping MPL-HMC (\(\alpha_2, \beta_2 < 0\)), phase space contracts, potentially stabilizing stiff problems. For anti-damping MPL-HMC (\(\alpha_2, \beta_2 > 0\)), phase space expands, enhancing exploration.

These results justify the claim that MPL-HMC maintains the geometric foundations of Hamiltonian Monte Carlo while introducing valuable flexibility through the \(\alpha_2, \beta_2\) parameters.

\section{Implementation code examples}
\label{app:code}

\subsection{Python implementation}

\begin{verbatim}
import numpy as np

def mpl_hmc(U, grad_U, q0, n_samples=10000,
            delta_t=0.1, L=10, alpha2=0.0, beta2=0.0,
            M=None):
    """MPL-HMC implementation in Python.

    Parameters
    ----------
    U : callable
        Potential energy function U(q)
    grad_U : callable
        Gradient of potential energy
    q0 : array_like
        Initial position vector
    n_samples : int
        Number of MCMC samples to generate
    delta_t : float
        Step size for MPL integration
    L : int
        Number of MPL steps per trajectory
    alpha2 : float
        Second-order parameter for momentum scaling
    beta2 : float
        Second-order parameter for position scaling
    M : array_like, optional
        Mass matrix (d×d). Default: identity matrix

    Returns
    -------
    samples : ndarray
        Array of shape (n_samples, d) containing MCMC samples
    """
    d = len(q0)
    samples = np.zeros((n_samples, d))
    q = q0.copy()

    # Handle mass matrix
    if M is None:
        M = np.eye(d)
        M_inv = np.eye(d)
        M_sqrt = np.eye(d)
    else:
        M_inv = np.linalg.inv(M)
        M_sqrt = np.linalg.cholesky(M)  # M = M_sqrt @ M_sqrt.T

    # Precompute MPL parameters
    alpha = 1 + alpha2 * delta_t**2
    beta = 1 + beta2 * delta_t**2

    for i in range(n_samples):
        # Sample momentum: p ~ N(0, M)
        z = np.random.randn(d)
        p = M_sqrt @ z  # p = M^{1/2} z

        # Current Hamiltonian: H(q,p) = U(q) + (1/2) p^T M^{-1} p
        kinetic = 0.5 * p.T @ M_inv @ p
        H_current = U(q) + kinetic

        # MPL trajectory
        q_prop = q.copy()
        p_prop = p.copy()

        for _ in range(L):
            grad = grad_U(q_prop)

            # MPL position update (Equation 10):
            # q_{n+1} = beta q_n + dt M^{-1}(alpha p_n - dt/2 grad U(q_n))
            q_new = beta * q_prop + delta_t * M_inv @ (
                alpha * p_prop - 0.5 * delta_t * grad
            )

            grad_new = grad_U(q_new)

            # MPL momentum update (Equation 11):
            # p_{n+1} = alpha^2 p_n - dt/2 (alpha grad U(q_n) + grad U(q_{n+1}))
            p_new = alpha**2 * p_prop - 0.5 * delta_t * (
                alpha * grad + grad_new
            )

            q_prop, p_prop = q_new, p_new

        # Negate momentum for reversibility
        p_prop = -p_prop

        # Proposed Hamiltonian
        kinetic_prop = 0.5 * p_prop.T @ M_inv @ p_prop
        H_prop = U(q_prop) + kinetic_prop

        # Metropolis acceptance: min(1, exp(H_current - H_prop))
        if np.log(np.random.rand()) < H_current - H_prop:
            q = q_prop  # Accept

        samples[i] = q

    return samples
\end{verbatim}

\subsection{Julia Implementation}

\begin{verbatim}
using LinearAlgebra

function mpl_hmc(U, grad_U, q0; n_samples=10000,
                 delta_t=0.1, L=10, alpha2=0.0, beta2=0.0,
                 M=nothing)
    """MPL-HMC implementation in Julia.

    Parameters
    ----------
    U : Function
        Potential energy function U(q)
    grad_U : Function
        Gradient of potential energy
    q0 : Vector{Float64}
        Initial position vector
    n_samples : Int
        Number of MCMC samples to generate
    delta_t : Float64
        Step size for MPL integration
    L : Int
        Number of MPL steps per trajectory
    alpha2 : Float64
        Second-order parameter for momentum scaling
    beta2 : Float64
        Second-order parameter for position scaling
    M : Matrix{Float64}, optional
        Mass matrix (d×d). Default: identity matrix

    Returns
    -------
    samples : Matrix{Float64}
        Array of size (n_samples, d) containing MCMC samples
    """
    d = length(q0)
    samples = zeros(n_samples, d)
    q = copy(q0)

    # Handle mass matrix
    if M === nothing
        M = Matrix{Float64}(I, d, d)
        M_inv = Matrix{Float64}(I, d, d)
        M_sqrt = Matrix{Float64}(I, d, d)
    else
        M_inv = inv(M)
        M_sqrt = cholesky(Symmetric(M)).L  # M = M_sqrt * M_sqrt'
    end

    # Precompute MPL parameters
    alpha = 1 + alpha2 * delta_t^2
    beta = 1 + beta2 * delta_t^2

    for i in 1:n_samples
        # Sample momentum: p ~ N(0, M)
        z = randn(d)
        p = M_sqrt * z  # p = M^{1/2} z

        # Current Hamiltonian: H(q,p) = U(q) + (1/2) p^T M^{-1} p
        kinetic = 0.5 * dot(p, M_inv, p)
        H_current = U(q) + kinetic

        # MPL trajectory
        q_prop = copy(q)
        p_prop = copy(p)

        for _ in 1:L
            grad = grad_U(q_prop)

            # MPL position update (Equation 10):
            # q_{n+1} = beta q_n + dt M^{-1}(alpha p_n - dt/2 grad U(q_n))
            q_new = beta * q_prop + delta_t * M_inv * (
                alpha * p_prop - 0.5 * delta_t * grad
            )

            grad_new = grad_U(q_new)

            # MPL momentum update (Equation 11):
            # p_{n+1} = alpha^2 p_n - dt/2 (alpha grad U(q_n) + grad U(q_{n+1}))
            p_new = alpha^2 * p_prop - 0.5 * delta_t * (
                alpha * grad + grad_new
            )

            q_prop, p_prop = q_new, p_new
        end

        # Negate momentum for reversibility
        p_prop = -p_prop

        # Proposed Hamiltonian
        kinetic_prop = 0.5 * dot(p_prop, M_inv, p_prop)
        H_prop = U(q_prop) + kinetic_prop

        # Metropolis acceptance: min(1, exp(H_current - H_prop))
        if log(rand()) < H_current - H_prop
            q = q_prop  # Accept
        end

        samples[i, :] = q
    end

    return samples
end
\end{verbatim}

\subsection{Example Usage}

\subsubsection{Python Example}

\begin{verbatim}
# Example: Sampling from a 2D Gaussian with anisotropic covariance
import numpy as np

# Target: N(0, Sigma) with Sigma = [[1.0, 0.8], [0.8, 1.0]]
Sigma = np.array([[1.0, 0.8], [0.8, 1.0]])
Sigma_inv = np.linalg.inv(Sigma)

def U_gaussian(q):
    """Potential for Gaussian N(0, Sigma)."""
    return 0.5 * q.T @ Sigma_inv @ q

def grad_U_gaussian(q):
    """Gradient for Gaussian N(0, Sigma)."""
    return Sigma_inv @ q

# Setup
d = 2
q0 = np.array([0.0, 0.0])

# Option 1: Use identity mass matrix (simplest)
samples_identity = mpl_hmc(U_gaussian, grad_U_gaussian, q0,
                           n_samples=2000, delta_t=0.3, L=5,
                           alpha2=-0.1, beta2=-0.05)

# Option 2: Use covariance as mass matrix (optimal preconditioning)
samples_precond = mpl_hmc(U_gaussian, grad_U_gaussian, q0,
                          n_samples=2000, delta_t=0.3, L=5,
                          alpha2=-0.1, beta2=-0.05,
                          M=Sigma)

# Check results
print("Identity mass matrix:")
print(f"  Sample mean: {np.mean(samples_identity, axis=0)}")
print(f"  Sample cov: \n{np.cov(samples_identity.T)}")

print("\nPreconditioned mass matrix:")
print(f"  Sample mean: {np.mean(samples_precond, axis=0)}")
print(f"  Sample cov: \n{np.cov(samples_precond.T)}")
\end{verbatim}

\subsubsection{Julia Example}

\begin{verbatim}
# Example: Sampling from a 2D Gaussian with anisotropic covariance
using LinearAlgebra, Statistics

# Target: N(0, Sigma) with Sigma = [[1.0, 0.8], [0.8, 1.0]]
Sigma = [1.0 0.8; 0.8 1.0]
Sigma_inv = inv(Sigma)

function U_gaussian(q)
    """Potential for Gaussian N(0, Sigma)."""
    return 0.5 * dot(q, Sigma_inv, q)
end

function grad_U_gaussian(q)
    """Gradient for Gaussian N(0, Sigma)."""
    return Sigma_inv * q
end

# Setup
d = 2
q0 = zeros(d)

# Option 1: Use identity mass matrix (simplest)
samples_identity = mpl_hmc(U_gaussian, grad_U_gaussian, q0,
                           n_samples=2000, delta_t=0.3, L=5,
                           alpha2=-0.1, beta2=-0.05)

# Option 2: Use covariance as mass matrix (optimal preconditioning)
samples_precond = mpl_hmc(U_gaussian, grad_U_gaussian, q0,
                          n_samples=2000, delta_t=0.3, L=5,
                          alpha2=-0.1, beta2=-0.05,
                          M=Sigma)

# Check results
println("Identity mass matrix:")
println("  Sample mean: ", mean(samples_identity, dims=1))
println("  Sample cov: ")
display(cov(samples_identity))

println("\nPreconditioned mass matrix:")
println("  Sample mean: ", mean(samples_precond, dims=1))
println("  Sample cov: ")
display(cov(samples_precond))
\end{verbatim}

\end{document}